\documentclass[11pt]{article}
\usepackage[utf8]{inputenc}
\usepackage{amsthm,amsmath,amsfonts,amssymb,dsfont}
\usepackage[dvipsnames]{xcolor}
\usepackage[margin=1in, top=0.8in, bottom=1in]{geometry}

\usepackage{natbib}
\usepackage{enumerate}
\usepackage{booktabs}
\usepackage{multirow}
\usepackage{tikz}
\usepackage{appendix}
\usepackage{authblk}
\usetikzlibrary{arrows.meta}

\usepackage{algorithm, algorithmicx}
\usepackage[noend]{algpseudocode}
\usepackage[]{setspace}
\usepackage[]{caption}
\RequirePackage[colorlinks,citecolor=blue,urlcolor=blue, linkcolor=blue, hypertexnames=false]{hyperref}
\let\Algorithm\algorithm
\renewcommand\algorithm[1][]{\Algorithm[#1]\setstretch{1.0}}

\theoremstyle{definition}
\newtheorem{theorem}{Theorem}
\newtheorem{proposition}{Proposition}
\newtheorem{assumption}{Assumption}
\newtheorem{lemma}{Lemma}
\newtheorem{definition}{Definition}

\newtheorem{example}{Example}
\newtheorem{remark}{Remark}

\newcommand{\given}{\,|\,}
\newcommand{\Given}{\,\Bigm|\,}

\newcommand{\cke}{\mathrm{CKE}}
\newcommand{\cfi}{\mathrm{CFI}}

\newcommand{\cdi}{\mathrm{CDI}}

\newcommand{\cai}{\mathrm{CAI}}

\newcommand{\cae}{\mathrm{CAE}}

\newcommand{\clr}{\mathrm{CLR}}
\newcommand{\cen}{\mathrm{cen}}
\newcommand{\mult}{\mathrm{mult}}
\newcommand{\unit}{\mathrm{unit}}
\newcommand{\gini}{\mathrm{Gini}}

\newcommand{\E}{\mathbb{E}}
\newcommand{\C}{\mathsf{C}}
\newcommand{\A}{\mathsf{A}}

\newcommand{\cov}{\mathrm{cov}}
\newcommand{\var}{\mathrm{var}}

\newcommand{\ind}{\mathds{1}}

\newcommand{\RN}[1]{%
  (\textup{\uppercase\expandafter{\romannumeral#1}})%
}

\definecolor{LKred}{HTML}{d62728}
\definecolor{LKgreen}{HTML}{2ca02c}

\title{\Large Perturbation-based Effect Measures for Compositional Data}
\author{Anton Rask Lundborg and Niklas Pfister}
\date{\normalsize \today}
\affil{
  \normalsize{\textit{Department of Mathematical Sciences, University of Copenhagen, Denmark}} 
}

\begin{document}

\maketitle

\begin{abstract}
  \noindent
  Existing effect measures for compositional features are inadequate for many
  modern applications, for example, in microbiome research, since
  they display traits such as high-dimensionality and sparsity that can be
  poorly modelled with traditional parametric approaches.
  Further, assessing -- in an unbiased way -- how summary
  statistics of a composition (e.g., racial diversity) affect a response
  variable is not straightforward. We propose a framework based on
  hypothetical data perturbations which defines interpretable statistical
  functionals on the compositions themselves, which we call average perturbation
  effects. These effects naturally account for confounding that biases
  frequently used marginal dependence analyses. We show how average perturbation
  effects can be estimated efficiently by deriving a perturbation-dependent
  reparametrization and applying semiparametric estimation techniques. We
  analyze the proposed estimators empirically on simulated and semi-synthetic
  data and demonstrate advantages over existing techniques on data from New York
  schools and microbiome data.

\end{abstract}

\section{Introduction}
\label{sec:intro}
Compositional data, that is, measurements consisting of parts of a whole, appear
in many scientific disciplines, for example, rock compositions in geology, green
house gases in the atmosphere, species compositions in ecology and racial or
gender distributions in social sciences. The defining property of such data is
that they only capture information relative to the whole. Depending on the
application, the absolute information either might not be of interest (e.g.,
when analyzing racial or gender distributions) or might not be easily accessible
with existing measuring techniques (e.g., in ecology). A $d$-coordinate
composition can -- without loss of information -- be normalized to a point on
the $(d-1)$-dimensional simplex
\begin{equation*}
  \Delta^{d-1}:=\left\{z=(z^1, \dots, z^d) \in [0,1]^d\mid
    \textstyle\sum_{j=1}^dz^j=1\right\}.
\end{equation*}

Modelling the relationship between a real-valued response $Y \in \mathbb{R}$ and
a composition $Z = (Z^1, \dots, Z^d) \in \Delta^{d-1}$ has a long history and
has, for example, been considered in the literature on `experimental design with
mixtures' \citep{claringbold1955use, cornell2002experiments} and 'log-ratio
based compositional data analysis' \citep{aitchison1982statistical}. We
summarize the historical development of these fields in
Section~\ref{sec:related_literature} (Sections beginning with `S' refer to the
supplement throughout the paper). A key aspect of this research is that it
attempts to achieve two separate goals simultaneously: (i) ensuring that the
model parameters are interpretable and (ii) providing parametric models that fit
to realistic data-generating processes with compositional covariates. More
recently, in the context of high-dimensional and sparse data, alternative
approaches \citep{knight2018best,cammarota2020gut,huang2022supervised} to
solving the second goal have been proposed that apply flexible machine learning
methods to model the relationship between $Y$ and $Z$ nonparametrically.
However, such approaches are at odds with the first goal of providing
interpretable effect measures and cannot be combined with the above-mentioned
approaches which heavily rely on parametric models to define such effects.

Our main contribution is to provide a unified framework for defining
interpretable effect measures for compositional features which we call
\emph{average perturbation effects} that can be combined with flexible
nonparametric estimation procedures.  Special cases of these average
perturbation effects coincide with model-based individual feature effect
measures, e.g.\ Cox-directions \citep{cox1971note} and log-contrast coefficients
\citep{aitchison1984log}, whenever the corresponding model is correctly
specified (see Section~\ref{sec:misc_results}); unlike those measures, however,
they are defined without reference to a regression model and so do not require
one to be correct. Importantly, this allows us to construct more flexible target
quantities, two of which we believe are particularly relevant for applications.
First, effects based on \emph{binary perturbations} are useful when summarizing
the effect of a discrete change, e.g.\ the effect of setting $Z^j$ to $0$ while
rescaling $Z^{-j} := (Z^1, \dots, Z^{j-1}, Z^{j+1}, \dots, Z^d)$ to the simplex.
In Section~\ref{sec:microbiome} we show empirically that the compositional
features with significant effects as measured by these binary perturbations
coincide with those deemed important by the standard log-contrast regression
procedures used in the microbiome literature. Second, effects based on
\emph{directional perturbations} summarize continuous changes in specific
directions. This can involve changing individual components of $Z$, but also
allows for more complex changes that are not tied to a specific coordinate of
$Z$. For example, we can define the effect of `diversifying' $Z$, i.e.\ moving
closer to the center of $\Delta^{d-1}$. Such diversifying effects have been
considered previously \citep{antonio2004effects} but are usually defined as the
effect of marginally changing a specific measure of diversity. However, as the
effect of diversifying can be confounded by other features of $Z$ such a
marginal analysis can be misleading and due to the simplex constraint, there is
no immediate way to correct for this confounding when estimating these effects.
In contrast, our perturbation-based framework can control for the remaining
variation in $Z$ leading to unconfounded effect estimates with valid asymptotic
inference.

The framework also stands in a specific relation to several existing lines
of work. The log-ratio approach to compositional data analysis proceeds by
transforming the simplex to an unconstrained space before analysis, which
requires special treatment of zeros. We instead define effects on the simplex
itself, so that zeros pose no problem. Separately, the difficulty of defining
interventions on variables that are deterministically related
\citep{spirtes2004causal, beckers2023causal} has the compositional case as an
instance, for which perturbations provide one concrete answer. Our feature-level
effects also serve the same purpose as model-agnostic variable importance
measures such as permutation importance \citep{breiman2001random} and the
nonparametric $R^2$ \citep{williamson2021nonparametric}, but are defined through
an explicit change to the composition rather than through deleting or permuting
a coordinate.

Our proposal is intimately related to the growing literature on targeted
learning and semiparametric inference
\citep{chernozhukov2018double,berk2021assumption,vansteelandt2022assumption}. An
important aspect of our work is, however, to clarify the need for tailored
target estimands in the presence of compositional structure in the covariates.
The estimation methods themselves are well developed, but the difficulty in
the compositional setting lies in deciding what should be estimated. Applying a
nonparametric effect measure that is not aware of the sum-to-one constraint of
the compositional vector can be just as misleading as naively estimating a
non-compositional effect based on a parametric model (see
Section~\ref{sec:microbiome}). Beyond defining average perturbation effects, we
also provide several algorithm-agnostic procedures to estimate the average
perturbation effects. These procedures are essentially well-known semiparametric
learning algorithms applied after appropriately transforming the compositional
covariates. While these algorithms require us to estimate nuisance functions
relating to the distribution of $(Y, Z)$, which are particularly delicate for
compositional covariates, they only require these estimates to be accurate and
not interpretable. This allows us to apply arbitrary regression methods as long
as they are sufficiently predictive.

Average perturbation effects are defined by comparing the change in the expected
response of $Y$ under a hypothetical perturbation of the composition $Z$. To
motivate this, consider unconstrained covariates $X \in \mathbb{R}^d$ instead of
the composition $Z \in \Delta^{d-1}$. Suppose that there exists $\beta =
(\beta^1, \dots, \beta^d) \in \mathbb{R}^d$ such that for all $x = (x^1, \dots,
x^d) \in \mathbb{R}^d$, $f(x) := \E[Y \given X=x] = \sum_{j=1}^d \beta^j x^j$,
that is, the conditional expectation of $Y$ given $X=x$ is a linear function
of $x$. Then, one interpretation of the coefficient $\beta^j$ is that it
corresponds to the expected rate of change of the response in the $j$th
coordinate, or more formally,
$
  \partial_{x^j} f(x^1, \dots, x^d) = \beta^j.
$
We can rephrase this interpretation by considering a perturbation
$\psi:\mathbb{R}^d\times [0,1)\rightarrow\mathbb{R}^d$ that increases
the $j$th coordinate, that is, $\psi(x, \gamma)=(x^1, \dots, x^j + \gamma,
\dots, x^d)$. Then, for any $x \in \mathbb{R}^d$, we have
\[
  \partial_{x^j} f(x) = \lim_{\gamma \to 0}\frac{f(\psi(x, \gamma)) - f(\psi(x, 0))}{\gamma} = \partial_{\gamma} f(\psi(x, \gamma)) \bigm|_{\gamma = 0}.
\]
Therefore, the coefficient $\beta^j$ also corresponds to the initial rate of
change (initial referring to $\gamma=0$) of $f$ along the perturbation $\psi$.
Similar observations are possible regarding other interpretable effect measures
for unconstrained covariates and under nonlinear models.
For compositional $Z$ we need perturbations tailored to the simplex structure
instead of the $\psi$ given above. In Section~\ref{sec:toy_examples} two
examples illustrate how conventional effect measures applied to compositional
data can be misleading and how our approach mitigates this.

A perturbation, formally defined in Section~\ref{sec:perturbations}, assigns to
each $z\in\Delta^{d-1}$ a trajectory $\gamma\mapsto\psi(z,\gamma)$ in the
simplex along which $z$ changes. Because $Z$ is compositional, such a map must
specify what happens to all coordinates, not just the one of interest as for
unconstrained covariates. A binary perturbation encodes a discrete change, for
example setting $Z^j$ to $0$ and rescaling the rest to the simplex; averaging
the resulting contrast in the regression function $f: z\mapsto\E[Y\given Z=z]$
over $Z$, normalized by the probability that the perturbation changes $Z$, gives
the average binary perturbation effect $\lambda_{\psi}$. A directional
perturbation encodes a continuous change, for example moving $z$ towards the
center $z_{\cen}=(1/d,\ldots,1/d)$ of the simplex, one way of formalizing an
increase in the diversity of $Z$. Averaging the initial rate of change of $f$
along the perturbation gives the average directional perturbation effect
$\tau_{\psi}$.

A key strength of perturbations is that they make target parameters easy to
communicate, interpret and discuss with practitioners. In our experience
practitioners understand well which perturbation they are interested in, and
formalizing it, as we illustrate in Section~\ref{sec:simplex_perturbations}, is
straightforward. Although $\lambda_{\psi}$ and $\tau_{\psi}$ have a clear causal
motivation, they are purely (observational) statistical quantities. Whether such
a quantity coincides with the causal effect one has in mind is settled by
assumptions on the data-generating mechanism, not by the choice of perturbation.
Our concern here is the statistical problem of selecting an estimand that is
well defined on the simplex and can be estimated efficiently. Causal
assumptions, target parameters and adjustments to the estimators are discussed
in Section~\ref{sec:causal_models}.

While the target parameters $\lambda_{\psi}$ and $\tau_{\psi}$ may seem simple
at first sight, they can be challenging to estimate directly as they depend on
estimating the regression function $z\mapsto \E[Y\given Z=z]$ which, in part due
to the simplex structure, is generally a complex functional. We therefore
propose an approach that reparametrizes the compositional vector
$Z\in\Delta^{d-1}$ via a bijection $\phi$ to a vector $(L,
W)\in\mathbb{R}\times\mathcal{W}$, for some set $\mathcal{W}$, in such a way
that the perturbation only changes $L$ in the new parametrization. Constructing
such a reparametrization reduces the problem to a setting in which any existing
method to estimate the effect of $L$ on $Y$ that adjusts for $W$ results in an
unbiased estimate of the desired effect. Given a perturbation $\psi$, we provide
a general recipe to construct the reparametrization $\phi$. We furthermore
propose several explicit perturbations (including their corresponding
reparametrizations) and suggest efficient estimators of the perturbation-based
target parameters using semiparametric theory based on the $(L,
W)$-reparametrization.

The paper is structured as follows. Our perturbation-based framework is formally
introduced and developed in Section~\ref{sec:perturbations}. We describe how to
construct perturbations and provide a table of pre-defined simplex effects in
Section~\ref{sec:simplex_perturbations}. In Section~\ref{sec:estimation} we
describe the semiparametric techniques used to estimate our proposed effects.
Section~\ref{sec:numerical_experiments} gives a semi-synthetic and two
real-world applications. The supplementary material contains
technical details on theory and algorithms along with additional examples and
experiments.

\subsection{Notation}\label{sec:notation} 

We write random variables in upper case and the values they take in lower case,
so that $Z \in \Delta^{d-1}$ is the random composition and $z \in \Delta^{d-1}$
a generic point of the simplex. We use superscripts to index coordinates of a
vector throughout, that is, for $x \in \mathbb{R}^d$ we write $x^j$ for the
$j$th coordinate of $x$ and denote by $e_j$ the $j$th canonical unit vector in
$\mathbb{R}^d$. For all $n \in \mathbb{N}$, we define $[n] := \{1, \dots, n\}$.
For a real-valued function $f:\mathcal{X}\rightarrow\mathbb{R}$ with
$\mathcal{X}\subseteq\mathbb{R}^d$, we denote by $\partial_{x^j}$ the partial
derivative of $f$ with respect to the $x^j$ coordinate and the gradient (vector
of partial derivatives) of $f$ is denoted by $\nabla f$. When differentiating a
function of a single variable at a boundary point of the domain, we only
consider the one-sided limit. For $p \geq 1$ and $x \in \mathbb{R}^d$, we define
$\|x\|_p := \left( \sum_{j=1}^d \left|x^j\right|^p \right)^{1/p}$. Lastly, we
define $\C(z):=z/\sum_{j=1}^d z^j$, the \emph{(simplex) closure operator}.

\section{Target parameters via perturbations}\label{sec:perturbations}

For a response $Y \in \mathbb{R}$ and a compositional predictor
$Z \in \mathcal{Z} \subseteq \Delta^{d-1}$, we seek to summarize the
effect of $Z$ on the expectation of $Y$ under a pre-specified change
of $Z$. We specify this change via perturbations.
\begin{definition}[Perturbations]
  \label{def:perturbation} 
  Let $\mathcal{D} \subseteq \mathcal{Z} \times [0, \infty)$ satisfy for all $z
  \in \mathcal{Z}$ that $\{z\} \times \{0\} \in \mathcal{D}$. We call $\psi =
  (\psi^1, \dots, \psi^d) : \mathcal{D} \to \Delta^{d-1}$ a \emph{perturbation}
  if for all $z\in\mathcal{Z}$, it holds that $\psi(z, 0) = z$. We define for
  all $z \in \mathcal{Z}$ the \emph{$\gamma$-domain} and \emph{endpoint} of
  $\psi$ (if it exists) by
  \[
    \mathcal{I}_{\psi}(z) := \{\gamma \in[0,\infty) \mid (z, \gamma) \in \mathcal{D}\} \quad\text{and}\quad \mathsf{E}_\psi(z) := \lim_{\gamma \to \sup \mathcal{I}_\psi(z)} \psi(z, \gamma),
  \]
  respectively. We furthermore say that
  \begin{enumerate}[(a)]
  \item $\psi$ is a \emph{binary perturbation} if for all
    $z \in \mathcal{Z}$ it holds that
    $\mathcal{I}_{\psi}(z) = \{0, 1\}$, i.e.\
      $\mathcal{D}=\mathcal{Z}\times\{0, 1\}$, and
  \item $\psi$ is a \emph{directional perturbation} if for all $z \in
  \mathcal{Z}$ there exists $\epsilon > 0$ with $[0, \epsilon) \subseteq
  \mathcal{I}_{\psi}(z)$ and $\omega_\psi(z):= \partial_{\gamma} \psi(z, \gamma)
  \big|_{\gamma = 0} $ exists. We further define the \emph{direction} and
  \emph{speed} of the perturbation for all $z \in \mathcal{Z}$  by
  \[
    v_\psi(z) := \ind_{\{\omega_\psi(z) \neq 0\}}\frac{\omega_\psi(z)}{\|\omega_\psi(z)\|_1} \quad\text{and}\quad s_\psi(z) := \|\omega_\psi(z)\|_1,
  \]
  respectively. 
  \end{enumerate}
\end{definition}
Binary perturbations describe discrete changes that move each point to a
specific endpoint, capturing the difference between $\gamma = 0$ (doing nothing)
and $\gamma = 1$. An example is setting the first component of $Z$ to $0$.
Directional perturbations describe `local' changes, that is, differences between
$\gamma = 0$ and $\gamma=\epsilon$ for small $\epsilon > 0$ (moving slightly
along the perturbation path). We later focus on the derivative at $\gamma = 0$,
which captures the rate of change of a directional perturbation at the initial
point, and as we only use this derivative it suffices to consider perturbations
of the form
\begin{equation}
  \label{eq:lin_pert}
  \psi(z, \gamma) = z + \gamma s_\psi(z)v_\psi(z).
\end{equation}
We construct interpretable perturbations in
Section~\ref{sec:deriv-iso}. Parametrizing direction and speed by the $1$-norm
is an arbitrary choice that is convenient on the simplex. Perturbations can be
defined more generally, but we stick to the simplex for simplicity

Both types of perturbation induce summary statistics capturing their effect on
the response. We define the following target parameters as contrasts of the
perturbation similar to how causal effects are defined via contrasts of
interventions in the causal inference literature.

\begin{definition}[Average perturbation effects]
  \label{def:avg_per_eff}
  Let $\mathcal{Z} \subseteq \Delta^{d-1}$, $(Y, Z)\in\mathbb{R}\times
  \mathcal{Z}$ be random variables and $f: z \mapsto \E[Y\given Z=z]$. For a
  binary perturbation $\psi$ with $\mathbb{P}(Z\neq\psi(Z, 1))>0$ such that $z
  \mapsto \psi(z, 1)$ preserves null sets of $Z$, the \emph{average (binary)
  perturbation effect} is
  \begin{equation*}
    \lambda_{\psi} := \frac{\E\left[f(\psi(Z, 1))\right]
      - \E\left[Y \right]}{\mathbb{P}(Z \neq \psi(Z, 1))}.
  \end{equation*}
  For a directional perturbation $\psi$, if for all $z \in \mathcal{Z}$ it holds
  that $f$ is differentiable on the simplex at $z$ (see Section~\ref{sec:diff}),
  we define the \emph{average (directional) perturbation effect} by
  \begin{equation*}
    \tau_\psi := \E\left[ \partial_{\gamma} f(
      \psi(Z, \gamma))\big|_{\gamma = 0} \right].
  \end{equation*}
\end{definition}
The conditions we impose on the perturbation and distribution of $Z$ are only to
ensure that the parameters are well-defined (see
Section~\ref{sec:well-defined_perturbations}). We showcase different potential
applications for both types of perturbation effect in the examples in
Section~\ref{sec:toy_examples} and the numerical experiments in
Section~\ref{sec:numerical_experiments}. For binary perturbations, the
denominator is included to scale the effect by the probability that the
perturbation does not change $Z$. This is not required but in many cases
provides a more meaningful parameter if one is interested in how the
perturbation affects points on the simplex that are not yet at their endpoint.
For instance, the denominator makes $\lambda_\psi$ invariant to modifications of
the distribution of $Z$ that only changes $\mathbb{P}(Z \neq \psi(Z, 1))$, i.e.\
the probability that $Z$ is not already perturbed.

Direct estimation of both target parameters relies on estimating the
regression function $f: z\mapsto \E[Y\mid Z=z]$ and taking averages (see
Section~\ref{sec:estimation}). An alternative procedure for $\lambda_\psi$
follows from noting that, if $L= \ind_{\{Z = \psi(Z, 1)\}}$ and $W = \psi(Z, 1)$, then
\[
  \E[Y \given L = 1, W] = \frac{\E[Y L \given W]}{\E[L \given W]} =
  \frac{\E[f(Z) L \given W]}{\E[L \given W]} = f(W).
\]
Therefore, we can write
\begin{equation}
  \label{eq:simple_dape}
  \lambda_\psi = \frac{\E[\E[Y \given L=1, W] - Y]}{\mathbb{P}(L=0)},
\end{equation}
an expression more amenable to semiparametric analysis, which allows us to
construct efficient estimators. If we assume that $f(\psi(Z, 1)) - f(Z)$ is
constant when $L=0$ and equal to $\theta$, then
\[
  \lambda_\psi = \frac{\E[\{f(W) - f(Z)\}(1-L)]}{\mathbb{P}(L=0)} = \theta
\]
and furthermore
\[
  \E[Y\given Z] = f(Z) = f(W)L + (f(Z) - f(W) + f(W))(1-L) = \theta L + h(W),
\]
where $h: w \mapsto f(w) - \theta$ is a function of $w$ alone because $\theta$
is constant, so $\theta$ is the coefficient of $L$ in a partially linear model.

It would be convenient to also phrase the average directional perturbation
effect $\tau_\psi$ as a quantity that is similarly amenable to semiparametric
analysis and permits us to impose a partially linear model. To motivate the
upcoming condition, we consider what happens when $Z$ takes values in
$\mathbb{R}^d$. In that case, we can write (with $\omega_\psi$ defined in
Definition~\ref{def:perturbation})
\[
  \partial_{\gamma} f(\psi(z, \gamma))\big|_{\gamma = 0} = \nabla f(z)^\top \omega_\psi(z).
\]
We want to isolate the change of $\psi$ along $v_{\psi}$ in a real-valued
variable $L$, keeping the remaining `information' in $Z$ in a variable $W$ with
values in the set $\mathcal{W}$. For a reparametrization $\phi$ of $z$ into
$(\ell, w) \in \mathbb{R} \times \mathcal{W}$, we define the reparametrized $f$
as $g: (\ell, w) \mapsto f(\phi^{-1}(\ell, w))$. Then, in order for the
reparametrization $\phi$ to separate the 'information' appropriately, we require
\[
  \partial_\ell g(\ell, w) = \partial_\ell f(\phi^{-1}(\ell, w)) = \nabla f(\phi^{-1}(\ell, w))^\top \partial_\ell \phi^{-1}(\ell, w) 
\]
to equal $\partial_{\gamma} f(\psi(z, \gamma))\big|_{\gamma = 0}$, that is,
\[
  \omega_\psi(\phi^{-1}(\ell, w)) = \partial_\ell \phi^{-1}(\ell, w).
\]
This derivation is only morally correct when $Z \in
\Delta^{d-1}$ ($\nabla f$ is not defined in the usual sense for
simplex-valued functions) and to make it formally correct one
needs to apply techniques from differential geometry. The proof of the
following result can be found in Section~\ref{sec:deriv-iso_proof}.

\begin{proposition}[Isolating derivatives]
  \label{prop:derivative-isolating}
  Consider the setting of Definition~\ref{def:avg_per_eff}. Let $\mathcal{W}$ be
  a set and $(\mathcal{L}_w)_{w \in \mathcal{W}}$ a family of subsets of
  $\mathbb{R}$. Suppose that $\phi = (\phi^L, \phi^W): \mathcal{Z} \to \{(\ell,
  w) \in \mathbb{R} \times \mathcal{W} \mid w \in \mathcal{W}, \ell \in
  \mathcal{L}_w\}$ is a bijection and $\psi$ a directional perturbation such
  that for all $z \in \mathcal{Z}$ it holds that $\omega_\psi(z) \neq 0$.
  Suppose that $\phi$ is \emph{derivative-isolating}, that is, for all $w \in
  \mathcal{W}$ and all $\ell \in \mathcal{L}_w$,
  \begin{equation}
    \label{eq:derivative-isolating}
    \omega_\psi(\phi^{-1}(\ell, w)) = \partial_{\ell} \phi^{-1}(\ell, w).
  \end{equation}
  Then, defining $g: (\ell, w) \mapsto f(\phi^{-1}(\ell, w))$ and $(L, W) =
  \phi(Z)$, the average perturbation effect satisfies
  \begin{equation}
    \label{eq:simple_adpe}
    \tau_\psi = \E\left[\partial_\ell g(\ell, W) \bigm|_{\ell = L}\right].
  \end{equation}
\end{proposition}

If a derivative-isolating reparametrization can be found, the average
directional perturbation effect is an average partial derivative of a
conditional expectation \eqref{eq:simple_adpe} -- a quantity previously studied
using semiparametric theory \citep{newey1993efficiency}. It can be estimated
nonparametrically or by assuming a partially linear model, that is, that there
exist $\theta \in \mathbb{R}$ and $h$ with
$
  \E[Y \given L, W] = \theta L + h(W).
$
The estimation approaches are discussed in Section~\ref{sec:estimation}.

\begin{remark}[Zero speeds]
  \label{rmk:zeros}
  While we do not allow $\omega_\psi(z) = 0$ in the definition of
  derivative-isolating reparametrizations, it is clear from the definition of
  $\tau_\psi$ that such points do not contribute to the perturbation effect. We
  can therefore allow for such points if we appropriately adapt the estimates.
  In Section~\ref{sec:zeros}, we describe how to
  modify the estimation procedures introduced in Section~\ref{sec:estimation} to
  allow for valid estimation even when $\omega_\psi(Z) = 0$ with positive
  probability.
\end{remark}

\section{Perturbations on the simplex} \label{sec:simplex_perturbations} We now
turn to the construction of perturbations on the simplex. For binary
perturbations, we only need to specify the endpoint of the perturbation, while
for directional perturbations it suffices to choose a direction and a speed to
construct a perturbation of the form \eqref{eq:lin_pert}. For a fixed endpoint
there is a canonical (geometry-dependent) choice of direction that is given by
the straight line from a point to its endpoint. The choice of speed is, however,
more subtle and it should be selected in a way that provides a clear
interpretation. We propose two general approaches for selecting speeds: (a) by
directly specifying an interpretable speed and (b) by starting from an
interpretable summary statistic, which decreases when moving away from the
endpoint, and interpreting the resulting effect as that of slightly increasing
the statistic. For both approaches, we provide -- under additional assumptions
-- explicit derivative-isolating reparametrizations, which can then be used for
efficient estimation. 

\subsection{Endpoints in the simplex}\label{sec:endpoints}
An endpoint function $\mathsf{E}:\Delta^{d-1}\rightarrow\Delta^{d-1}$, if it
exists, gives the point where the perturbation ends when $\gamma$ is increased
as far as possible. Definition~\ref{def:perturbation} extracts $\mathsf{E}_\psi$
from a given $\psi$; here we start from $\mathsf{E}$ and construct perturbations
realising it, and so drop the subscript.

The simplest type of endpoint function on the simplex is a constant function,
that is, there exists a fixed point $z_* \in \Delta^{d-1}$ such that for all
$z\in\Delta^{d-1}$ it holds that $\mathsf{E}(z)=z_*$. For example, we can choose
the fixed endpoint $z_*=e_j$, which corresponds to a vertex on the simplex and
hence a perturbation with this endpoint moves proportions from the other
coordinates into the $j$th one. A further example, is to use the center of the
simplex, i.e.\ $z_*=z_{\cen}$. As alluded to in the introduction, perturbations
with $z_{\cen}$ as the endpoint can be seen as perturbing observations in the
direction of equal distribution and hence capture a notion of diversity. For
later reference, we call all perturbations $\psi$ with endpoint function
$\mathsf{E}_{\psi} : z \mapsto z_{\cen}$ \emph{centering} perturbations.

A second type of endpoint function, that can be non-constant, is given by
mapping each $z \in \Delta^{d-1}$ to amalgamations on the simplex. An
amalgamation is a well-studied operation \citep[e.g.,][]{aitchisonbook} that
adds together groups of components of $z$, giving a lower dimensional
composition. Formally, for disjoint $A, B \subseteq [d]$, we define the
\emph{$(A \to B)$-amalgamation operator}
$\A_{A \to B}: \Delta^{d-1} \to \Delta^{d-1}$ for all
$z\in\Delta^{d-1}$ by
\[
  \A_{A \to B}(z)^A := 0, \quad \A_{A \to B}(z)^B := \C(z^B)(\|z^A\|_1 + \|z^B\|_1), \quad \A_{A \to B}(z)^{[d] \setminus (A \cup B)} := z^{[d] \setminus (A \cup B)}.
\]
The $(A\to B)$-amalgamation has the effect of setting components in $A$ to $0$,
scaling up the components in $B$ such that the subcomposition in $B$ remains
fixed (i.e., $\C(\A_{A \to B}(z)^B) = \C(z^B)$) and leaving the remaining
components fixed. We call all perturbations $\psi$ with endpoint function
$\mathsf{E}_{\psi}: z\mapsto \A_{A \to B}(z)$ \emph{amalgamating} perturbations.
General amalgamating perturbations are discussed in
Section~\ref{sec:amalgamations}. Here we focus on those with $A = [d] \setminus
\{j\}$ and $B = \{j\}$ so that $\A_{A \to B}(z) = e_j$, i.e.\ perturbations that
push towards a vertex of the simplex.

\subsection{Directions, speeds and derivative-isolating reparametrizations}
\label{sec:deriv-iso}

When constructing directional perturbations it suffices to consider $\psi$
parametrized by direction $v_\psi$ and speed $s_\psi$ as
$
  \psi: (z, \gamma) \mapsto z + \gamma s_\psi(z)v_\psi(z).
$
An endpoint function $\mathsf{E}:\mathcal{Z}\rightarrow\Delta^{d-1}$ gives rise
to a canonical direction, defined for all $z\in\{z \in \Delta^{d-1} \mid
\mathsf{E}(z) \neq z\}$, by
\begin{equation}
  \label{eq:endpoint_dir}
  v_{\mathsf{E}}(z) := \frac{\mathsf{E}(z)-z}{\|\mathsf{E}(z)-z\|_1},
\end{equation}
that is, for each $z$ not already at its own endpoint, $v_{\mathsf{E}}(z)$ is
the unit vector in the direction of the straight line from $z$ to the endpoint.

To avoid zero speeds, we restrict the domain of the perturbations to $\{z\in
\Delta^{d-1} \mid \mathsf{E}(z) \neq z\}$ but as discussed in
Remark~\ref{rmk:zeros} the target parameters and estimates can easily be
extended to settings where zero speeds are allowed. In
Section~\ref{sec:other_geometries} we describe how working in other geometries
can lead to different notions of direction. Taking the direction in
\eqref{eq:endpoint_dir} as given, we now select a speed $s_{\psi}$ giving an
interpretable average directional perturbation effect, and construct a
derivative-isolating reparametrization for efficient estimation.

As a building block for the general construction, consider $\psi$ with unit
speed $s_{\psi}=1$ and direction $v_{\psi}=v_{\mathsf{E}}$, that is, for all $z
\in \Delta^{d-1}$ with $z \neq \mathsf{E}(z)$ and all $\gamma \in [0,
\|\mathsf{E}(z)- z\|_1]$ we define
\[
  \psi(z, \gamma) := z + \gamma \frac{\mathsf{E}(z)-z}{\|\mathsf{E}(z)-z\|_1}. 
\]
A derivative-isolating reparametrization for $\psi$ should separate changes
along $v_{\mathsf{E}}$ from changes due to varying directions. We therefore
consider, for all $z\in\mathcal{Z}$, $\phi_{\unit}(z):=\bigl(-\| \mathsf{E}(z) -
z\|_1, (\mathsf{E}(z), v_{\mathsf{E}}(z))\bigr)$. The inverse of this
reparametrization is $\phi_{\unit}^{-1}(\ell, (w_{\mathsf{E}}, w_v)) = \ell w_v
+ w_{\mathsf{E}}$, so
\begin{equation*}
  \partial_{\ell} \phi_{\unit}^{-1}(\ell, (w_{\mathsf{E}}, w_v)) =
  w_v=\partial_{\gamma}\psi(\phi_{\unit}^{-1}(\ell, (w_{\mathsf{E}}, w_v)), \gamma)\big\vert_{\gamma=0}=\omega_{\psi}(\phi_{\unit}^{-1}(\ell, w))
\end{equation*}
and hence $\phi_{\unit}$ is derivative-isolating for $\psi$. 

A derivative-isolating reparametrization moves $L$ by one unit for each unit of
$\gamma$, so an average directional perturbation effect is always the expected
change in $Y$ per unit increase in $\phi^L$ while fixing $\phi^W$. Choosing a
speed is therefore choosing the units in which the effect is reported. At each
$z$ the speed multiplies the derivative of $f$ in the normalized direction
$v_\psi(z)$, so two perturbations with the same direction have proportional
effects when their speeds differ by a constant factor, but not in general
otherwise. For the unit-speed perturbation above, $\phi^L_{\unit}$ is the
negative $\ell^1$ distance to the endpoint, so the effect is expressed per unit
of movement towards $\mathsf{E}(z)$. This is a natural geometric scale, though
not necessarily the one in which a given application measures change. For this
reason, a unit-speed perturbation may not always be easy to interpret, so we
extend the basic construction above in Section~\ref{sec:derivative-isolating} to
start from either a given interpretable speed or a summary statistic.

\subsection{Effects of changes in individual components and diversity}
\label{sec:simplex_effects}
We now give examples of directional perturbations summarizing effects of practical interest.

\begin{example}[Feature influence]
  \label{ex:cfi}
  Suppose we are interested in the effect of increasing the value of a
  particular component $z^j$ of $z \in \Delta^{d-1}$. We can investigate the
  effect of such a change by means of a $([d] \setminus \{j\} \to
  \{j\})$-amalgamating perturbation, that is, a perturbation which satisfies for
  all $z \in \Delta^{d-1} \setminus \{e_j\}$ that $v_\psi(z) = \frac{e_j -
  z}{\|e_j-z\|_1}$.  While for some purposes it may be perfectly reasonable to
  use a constant speed $s_\psi\equiv 1$, in other settings this speed can be
  unintuitive. For example, if we imagine the observed $Z$ as being generated
  from an unobserved $X \in \mathbb{R}^{d}_+ \setminus \{0\}$ on an absolute
  scale by $Z:= \C(X) = \frac{1}{\|X\|_1} X$, it might be more natural to
  consider which speeds are induced on the simplex from hypothetical changes to
  $X^j$. A first attempt might be to consider the speed resulting from adding a
  constant $c$ to $X^j$. However, since
  \[
    \| \partial_c \C(X + c e_j) \bigm|_{c = 0} \|_1 = 2\tfrac{1}{\|X\|_1}
    \left(1 - \tfrac{X^j}{\|X\|_1}\right),
  \] 
  this speed is not scale-invariant for $X$, since it may differ for the counts
  $X$ and $cX$ when $c>0$. The speed on the simplex would then depend on the
  unobserved $\|X\|_1$, so the perturbation is ill-defined. Instead, multiplying
  $X^j$ by $1+c$ (so that $c=0$ is no change) is scale-invariant and yields a
  well-defined perturbation on the simplex. More formally, letting $\odot$
  denote the Hadamard (point-wise) product, the resulting speed is
 \[
   \| \partial_c \C(X \odot (1 + c e_j)) \bigm|_{c = 0} \|_1 =2\tfrac{X^j}{\|X\|_1}
     \left(1 - \tfrac{X^j}{\|X\|_1}\right),
 \]
  which is scale-invariant and thus meaningful on the simplex, where it amounts
  to the speed $s_\psi(z) = 2 z^j (1 - z^j)$ denoted \emph{multiplicative
  speed}. By Proposition~\ref{prop:find_deriv-iso}~(a) (see
  Example~\ref{ex:feature_influence_details}),
\[
  \phi = (\phi^L, \phi^W): z \mapsto \left(\log\left(\frac{z^j}{1-z^j}\right), v_\psi(z) \right)
\]
is derivative-isolating for the amalgamating perturbation $\psi$ with multiplicative speed.
\end{example}

Perturbation effects that capture changes in individual components, are, for
example, useful as a variable selection tool in compositional data. In
Section~\ref{sec:microbiome}, we illustrate this for a human gut microbiome
study. For each study participant the body mass index (BMI) $Y$ and their gut
microbiome composition $Z$ are measured. We then want to determine for which
microbes abundance is associated with BMI. The simplex constraint
implies that for all coordinates $j$ we have that $Y$ is independent of $Z^j$
given $Z^{-j}$ hence any feature importance or influence measure based on
measuring this conditional dependence is meaningless.

We define the \emph{compositional feature influence} for the $j$th feature or
simply $\cfi^j$ as $\tau_\psi$ for any directional perturbation with endpoint
$e_j$. We denote by $\cfi^j_{\unit}$ the perturbation that moves at a constant
(unit) speed towards the vertex and $\cfi^j_{\mult}$ the perturbation that moves
at multiplicative speed. $\cfi_{\mult}^j$ has previously been proposed by
\citet{huang2022supervised}, however without an efficient estimation procedure.
The two speeds differ crucially in how they treat observations with $Z^j=0$. As
the multiplicative speed is $0$ in this case such points do not contribute to
$\cfi^j_{\mult}$, however, they do contribute to $\cfi^j_{\unit}$.

In datasets with many zeros, which is common for example in microbiome data, it
can also be interesting to focus only on the effect of zeros in individual
components. This can be achieved by considering a binary perturbation $\psi$
given by $\psi(z, 1)^j = 0$ and $\psi(z, 1)^{-j} = \C(z^{-j})$. We call the
resulting average binary perturbation effect $\lambda_\psi$ the
\emph{compositional knock-out effect} for the $j$th feature, or $\cke^j$ for
short. The $\cke^j$ permits a straightforward analysis of the effect of zeros, a long-standing
problem for existing approaches.

\begin{example}[Diversity with Gini-speed]
  \label{ex:gini}
  Suppose we are interested in the effect of diversifying $z \in \Delta^{d-1}$.
  We can investigate such an effect by using a centering directional
  perturbation, that is, a perturbation which satisfies for all $z \in
  \Delta^{d-1} \setminus \{z_{\cen}\}$ that $v_\psi(z) = \frac{z_{\cen} -
  z}{\|z_{\cen}-z\|_1}$. In contrast to Example~\ref{ex:cfi} it is less obvious how
  to choose an interpretable speed. A natural choice of similarity measure is
  the Gini coefficient $G: \Delta^{d-1} \to [0, 1]$, defined for all $z \in
  \Delta^{d-1}$ by
  \begin{equation}
    \label{eq:gini}
    G(z) := \frac{1}{2d} \sum_{j=1}^d \sum_{k=1}^d |z^j - z^k|,
  \end{equation}
  whence $D(z) := 1 - G(z)$ is a diversity measure. By
  Proposition~\ref{prop:find_deriv-iso}(b) (see Example~\ref{ex:gini_details}),
  \[
    \phi = (\phi^L, \phi^W) : z \mapsto \left(1-G(z), \frac{z_{\cen}-z}{\|z_{\cen}- z\|_1} \right)
  \]
  is derivative-isolating for perturbations with direction $v_\psi$ and speed
  \begin{equation}
    \label{eq:gini_speed}
    s_\psi(z) := \frac{2d}{ \sum_{j=1}^d \sum_{k=1}^d |v_\psi(z)^j - v_\psi(z)^k|}.
  \end{equation}
  Similar arguments apply to other diversity measures such as the Shannon
  entropy or the Gini--Simpson index.
\end{example}

The effect of increasing diversity is often relevant in social sciences and
econometrics, e.g., when considering racial or income distributions. As we
outlined in the introduction, it can be misleading to only model summary
statistics of diversity and ignore the remaining compositional structure as this
structure can confound the relationship between the diversity measure and the
response. Instead, we propose to consider the effect of directional
perturbations that push towards $z_{\cen}$ (centering perturbations as defined
in Section~\ref{sec:endpoints}) and call $\tau_\psi$ for such perturbations a
\emph{compositional diversity influence} ($\cdi$). We denote the effect by
$\cdi_{\unit}$ if the perturbation has unit speed and by $\cdi_{\gini}$ if the
perturbation was constructed as in Example~\ref{ex:gini} such that $L$ is equal
to one minus the Gini coefficient. In Section~\ref{sec:diversity_exps} we give a
semi-synthetic and real data example of the application of $\cdi_{\gini}$.

Not every perturbation one might wish to consider admits a
derivative-isolating reparametrization. The construction requires the speed of a
perturbation to be strictly positive. Under the multiplicative speed of
Example~\ref{ex:cfi} this excludes every observation with $z^j=0$, which in
microbiome data is typically the majority. This restriction is not binding in
practice: as noted above, such points simply do not contribute to
$\cfi^j_{\mult}$, and Remark~\ref{rmk:zeros} and Section~\ref{sec:zeros}
describe how to adapt the estimators so that the remaining observations can
still be used. A genuine limitation arises instead for perturbations specified
through a summary statistic, as in Example~\ref{ex:gini}. Here the statistic
must be differentiable and must increase monotonically towards the endpoint. A
statistic that is not differentiable, such as the number of components with
$z^j>0$, fails the first requirement, and one whose value does not increase
towards the chosen endpoint fails the second. This would happen, for instance,
if one took the statistic to be a single component $z^j$ while perturbing
towards $z_{\cen}$: moving towards the centre raises $z^j$ when $z^j<1/d$ but
lowers it when $z^j>1/d$, so the requirement holds on part of the simplex and
fails on the rest. When a statistic fails one of these two requirements,
$\tau_\psi$ remains well defined but what is lost is the reduction to an average
derivative in \eqref{eq:simple_adpe}, and hence the estimators of
Section~\ref{sec:estimation}. The precise conditions are those of
Proposition~\ref{prop:find_deriv-iso}.

In Section~\ref{sec:amalgamations} we discuss average perturbation effects that
capture the effect of applying amalgamations to $Z$ on $Y$ (see
Section~\ref{sec:endpoints} for a definition of a general amalgamation). 

We summarize our predefined perturbation effects in Table~\ref{tab:targets}. The
units of each effect can be read off its $\phi^L$ column, and the choice of
speed groups them. For the unit-speed effects, $\phi^L$ is the negative $\ell^1$
distance to the endpoint, so one unit is a movement of one in $\ell^1$ towards
it. For $\cfi^j_{\unit}$ this is an increase of $1/2$ in $z^j$, since the
sum-to-one constraint makes the gain in component $j$ equal the total loss
across the others and the $\ell^1$ norm counts both. For the
multiplicative-speed effects $\phi^L$ is a log-ratio, so one unit is an increase
of one in the log-odds of one part against another. The Gini-speed diversity
effect is in units of the diversity index $1-G$, which takes values in $[1/d,
1]$. For the binary effects speed is irrelevant, as the effects are contrasts.
Section~\ref{sec:microbiome} illustrates how the choice of speed can matter in
practice.

\begin{table}[h]
  \centering
  \renewcommand{\arraystretch}{1.3}
  \begin{tabular}{llccc}
    \toprule
    Effect of changes in & Target & Speed $s(z)$ & Endpoint $\mathsf{E}(z)$ & $\phi^L(z)$  \\ \toprule
    \multirow{3}*{individual components} & $\cfi^j_{\unit}$ & $1$  & $e_j$  &  $-2(1-z^j)$\\
                          & $\cfi^j_{\mult}$ & $2 z^j (1 - z^j)$ & $e_j$ & $\log(\frac{z^j}{1-z^j})$ \\
                          & $\cke^j$ & binary & $e_j$ & $\ind_{\{z^j=0\}}$ \\
    \midrule
    \multirow{2}*{diversity} & $\cdi_{\unit}$ & $1$ & $z_{\cen}$  & $-\|z_\cen-z\|_1$ \\
                         & $\cdi_{\gini}$ & see \eqref{eq:gini_speed}  & $z_{\cen}$ & $1-G(z)$, see \eqref{eq:gini}\\ 
    \midrule
    \multirow{3}*{amalgamations} & $\cai^{A \to B}_{\unit}$ & $1$ & $\A_{A \to B}(z)$ & $-2\|z^A\|_1$ \\
                          & $\cai^{A \to B}_{\mult}$ & $\frac{2 \|z^A\|_1 \|z^B\|_1}{(\|z^A\|_1 + \|z^B\|_1)^2}$ & $\A_{A \to B}(z)$ & $\log\big(\tfrac{\|z^B\|_1}{\|z^A\|_1}\big)$ \\
                          & $\cae^{A \to B}$ & binary & $\A_{A \to B}(z)$ & $\ind_{\{z^A=0\}}$ \\
    \bottomrule
  \end{tabular}
  \caption{List of perturbation effects considered in this work and
    what effects they capture. For each perturbation effect the
    corresponding speed, endpoint and $L$-term of the
    reparametrization used for efficient estimation is included. All
    reparametrizations use the canonical direction $v_{\mathsf{E}}$ based on the
    endpoint, see \eqref{eq:endpoint_dir}. The $W$-term $\phi^W(z)$ of
    the reparametrization is $\mathsf{E}(z)$ for the binary effects and
    $(\mathsf{E}(z), v_{\mathsf{E}}(z))$ for the directional effects.}\label{tab:targets}
\end{table}

\section{Estimating perturbation effects}\label{sec:estimation}

In this section, we describe how to estimate and construct confidence intervals
for the target parameters $\tau_\psi$ and $\lambda_\psi$. We do not restrict to
any one perturbation $\psi$, as the approaches described are agnostic to this
choice. However, we will assume that any directional perturbation comes with a
derivative-isolating reparametrization $\phi = (\phi^L, \phi^W)$ and for any
binary perturbation we define $\phi^L : z \mapsto \ind_{\{z=\psi(z, 1)\}}$ and
$\phi^W : z \mapsto \psi(z, 1)$. This reparametrizes the estimation problem
from a compositional one with $n$ i.i.d.\ observations of $(Y, Z) \in \mathbb{R}
\times \Delta^{d-1}$ into a generic one with $n$ i.i.d.\ observations of $(Y, L,
W)$, $L \in \mathbb{R}$, $W \in \mathcal{W}$, where our interest lies in the
relationship between $Y$ and $L$ given $W$. We write $f$ for the regression
function on either scale, so that $f(z) = \E[Y \given Z=z]$ and $f(\ell, w) =
\E[Y \given L=\ell, W=w]$; the two agree in the sense that $f(\ell, w) =
f(\phi^{-1}(\ell, w))$, which is the function denoted $g$ in
Proposition~\ref{prop:derivative-isolating}.

Because the compositional structure is absorbed entirely into this
reparametrization, what remains is an ordinary semiparametric problem: the
conditions under which our estimators are asymptotically normal, given in
Section~\ref{sec:theory}, are the standard conditions for that problem and make
no reference to the simplex. A high-dimensional or sparse composition therefore
affects the estimators only through the difficulty of the associated nuisance
estimation problems, exactly as in any other application of these techniques: a
large number of components is not by itself an obstacle, but estimating the
nuisance functions sufficiently well generally requires structure, such as
smoothness or sparsity, that the chosen regression method can exploit. The one
consideration that remains specific to the simplex is that a speed vanishing on
the boundary reduces the effective sample size; under the multiplicative speed,
for instance, observations with $z^j=0$ do not contribute to $\cfi^j_{\mult}$. 

We give an overview of the estimators here; algorithms and theory are in Section~\ref{sec:theory}.

\subsection{Efficient estimation using semiparametric theory}
We briefly introduce semiparametric estimation of statistical functionals and
describe how naive plug-in estimation based on nonparametric methods rarely
leads to efficient confidence intervals, and how to remedy this. This
perspective is not new and goes back to at least
\citet{pfanzagl1982contributions} with several more recent monographs and
reviews
\citep{bickel1993efficient,tsiatis2006semiparametric,chernozhukov2018double,kennedy2023semiparametric}.
Their recent popularity is fueled by machine learning methods that use flexible
black box models to learn the nuisance functions involved. For illustration, we
focus on the estimation of $\tau_\psi$, that is, the estimation of
$
  \tau_\psi = \E[\partial_\ell f(\ell, W) \bigm|_{\ell = L}],
$
where $f: (\ell, w) \mapsto \E[Y \given L=\ell, W=w]$. A naive
approach, which is known as \emph{plug-in estimation}, is to
estimate $\tau_{\psi}$ by first learning a differentiable estimate
$\widehat{f}$ of $f$ using all $n$ observations and then computing
$
  \widetilde{\tau}_\psi := \frac{1}{n} \sum_{i=1}^n \partial_\ell \widehat{f}(L_i, W_i).
$
While such an estimate can be consistent, it is rarely asymptotically
Gaussian, which invalidates standard asymptotic inference. The plug-in
estimate of the asymptotic variance, given by
\begin{equation}
  \label{eq:naive_variance}
  \widetilde{\sigma}^2 := \frac{1}{n} \sum_{i=1}^n (\partial_\ell \widehat{f}(L_i, W_i))^2 - \widetilde{\tau}_\psi^2,
\end{equation}
hence does not provide valid inference (see
Figure~\ref{fig:semiparametric_illustration} (left)) as the confidence intervals
can be biased or too narrow.
\begin{figure}[!ht]
  \centering
  \includegraphics*[width=\textwidth]{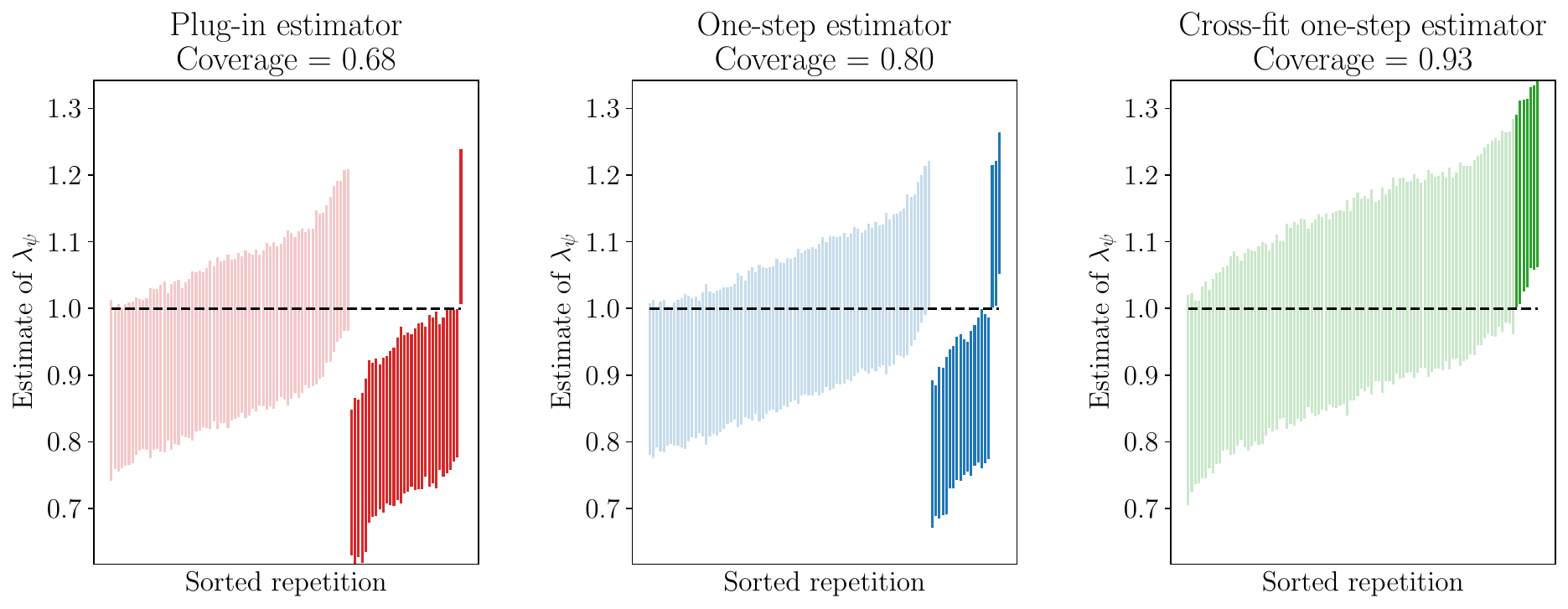}
  \caption{$95\%$ confidence intervals for plug-in, one-step and cross-fit
  one-step estimators of $\lambda_\psi$ over 100 repetitions in the partially
  linear, binary simulation of Section~\ref{sec:est_exp} with $d=15$, $n=1000$.
  The dashed line is the true parameter value; intervals not containing it are
  highlighted. Intervals are ordered first by whether they contain the true
  value, then by the estimate.}
  \label{fig:semiparametric_illustration}
\end{figure}
There are two main issues with this estimate: (a) repeated usage of the same
data for both the estimation of $f$ and the computation of the target
(double-dipping) and (b) estimation bias inherited from the bias in
$\widehat{f}$. Issue (a) can be resolved by splitting the data into two halves,
using the first to compute $\widehat{f}$ and the second to compute
$\widetilde{\tau}_\psi$. Swapping the roles of the halves and averaging improves
finite-sample performance. This procedure is known as \emph{cross-fitting}.
Sometimes it is preferable to split in more than two folds to increase the
sample size in the estimation of $\widehat{f}$ and to derandomize the resulting
estimate. Issue (b) requires the machinery of semiparametric theory. One
well-established approach is to start from the plug-in estimator and derive a
bias-correction, known as a \emph{one-step correction}, that can be added to the
plug-in estimator and weakens the conditions under which the estimator is
asymptotically Gaussian. The one-step correction can be found by computing the
influence function of the target parameter, which can be thought of as analogous
to a derivative when the target parameter is viewed as a function defined on a
space of distributions. For more details on this general procedure, see e.g.\
\citet{kennedy2023semiparametric}. These bias-corrected estimators are
asymptotically efficient, having the smallest possible asymptotic variance of
any estimator. Influence functions have already been computed for $\tau_\psi$
and $\lambda_\psi$ so they can be estimated entirely nonparametrically under
mild smoothness assumptions. We describe these estimators in
Section~\ref{sec:np_estimation}. Figure~\ref{fig:semiparametric_illustration}
illustrates the coverage of the resulting confidence intervals for the different
methods above.

\subsection{Partially linear estimation of perturbation effects}
\label{sec:estimation_plm} 
We now present a different proposal for estimating perturbation effects. The
aforementioned nonparametric estimators rely on accurate estimates of quantities
that may be difficult to estimate, such as derivatives of densities and
regression functions (see Section~\ref{sec:estimator_robustness} for an
illustration of this sensitivity). These estimations can be avoided by
additionally assuming a partially linear model on $L$ and $W$, for both binary
and directional perturbations (with a derivative-isolating reparametrization),
that is, the existence of $h$ and $\theta$ such that
$
  \E[Y \given L, W] = \theta L + h(W).
$
With this assumption, $\theta$ is exactly the perturbation effect. Efficient
estimation of $\theta$ in a partially linear model is well studied
\citep[e.g.,][]{robinson1988root,chernozhukov2018double}. The \emph{partialling
out} estimator is
\begin{equation}
  \label{eq:PLM_effect}
  \widehat{\theta} := \frac{\sum_{i=1}^n \{Y_i - \widehat{g}(W_i)\}\{L_i - \widehat{m}(W_i)\}}{\sum_{i=1}^n \{L_i - \widehat{m}(W_i)\}^2},
\end{equation}
where $\widehat{g}$ is an estimator of $g:w\mapsto \E[Y \given W=w]$ and
$\widehat{m}$ is an estimator of $m: w \mapsto \E[L \given W=w]$. In contrast to
the estimators above, this one relies on estimating only two conditional mean
functions, at the expense of assuming no interaction between $L$ and $W$ on $Y$
and a linear effect of $L$. Even when these assumptions are not satisfied, the
estimator may still provide a reasonable estimate of the effect of $L$ on $Y$:
one can show that the quantity $\E[(\E[Y \given L, W] - \theta L - h(W))^2]$ is
minimized by the population version of the partially linear model estimate, so
one can think of the partially linear model as a `best approximation' to a
generic nonparametric model (see Proposition~\ref{prop:plm_approximation} in
Section~\ref{sec:estimator_robustness_plm} for a proof). Separately, its
nuisance functions are easier to estimate than those of the nonparametric
estimators (see Sections~\ref{sec:estimator_robustness_cont}
and~\ref{sec:estimator_robustness_binary}), and in our experiments it is
correspondingly more robust to errors in their estimation (see
Section~\ref{sec:est_exp}). Full algorithms for both are given in
Section~\ref{sec:theory}.

\section{Numerical experiments}\label{sec:numerical_experiments} 
We now present numerical experiments illustrating the usage and performance of
our methods. In addition to the experiments included here,
Section~\ref{sec:est_exp} contains an
investigation of the performance of the semiparametric estimators discussed in
Section~\ref{sec:estimation} when the conditioning variable $W$ is
compositional. The code for all of our experiments is available at
\url{https://github.com/ARLundborg/PerturbationCoDa}. For all regressions we use
the \texttt{scikit-learn} package \citep{scikit-learn} except for the
$\ell^1$-penalized log-contrast regression in Section~\ref{sec:microbiome} where
we used the \texttt{c-lasso} package \citep{simpson2021classo}.

\subsection{Estimating effects of diversity on socio-economic outcomes}
 \label{sec:diversity_exps}
\subsubsection{Effects of diversity on income in semi-synthetic US census
data}
\label{sec:diversity_income}
In this section we provide an example based on semi-synthetic data that
illustrates how failing to adjust for compositional confounding can lead to
incorrect conclusions about effects of compositional variables. We consider the
`Adult' dataset available from the UCI machine learning repository
\citep{adult_dataset}. This dataset is based on an extract from the 1994 US
Census database and consists of data on 48,842 individuals. The response
variable in the data, \texttt{compensation}, is a binary variable where $0$ and
$1$ encode whether an individual makes below or above 50,000 USD per year,
respectively. We also consider \texttt{age} (numeric, 17 to 90),
\texttt{education} (categorical, 16 categories), \texttt{sex} (binary, $0$ male
and $1$ female) and \texttt{race} (categorical, 'White', 'Black' and 'Other').

We now use this raw individual-level data to construct community-level datasets
by artificially grouping individuals together.  The grouping is deliberately
constructed in a way such that diversity affects income and this relationship is
confounded by the remaining variation of the racial composition. This
semi-synthetic example is only intended as an illustration. A genuine analysis
of this data would require knowledge of the true community structure, which is
not available in the data extract we are using here. We emphasize that even if
the individual-level data is available, one may still be interested in effects
of the racial composition of the community an individual lives in on that
individual's income. For such an analysis, one can add the racial composition of
an individual's community as a variable and then perform the same analysis as
here but now on the individual-level data while adjusting for individual-level
confounding factors.  Importantly, even for such an individual-level analysis,
the simplex confounding can bias the results in the same way as we illustrate
here. We provide an example of this bias occurring on the individual-level effect
in Section~\ref{sec:diversity_income_individual}.

We group all individuals into groups of around $50$. The response $Y\in[0,1]$ is
the average \texttt{compensation} in each group. The predictor
$X\in\mathbb{R}\times[0,1]\times\{1, \dots, 16\}$ averages \texttt{age} and
\texttt{sex} and takes the majority class of \texttt{education} (ties broken
randomly). Finally, $Z\in\Delta^2$ holds the proportions of `White', `Black' and
`Other' individuals in each group.

To ensure a (partially) known effect of increased racial diversity on
compensation, we construct three categories of groups (0, 1 and 2). The
categories are based on either education or compensation: category 1 groups are
sampled to be similarly diverse and with lower educational (compensation) levels
than the total population, category 2 groups to be more diverse and with higher
levels, and category 0 groups randomly from the remaining individuals.

We then sample groups from each category, using a weighted sampling approach
described in detail in Section~\ref{sec:diversity_groups}. This results in $978$
observations with $811$ in category $0$, $92$ in category $1$ and $75$ in
category $2$ when grouping on education and $782$ in category $0$, $129$ in
category $1$ and $67$ in category $2$ when grouping on compensation. By
construction there is a positive effect of increased diversity on compensation,
since education is positively associated with compensation. The data and this
confounding are visualized in Section~\ref{sec:additional_adult}.

Based on the two constructed community-level datasets, we now compare a naive
approach that neglects the simplex structure and is applied in the literature
(e.g.
\citet{richard2000racial,richard2007impact,lu2015effect,paarlberg2018heterogeneity})
with the perturbation-based approach. The naive approach summarizes racial
diversity in a single measure, here we use one minus the Gini coefficient
$L=1-\frac{1}{2d}\sum_{j=1}^d \sum_{k=1}^d |Z^j - Z^k|$ as in
Example~\ref{ex:gini}, for which large values correspond to high diversity, and
estimate associations between this measure and the average compensation $Y$. We
call the method that regresses $Y$ on $L$ linearly \texttt{naive\_diversity}.
This estimator can easily be adapted to also adjust for the additional
covariates $X$, the racial composition $Z$ or both under a partially linear
model assumption. The point estimates together with asymptotic confidence
intervals (based on an efficient partialling out estimator like the estimator in
Section~\ref{sec:estimation_plm}) are
provided in Table~\ref{tab:adult} where random forests are used for nuisance
function estimation when including additional covariates. The naive estimators
that condition either on the empty set or on $X$ both result in significant
negative effects which by construction of the datasets is incorrect.
Furthermore, additionally conditioning on $Z$, leads to an ill-conditioned
estimator as $Z$ is now a perfect predictor of the diversity measure. This is
also visible in the large resulting confidence intervals. While we here
constructed the groups in a way to highlight how compositional confounding can
be misleading, we believe that this type of confounding is likely also present
in real-world settings as groupings are often confounded e.g., educational
levels may well be correlated with race and average compensation in actual
communities.

These pitfalls can be avoided with a perturbation-based analysis. In this
setting, we propose to apply $\cdi_{\gini}$ with and without conditioning on the
additional covariates $X$. The results are also shown in Table~\ref{tab:adult}
again using random forests for nuisance function estimation. As desired one gets
significant positive effects with this method as it is adapted to the simplex.
For the grouping on education, the estimate $\cdi_{\gini} = 0.233$ says
that a $0.1$ increase in the diversity index corresponds to $2.3$ percentage
points more of a group earning above $50{,}000$ USD, against a base rate of
$24\%$, while the naive estimate of $-0.082$ gives $0.8$ percentage points in
the opposite direction. Interestingly, for $\cdi_{\gini}$ that adjusts for $X$
the effect disappears for the grouping based on education. This is expected
because the association between diversity and compensation was constructed based
on education which is also part of $X$.

\begin{table}
  \centering
  \begin{tabular}{l|cc|cc}
    \toprule
      \multirow{2}*{Method}  & \multicolumn{2}{c|}{Grouping on education} & \multicolumn{2}{c}{Grouping on compensation} \\ 
                             & Estimate & $95\%$ CI & Estimate & $95\%$ CI \\
    \midrule
    $\texttt{naive\_diversity}$ & $\textcolor{LKred}{\mathbf{-0.082}}$ & $(-0.144, -0.019)$ & $\textcolor{LKred}{\mathbf{-0.120}}$ & $(-0.194, -0.046)$\\
    $\texttt{naive\_diversity} \given Z$ & $-0.232$ & $(-1.246, 0.781)$ & $-0.662$ & $(-1.892, 0.568)$\\
    $\texttt{naive\_diversity} \given X$ & $\textcolor{LKred}{\mathbf{-0.085}}$ & $(-0.142, -0.028)$ & $\textcolor{LKred}{\mathbf{-0.105}}$ & $(-0.178, -0.033)$\\
    $\texttt{naive\_diversity} \given X, Z$ & $-0.091$ & $(-0.870, 0.688)$ & $0.407$ & $(-0.394, 1.207)$ \\
    $\cdi_{\mathrm{Gini}}$ & $\textcolor{LKgreen}{\mathbf{0.233}}$ &$(0.060, 0.406)$ & $\textcolor{LKgreen}{\mathbf{0.614}}$ & $(0.429, 0.799)$\\
    $\cdi_{\mathrm{Gini}} \given X$ & $0.074$ & $(-0.046, 0.194)$ & $\textcolor{LKgreen}{\mathbf{0.611}}$ & $(0.455, 0.768)$\\
    \bottomrule
  \end{tabular}
  \caption{Table of estimates in the `Adult' data
    experiment. Bold red numbers indicate a significant negative
      effect of diversity at a $5\%$ level, while bold green indicates
      a significant positive effect. The
    \texttt{naive\_diversity} method indicates that either there is a
    significant negative or a highly uncertain effect of diversity on
    compensation despite the fact that the data is constructed with a
    positive effect of diversity. The CDI correctly identifies this
    effect.}\label{tab:adult}
\end{table}

\subsubsection{Effects of diversity on grades in New York schools}
\label{sec:american_schools}
We now analyze data from a 2018 Kaggle competition on schools in New York
\citep{passnycc}, containing average grades in English Language Arts (ELA) and
in math for 1217 schools along with additional background information on the
school such as a summary of the economic need of the students, percentage of
students with English as second language, racial composition of the school and
various indicators of school performance (for a full list of control variables
and our pre-processing of these, see
Section~\ref{sec:american_schools_preprocessing}). We let $X$ denote all
background variables except racial composition which we denote by $Z$. Race is
given as whole percentages of `Asian', `Black', `Hispanic' and `White' students
which we rescale so that $Z \in \Delta^{3}$. We are interested in the effect of
diversity, again measured by one minus the Gini coefficient, on the average
grade of the schools (either ELA or math). We repeat the analysis in
Section~\ref{sec:diversity_income} on this real data and the results can be seen
in Table~\ref{tab:american_schools}.

The \texttt{naive\_diversity} method indicates a significant effect of diversity
on grades even when controlling for $X$. In contrast, the CDI estimates smaller
positive effects for ELA and none distinguishable from $0$ for math. Concretely,
a $0.1$ increase in a school's diversity index corresponds to $0.035$ additional
proficiency points in ELA, against $0.108$ for the naive estimate. The
discrepancy indicates substantial compositional confounding, which invalidates
the marginal analysis.
\begin{table}
  \centering
  \begin{tabular}{l|cc|cc}
    \toprule
      \multirow{2}*{Method}  & \multicolumn{2}{c|}{Effect on ELA grades} & \multicolumn{2}{c}{Effect on math grades} \\ 
                             & Estimate & $95\%$ CI & Estimate & $95\%$ CI \\
    \midrule
    $\texttt{naive\_diversity}$ & $\textcolor{LKgreen}{\mathbf{1.079}}$ & $(0.934, 1.223)$ & $\textcolor{LKgreen}{\mathbf{1.231}}$ & $(1.040, 1.422)$\\
    $\texttt{naive\_diversity} \given X$ & $\textcolor{LKgreen}{\mathbf{0.300}}$ & $(0.192, 0.408)$ & $\textcolor{LKgreen}{\mathbf{0.354}}$ & $(0.212, 0.497)$\\
    $\cdi_{\mathrm{Gini}}$ & $\textcolor{LKgreen}{\mathbf{0.348}}$ &$(0.108, 0.589)$ & $0.101$ & $(-0.171, 0.373)$\\
    $\cdi_{\mathrm{Gini}} \given X$ & $\textcolor{LKgreen}{\mathbf{0.161}}$ & $(0.002, 0.320)$ & $0.042$ & $(-0.171, 0.255)$\\
    \bottomrule
  \end{tabular}
  \caption{Table of estimates in the New York schools data experiment. Bold
  green numbers indicate a significant positive effect of diversity at a $5\%$
  level. The \texttt{naive\_diversity} method indicates that there is a
  significant effect of diversity on grades even when controlling for additional
  covariates $X$. The CDI (with or without controlling for $X$) generally
  estimates smaller positive effects or even effects that cannot be
  distinguished from $0$, indicating a significant amount of compositional
  confounding in this data.}\label{tab:american_schools}
\end{table}

\subsection{Variable influence measures when predicting BMI from gut
  microbiome}\label{sec:microbiome}

In this section, we analyze the `American Gut'
\citep{mcdonald2018american} dataset which contains a collection of
microbiome measurements and metadata from over 10,000 human
participants. We use the pre-processed version used in
\citet{bien2021tree}
(available at
\url{https://github.com/jacobbien/trac-reproducible/tree/main/AmericanGut}). Our
focus is on the relationship between body mass index (BMI) and gut microbiome
composition at the species-level and our primary goal is to determine which
species are important for predicting BMI. We consider the subset of individuals
from the United States with BMI between 15 and 40, height between 145 and 220 cm
and age above 16 years. To reduce the dimensionality of the microbiome
composition (i.e., the number of species), we perform an initial variable
screening and only include species that appear in at least $10\%$ of the
observations and where the mean abundance is greater than $10^{-4}$ (this is
after scaling the absolute counts by the total number of reads for each
observation). The resulting dataset consists of 4,581 observations of
$\mathrm{BMI}$ measurements $Y\in\mathbb{R}$ and the relative abundances of
$561$ microbial species $Z\in\Delta^{560}$. 

A variable importance measure should be a property of the distribution of $(Y,
Z)$ rather than of the algorithm used to estimate it, so it ought to identify
the same features whichever sufficiently predictive regression method is used.
In Section~\ref{sec:microbiome_consistency} we show that neither
permutation-based feature importance \citep{breiman2001random} nor nonparametric
$R^2$ \citep{williamson2021nonparametric} does so here, since both vary a single
coordinate while holding the rest fixed, which the simplex constraint does not
permit. The perturbation effects, being functionals of the distribution by
construction, are estimated consistently across regression methods.

We therefore compare different average perturbation effects with the commonly
used compositional importance measures based on log-contrast models. This model
assumes there exists $\beta \in \mathbb{R}^d$ such that
\[
  \E[Y \given Z] = \sum_{j=1}^d \beta^j \log(Z^j) \quad \text{with} \quad \sum_{j=1}^d \beta^j = 0.
\]
We fit this model and an $\ell^1$-penalized version of the same model and seek
to use the coefficients as feature influences. Since the log-contrast model is
only defined on the open simplex, that is, it does not permit zeros in the data,
one usually adds a `pseudocount' to each observation before rescaling to the
simplex. We follow a standard recommendation and use the minimum non-zero
observation of $Z$ divided by $2$ (this is not required for the
perturbation-based measures). All perturbation effects in this section are
estimated using the partially linear estimator of
Section~\ref{sec:estimation_plm} with random forests for the nuisance
regressions.

To illustrate the difference between the approaches, we present the $p$-values
from $\cfi_{\mult}$, $\cfi_{\unit}$, $\cke$ and the log-contrast method for
several species in the left of Figure~\ref{fig:microbiome_confidence}. The
log-contrast and $\cke$ $p$-values are highly correlated and appear to select
similar species. In fact, across all $561$ species the rank correlation between
their estimates is $-0.79$, by far the strongest among the four methods. The
correlation is negative because the $\cke$ measures removal of a species while
the log-contrast coefficient measures increase. In comparison, the two
continuous measures correlate with the log-contrast estimates at only $0.37$ and
$0.32$. This is despite the log-contrast model being intended to model the
continuous effect of each species. The similarity to $\cke$ suggests (likely due
to the pseudocounts) that its coefficients instead capture presence and absence.
To investigate this, the right of Figure~\ref{fig:microbiome_confidence} plots
log-contrast estimates (scaled by the standard deviation of $\log(Z^j +
\text{pseudocount})$) against the pseudocount, highlighting how strongly the
coefficients depend on it.

That the perturbation-based effects explicitly account for zeros and distinguish
continuous effects (e.g., $\cfi_{\unit}$ or $\cfi_{\mult}$) from
presence/absence effects ($\cke$) is an important feature of the approach. Among
the continuous effects, the choice of speed determines the scale on which an
effect is reported. The two speeds differ by $2 z^j (1 -
z^j)$ rather than a constant, so $\cfi_{\unit}$ and $\cfi_{\mult}$ are not
proportional: they weight the derivative in the shared direction differently
across the simplex, and species with small abundances contribute little to
$\cfi_{\mult}$. Although $\cfi_{\unit}$ and $\cfi_{\mult}$ rank species
similarly (Spearman correlation $0.55$), their median absolute estimates are
$10.6$ and $0.022$, a factor of around $490$ apart. One unit of $\cfi_{\unit}$
is an increase of $1/2$ in $z^j$, around $220$ times the median standard
deviation of a species' abundance, so this may not be a natural scale to work
with for this data. The other two measures are on scales of their own, with median
absolute estimates of $0.13$ for $\cke$, whose unit is the whole contrast
between presence and absence, and $0.025$ for the log-contrast coefficients,
which like $\cfi_{\mult}$ are per unit change in log-odds. Only those two share
a scale, and they nonetheless disagree over which species matter.

Section~\ref{sec:additional_microbiome} gives additional plots comparing the
significant hits of $\cfi_{\mult}$, the log-contrast estimates and a naive
marginal analysis, and illustrating how $\cfi_{\unit}$ and $\cke$ differ from
$\cfi_{\mult}$.

\begin{figure}[!ht]
  \centering
  \includegraphics*[width=\textwidth]{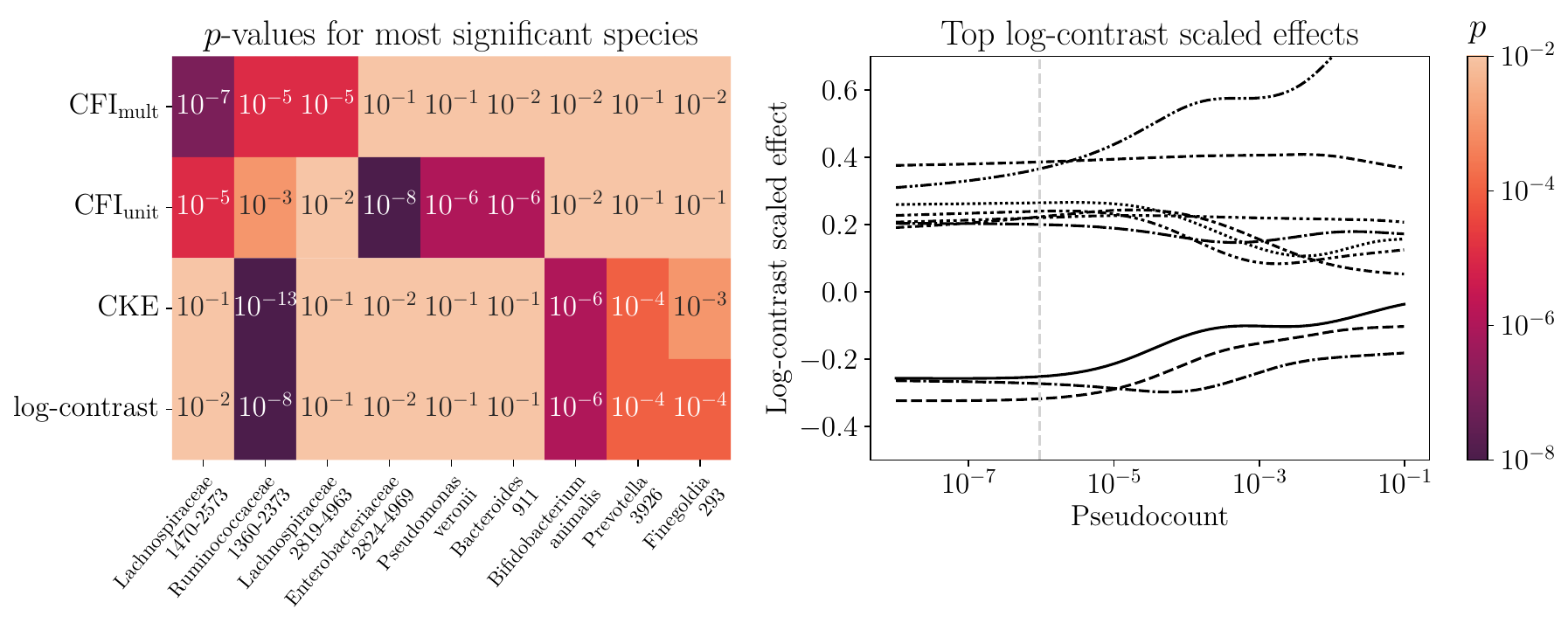}
  \caption{The left plot shows $p$-values from each method for its three most
  significant species (text color only changes to enhance readability). The
  log-contrast model mostly selects the same species as the $\cke$, which
  explicitly captures effects of zeros. The right plot shows the log-contrast
  estimates, scaled by the standard deviation of $\log(Z^j +
  \text{pseudocount})$, for the species listed on the right as a function of the
  added pseudocount; the grey dashed line is the pseudocount used in the
  remaining experiments. The relative influences of the species are highly
  contingent on the pseudocount.}
  \label{fig:microbiome_confidence} 
\end{figure}

\section{Discussion}
We introduced perturbations on the simplex as a formalization of hypothetical
changes to a compositional vector and used them to construct interpretable
compositional effect measures, which we call average perturbation effects. In
our view, perturbations provide a natural way for practitioners to formalize
research questions about non-standard changes to a system. We illustrated this,
for example, by using perturbations to formalize the otherwise vague concept of
an effect of diversity as well as to construct nonparametric feature influence
measures for compositional coordinates.  Unfortunately, many applied researchers
working with compositional data, generally due to a lack of alternatives, still
rely on marginal effects by analyzing dependence between a summary statistic and
a response (e.g., by regressing a response on a diversity measure). As we
demonstrated, such approaches can be misleading since they do not account for
the confounding induced by the remaining variation in the compositional vector.
We hope that this work can draw attention to such pitfalls and improve the
analysis of such data by providing a set of easily accessible alternative
methods.

We additionally proposed efficient estimators for the average perturbation
effects based on semiparametric theory and provided theoretical guarantees that
enable us to perform valid inference. Most statistical methods for compositional
data analysis to date are based on the parametric log-ratio approach introduced
in the seminal paper \citet{aitchison1982statistical} or other parametric and
transformation-based approaches. Our nonparametric approach goes beyond these
established methods and allows practitioners to apply theoretically justified
procedures to modern compositional data applications. Importantly, our target
parameters and estimators circumvent the zero problem central to any log-ratio
based analysis and naturally apply to high-dimensional data. This makes our
proposed estimators applicable to a range of modern applications that previously
had to rely on ad hoc methods or simply ignore the compositional constraints. An
additional advantage of a perturbation-based analysis is that it lends itself to
a causal rather than observational interpretation, allowing practitioners to
make standard causal inference assumptions to derive rigorous causal
conclusions.

In this work we have focused on data constraints induced by the simplex. Many
other constraints arise, e.g., directional data, general manifold-valued data or
otherwise constrained data. The idea of defining effect measures via
perturbations easily extends to those settings. Similarly, perturbations can
also be used to define meaningful effects of summary statistics on unconstrained
data, which has received surprisingly little attention but relates to problems
in causal inference (e.g., how to define causal variables). Furthermore, there
may be settings where the perturbation of interest depends on additional
covariates which then changes the resulting target of inference. For all of
these extensions, more work is required to develop reparametrizations yielding
efficient estimators akin to ours. We hope our work will spark new developments
in this area and help provide practitioners with the statistical tools they
need.

\section*{Acknowledgements}
We thank Alexander Mangulad Christgau and Lucas Kook for helpful discussions. We
also thank Viktor Skjold Lindegaard for spotting some minor errors in an earlier
version of the manuscript. ARL and NP were supported by a research grant
(0069071) from Novo Nordisk Fonden. Claude Opus 5 was used for language editing
and coding assistance.

\bibliographystyle{abbrvnat}
\bibliography{bibliography}

@ARTICLE{marcos2021applications,
AUTHOR={Marcos-Zambrano, Laura Judith and Karaduzovic-Hadziabdic, Kanita and Loncar Turukalo, Tatjana and Przymus, Piotr and Trajkovik, Vladimir and Aasmets, Oliver and Berland, Magali and Gruca, Aleksandra and Hasic, Jasminka and Hron, Karel and Klammsteiner, Thomas and Kolev, Mikhail and Lahti, Leo and Lopes, Marta B. and Moreno, Victor and Naskinova, Irina and Org, Elin and Paciência, Inês and Papoutsoglou, Georgios and Shigdel, Rajesh and Stres, Blaz and Vilne, Baiba and Yousef, Malik and Zdravevski, Eftim and Tsamardinos, Ioannis and Carrillo de Santa Pau, Enrique and Claesson, Marcus J. and Moreno-Indias, Isabel and Truu, Jaak},
TITLE={Applications of Machine Learning in Human Microbiome Studies: A Review on Feature Selection, Biomarker Identification, Disease Prediction and Treatment},
JOURNAL={Frontiers in Microbiology},
VOLUME={12},
YEAR={2021},
DOI={10.3389/fmicb.2021.634511},      	
ISSN={1664-302X},
}

@article{meinshausen2006quantile,
  title={Quantile regression forests.},
  author={Meinshausen, Nicolai},
  journal={Journal of Machine Learning Research},
  volume={7},
  number={6},
  year={2006},
  publisher={JMLR.org}
}

@article{athey2019,
author = {Susan Athey and Julie Tibshirani and Stefan Wager},
title = {{Generalized random forests}},
volume = {47},
journal = {The Annals of Statistics},
number = {2},
publisher = {Institute of Mathematical Statistics},
pages = {1148--1178},
year = {2019},
doi = {10.1214/18-AOS1709}
}

@article{de2013derivative,
  title={Derivative estimation with local polynomial fitting},
  author={De Brabanter, Kris and De Brabanter, Joseph and Gijbels, Irene and De Moor, Bart},
  journal={Journal of Machine Learning Research},
  volume={14},
  number={1},
  pages={281--301},
  year={2013},
  publisher={JMLR.org}
}

@article{masry1997,
author = {Masry, Elias and Fan, Jianqing},
title = {Local Polynomial Estimation of Regression Functions for Mixing Processes},
journal = {Scandinavian Journal of Statistics},
volume = {24},
number = {2},
pages = {165-179},
doi = {10.1111/1467-9469.00056},
year = {1997}
}

@article{li2022reproducing,
      title={Reproducing Kernels and New Approaches in Compositional Data Analysis}, 
      author={Binglin Li and Jeongyoun Ahn},
      year={2022},
      journal={arXiv preprint},
      eprint={2205.01158},
      archivePrefix={arXiv},
      primaryClass={stat.ML},
      doi = {10.48550/arXiv.2205.01158}
}

@article{kurisu2024geodesic,
      title={Geodesic Causal Inference}, 
      author={Daisuke Kurisu and Yidong Zhou and Taisuke Otsu and Hans-Georg Müller},
      year={2024},
      journal={arXiv preprint},
      eprint={2406.19604},
      archivePrefix={arXiv},
      doi = {10.48550/arXiv.2406.19604},
      primaryClass={stat.ME},
}

@misc{bhaduri2025compositional,
      title={Compositional Covariate Importance Testing via Partial Conjunction of Bivariate Hypotheses}, 
      author={Ritwik Bhaduri and Siyuan Ma and Lucas Janson},
      year={2025},
      journal={arXiv preprint},
      eprint={2501.00566},
      archivePrefix={arXiv},
      primaryClass={stat.ME},
      doi = {10.48550/arXiv.2501.00566},
}

@article{greenacre2009,
title = {Power transformations in correspondence analysis},
journal = {Computational Statistics and Data Analysis},
volume = {53},
number = {8},
pages = {3107-3116},
year = {2009},
issn = {0167-9473},
doi = {10.1016/j.csda.2008.09.001},
author = {Michael Greenacre},
}

@article{greenacre2023aitchison,
  title={Aitchison's compositional data analysis 40 years on: {A} reappraisal},
  author={Greenacre, Michael and Grunsky, Eric and Bacon-Shone, John and Erb, Ionas and Quinn, Thomas},
  journal={Statistical Science},
  volume={1},
  number={1},
  pages={1--25},
  year={2023},
  publisher={Institute of Mathematical Statistics},
  doi={10.1214/22-STS880}
}

@article{scealy2011,
    author = {Scealy, J. L. and Welsh, A. H.},
    title = {Regression for Compositional Data by using Distributions Defined on the Hypersphere},
    journal = {Journal of the Royal Statistical Society Series B: Statistical Methodology},
    volume = {73},
    number = {3},
    pages = {351-375},
    year = {2011},
    issn = {1369-7412},
    doi = {10.1111/j.1467-9868.2010.00766.x},
}

@article{scealy2014colours,
author = {Scealy, J. L. and Welsh, A. H.},
title = {Colours and Cocktails: Compositional Data Analysis 2013 Lancaster Lecture},
journal = {Australian \& New Zealand Journal of Statistics},
volume = {56},
number = {2},
pages = {145-169},
doi = {https://doi.org/10.1111/anzs.12073},
year = {2014}
}

@article{li2023its,
    author = {Li, Gen and Li, Yan and Chen, Kun},
    title = {It's All Relative: Regression Analysis with Compositional Predictors},
    journal = {Biometrics},
    volume = {79},
    number = {2},
    pages = {1318-1329},
    year = {2023},
    month = {06},
    doi = {10.1111/biom.13703},
}

@book{pearl2009causality,
  title={Causality},
  author={Pearl, Judea},
  year={2009},
  publisher={Cambridge university press}
}

@article{bongers2021foundations,
  title={Foundations of structural causal models with cycles and latent variables},
  author={Bongers, Stephan and Forr{\'e}, Patrick and Peters, Jonas and Mooij, Joris M},
  journal={The Annals of Statistics},
  volume={49},
  number={5},
  pages={2885--2915},
  year={2021},
  publisher={Institute of Mathematical Statistics},
  doi={10.1214/21-AOS2064}
}

@article{rubin2005causal,
  title={Causal inference using potential outcomes: Design, modeling, decisions},
  author={Rubin, Donald B},
  journal={Journal of the American Statistical Association},
  volume={100},
  number={469},
  pages={322--331},
  year={2005},
  publisher={Taylor \& Francis}
}

@book{aitchisonbook,
  title={The Statistical Analysis of Compositional Data},
  author={Aitchison, J.},
  series={Monographs on Statistics and Applied Probability},
  year={1986},
  publisher={Chapman and Hall}
}

@article{pearson1897mathematical,
  title={Mathematical contributions to the theory of evolution.-- {O}n a form of spurious correlation which may arise when indices are used in the measurement of organs},
  author={Pearson, Karl},
  journal={Proceedings of the Royal Society of London},
  volume={60},
  number={359-367},
  pages={489--498},
  year={1897},
  publisher={The Royal Society London}
}

@article{aitchison1982statistical,
  title={The statistical analysis of compositional data},
  author={Aitchison, John},
  journal={Journal of the Royal Statistical Society: Series B (Methodological)},
  volume={44},
  number={2},
  pages={139--160},
  year={1982},
  publisher={Wiley Online Library}
}

@article{aitchison1984log,
  title={Log contrast models for experiments with mixtures},
  author={Aitchison, John and Bacon-Shone, John},
  journal={Biometrika},
  volume={71},
  number={2},
  pages={323--330},
  year={1984},
  publisher={Oxford University Press},
  doi={10.1093/biomet/71.2.323}
}

@article{egozcue2006simplicial,
  title={Simplicial geometry for compositional data},
  author={Egozcue, Juansnm Jos{\'e} and Pawlowsky-Glahn, Vera},
  journal={Geological Society, London, Special Publications},
  volume={264},
  number={1},
  pages={145--159},
  year={2006},
  publisher={The Geological Society of London},
  doi={10.1144/GSL.SP.2006.264.01.11}
}

@article{antonio2004effects,
  title={Effects of racial diversity on complex thinking in college students},
  author={Antonio, Anthony Lising and Chang, Mitchell J and Hakuta, Kenji and Kenny, David A and Levin, Shana and Milem, Jeffrey F},
  journal={Psychological Science},
  volume={15},
  number={8},
  pages={507--510},
  year={2004},
  publisher={Sage Publications Sage CA: Los Angeles, CA},
  doi={10.1111/j.0956-7976.2004.00710.x}
}

@article{shin2012income,
  title={Income inequality and economic growth},
  author={Shin, Inyong},
  journal={Economic Modelling},
  volume={29},
  number={5},
  pages={2049--2057},
  year={2012},
  publisher={Elsevier}
}

@article{egozcue2003isometric,
  title={Isometric logratio transformations for compositional data analysis},
  author={Egozcue, Juan Jos{\'e} and Pawlowsky-Glahn, Vera and Mateu-Figueras, Gl{\`o}ria and Barcelo-Vidal, Carles},
  journal={Mathematical Geology},
  volume={35},
  number={3},
  pages={279--300},
  year={2003},
  publisher={Springer},
  doi={10.1023/A:1023818214614}
}

@incollection{martin2000zero,
  title={Zero replacement in compositional data sets},
  author={Mart{\'\i}n-Fern{\'a}ndez, JA and Barcel{\'o}-Vidal, Carles and Pawlowsky-Glahn, Vera},
  booktitle={Data analysis, classification, and related methods},
  pages={155--160},
  year={2000},
  publisher={Springer}
}

@article{lubbe2021comparison,
  title={Comparison of zero replacement strategies for compositional data with large numbers of zeros},
  author={Lubbe, Sugnet and Filzmoser, Peter and Templ, Matthias},
  journal={Chemometrics and Intelligent Laboratory Systems},
  volume={210},
  pages={104248},
  year={2021},
  publisher={Elsevier},
  doi={10.1016/j.chemolab.2021.104248}
}

@article{palarea2008modified,
  title={A modified EM alr-algorithm for replacing rounded zeros in compositional data sets},
  author={Palarea-Albaladejo, Javier and Mart{\'\i}n-Fern{\'a}ndez, Josep-Antoni},
  journal={Computers \& Geosciences},
  volume={34},
  number={8},
  pages={902--917},
  year={2008},
  publisher={Elsevier}
}

@article{martin2015bayesian,
  title={Bayesian-multiplicative treatment of count zeros in compositional data sets},
  author={Mart{\'\i}n-Fern{\'a}ndez, Josep-Antoni and Hron, Karel and Templ, Matthias and Filzmoser, Peter and Palarea-Albaladejo, Javier},
  journal={Statistical Modelling},
  volume={15},
  number={2},
  pages={134--158},
  year={2015},
  publisher={Sage Publications Sage India: New Delhi, India},
  doi={10.1177/1471082X14535524}
}

@article{martin2003dealing,
  title     = {Dealing with zeros and missing values in compositional data sets using nonparametric imputation},
  author    = {Mart{\'\i}n-Fern{\'a}ndez, Josep A and Barcel{\'o}-Vidal, Carles and Pawlowsky-Glahn, Vera},
  journal   = {Mathematical Geology},
  volume    = {35},
  number    = {3},
  pages     = {253--278},
  year      = {2003},
  publisher = {Springer}
}

@article{huang2022supervised,
    doi = {10.1371/journal.pcbi.1011240},
    author = {Huang, Shimeng AND Ailer, Elisabeth AND Kilbertus, Niki AND Pfister, Niklas},
    journal = {PLOS Computational Biology},
    publisher = {Public Library of Science},
    title = {Supervised learning and model analysis with compositional data},
    year = {2023},
    volume = {19},
    pages = {1-19},
    number = {6},
}

@book{lee2012introduction,
	doi = {10.1007/978-1-4419-9982-5},
	year = 2012,
	publisher = {Springer New York},
	author = {John M. Lee},
	title = {Introduction to Smooth Manifolds}
}

@article{newey1993efficiency,
 ISSN = {00129682, 14680262},
 author = {Whitney K. Newey and Thomas M. Stoker},
 journal = {Econometrica},
 number = {5},
 pages = {1199--1223},
 publisher = {[Wiley, Econometric Society]},
 title = {Efficiency of Weighted Average Derivative Estimators and Index Models},
 volume = {61},
 year = {1993}
}

@article{klyne2023average,
      title={Average partial effect estimation using double machine learning}, 
      author={Harvey Klyne and Rajen D. Shah},
      year={2023},
      eprint={2308.09207},
      archivePrefix={arXiv},
      primaryClass={math.ST},
      doi={10.48550/arXiv.2308.09207},
      journal={arXiv preprint}
}

@article{shah2020hardness,
author = {Rajen D. Shah and Jonas Peters},
title = {{The hardness of conditional independence testing and the generalised covariance measure}},
volume = {48},
journal = {The Annals of Statistics},
number = {3},
publisher = {Institute of Mathematical Statistics},
pages = {1514 -- 1538},
year = {2020},
doi = {10.1214/19-AOS1857},
}

@article{chernozhukov2018double,
    author = {Chernozhukov, Victor and Chetverikov, Denis and Demirer, Mert and Duflo, Esther and Hansen, Christian and Newey, Whitney and Robins, James},
    title = "{Double/debiased machine learning for treatment and structural parameters}",
    journal = {The Econometrics Journal},
    volume = {21},
    number = {1},
    pages = {C1-C68},
    year = {2018},
    issn = {1368-4221},
    doi = {10.1111/ectj.12097},
}

@Book{schilling2017measures,
 author = {Schilling, Ren\'e L.},
 title = {Measures, Integrals and Martingales},
 publisher = {Cambridge University Press},
 year = {2017},
 address = {Cambridge, United Kingdom New York, NY},
 isbn = {9781316620243},
 edition ={2}
 }

@book{rudin1986real,
  author = {Rudin, Walter},
  publisher = {McGraw-Hill Science/Engineering/Math},
  title = {Real and Complex Analysis},
  year = 1986
}

@article{kennedy2023semiparametric,
      title={Semiparametric doubly robust targeted double machine learning: a review}, 
      author={Edward H. Kennedy},
      year={2023},
      eprint={2203.06469},
      archivePrefix={arXiv},
      primaryClass={stat.ME},
      journal={arXiv preprint},
      doi={10.48550/arXiv.2203.06469}
}

@article{cox1985penalty,
author = {Cox, Dennis},
year = {1985},
pages = {271-288},
title = {A penalty method for nonparametric estimation of the logarithmic derivative of a density function},
volume = {37},
journal = {Annals of the Institute of Statistical Mathematics},
doi = {10.1007/BF02481097}
}

@article{ng1994smoothing,
author = {Ng, Pin T.},
title = {Smoothing Spline Score Estimation},
journal = {SIAM Journal on Scientific Computing},
volume = {15},
number = {5},
pages = {1003-1025},
year = {1994},
doi = {10.1137/0915061},
}

@article{robins1994estimation,
author = {James M.   Robins  and  Andrea   Rotnitzky  and  Lue   Ping   Zhao },
title = {Estimation of Regression Coefficients When Some Regressors are not Always Observed},
journal = {Journal of the American Statistical Association},
volume = {89},
number = {427},
pages = {846-866},
year  = {1994},
publisher = {Taylor & Francis},
doi = {10.1080/01621459.1994.10476818},
}

@article{chernozhukov2022debiased,
    author = {Chernozhukov, Victor and Newey, Whitney K and Singh, Rahul},
    title = "{Debiased machine learning of global and local parameters using regularized Riesz representers}",
    journal = {The Econometrics Journal},
    volume = {25},
    number = {3},
    pages = {576-601},
    year = {2022},
    issn = {1368-4221},
    doi = {10.1093/ectj/utac002},
}

@article{robinson1988root,
 ISSN = {00129682, 14680262},
 author = {P. M. Robinson},
 journal = {Econometrica},
 number = {4},
 pages = {931--954},
 publisher = {[Wiley, Econometric Society]},
 title = {Root-N-Consistent Semiparametric Regression},
 volume = {56},
 year = {1988}
}

@book{pfanzagl1982contributions,
	doi = {10.1007/978-1-4612-5769-1},
	year = 1982,
	publisher = {Springer New York},
	author = {J. Pfanzagl},
	title = {Contributions to a General Asymptotic Statistical Theory}
}

@book{bickel1993efficient,
  year={1993},
  publisher={Johns Hopkins University Press},
  author= {P. J. Bickel and C. A. Klaassen and Y. Ritov and J. A. Wellner},
  title={Efficient and Adaptive Estimation for Semiparametric Models}
}

@book{tsiatis2006semiparametric,
  year=2006,
  publisher={Springer},
  title={Semiparametric Theory and Missing Data},
  author={Anastasios A. Tsiatis}
}

@misc{adult_dataset,
  author       = {Becker,Barry and Kohavi,Ronny},
  title        = {Adult},
  year         = {1996},
  howpublished = {UCI Machine Learning Repository},
  doi = {10.24432/C5XW20}
}

@article{lu2015effect,
title = {Effect of diversity on human resource management and organizational performance},
journal = {Journal of Business Research},
volume = {68},
number = {4},
pages = {857-861},
year = {2015},
issn = {0148-2963},
doi = {10.1016/j.jbusres.2014.11.041},
author = {Chia-Mei Lu and Shyh-Jer Chen and Pei-Chi Huang and Jui-Ching Chien},
}

@article{richard2000racial,
author = {Richard, Orlando C.},
title = {Racial Diversity, Business Strategy, and Firm Performance: A Resource-Based View},
journal = {Academy of Management Journal},
volume = {43},
number = {2},
pages = {164-177},
year = {2000},
doi = {10.5465/1556374},
}

@article{paarlberg2018heterogeneity,
author = {Paarlberg, Laurie E. and Hoyman, Michele and McCall, Jamie},
title = {Heterogeneity, Income Inequality, and Social Capital: A New Perspective},
journal = {Social Science Quarterly},
volume = {99},
number = {2},
pages = {699-710},
doi = {10.1111/ssqu.12454},
year = {2018}
}

@article{richard2007impact,
author = {Richard, Orlando C. and Murthi, B. P. S and Ismail, Kiran},
title = {The impact of racial diversity on intermediate and long-term performance: The moderating role of environmental context},
journal = {Strategic Management Journal},
volume = {28},
number = {12},
pages = {1213-1233},
doi = {10.1002/smj.633},
year = {2007}
}

@article{mcdonald2018american,
author = {Daniel McDonald and Embriette Hyde and Justine W. Debelius and James T. Morton and Antonio Gonzalez and Gail Ackermann and Alexander A. Aksenov and Bahar Behsaz and Caitriona Brennan and Yingfeng Chen and Lindsay DeRight Goldasich and Pieter C. Dorrestein and Robert R. Dunn and Ashkaan K. Fahimipour and James Gaffney and Jack A. Gilbert and Grant Gogul and Jessica L. Green and Philip Hugenholtz and Greg Humphrey and Curtis Huttenhower and Matthew A. Jackson and Stefan Janssen and Dilip V. Jeste and Lingjing Jiang and Scott T. Kelley and Dan Knights and Tomasz Kosciolek and Joshua Ladau and Jeff Leach and Clarisse Marotz and Dmitry Meleshko and Alexey V. Melnik and Jessica L. Metcalf and Hosein Mohimani and Emmanuel Montassier and Jose Navas-Molina and Tanya T. Nguyen and Shyamal Peddada and Pavel Pevzner and Katherine S. Pollard and Gholamali Rahnavard and Adam Robbins-Pianka and Naseer Sangwan and Joshua Shorenstein and Larry Smarr and Se Jin Song and Timothy Spector and Austin D. Swafford and Varykina G. Thackray and Luke R. Thompson and Anupriya Tripathi and Yoshiki Vázquez-Baeza and Alison Vrbanac and Paul Wischmeyer and Elaine Wolfe and Qiyun Zhu and Rob Knight},
title = {American Gut: an Open Platform for Citizen Science Microbiome Research},
journal = {mSystems},
volume = {3},
number = {3},
pages = {10.1128/msystems.00031-18},
year = {2018},
doi = {10.1128/msystems.00031-18}
}

@article{bien2021tree,
  title={Tree-aggregated predictive modeling of microbiome data},
  volume={11},
  ISSN={2045-2322},
  DOI={10.1038/s41598-021-93645-3},
  number={1},
  journal={Scientific Reports},
  publisher={Springer Science and Business Media LLC},
  author={Bien, Jacob and Yan, Xiaohan and Simpson, Léo and Müller, Christian L.},
  year={2021},
}

@article{williamson2021nonparametric,
author = {Williamson, Brian D. and Gilbert, Peter B. and Carone, Marco and Simon, Noah},
title = {Nonparametric variable importance assessment using machine learning techniques},
journal = {Biometrics},
volume = {77},
number = {1},
pages = {9-22},
doi = {10.1111/biom.13392},
year = {2021}
}

@article{breiman2001random,
 volume={45},
  title={Random Forests},
  ISSN={0885-6125},
  DOI={10.1023/a:1010933404324},
  number={1},
  journal={Machine Learning},
  publisher={Springer Science and Business Media LLC},
  author={Breiman, Leo}, 
  year={2001}, 
  pages={5-32} 
}

@article{hu2021rarefaction,
    author = {Hu, Yi-Juan and Lane, Andrea and Satten, Glen A},
    title = "{A rarefaction-based extension of the LDM for testing presence-absence associations in the microbiome}",
    journal = {Bioinformatics},
    volume = {37},
    number = {12},
    pages = {1652-1657},
    year = {2021},
    issn = {1367-4803},
    doi = {10.1093/bioinformatics/btab012},
}

@article{scikit-learn,
 title={Scikit-learn: Machine Learning in {P}ython},
 author={Pedregosa, F. and Varoquaux, G. and Gramfort, A. and Michel, V.
         and Thirion, B. and Grisel, O. and Blondel, M. and Prettenhofer, P.
         and Weiss, R. and Dubourg, V. and Vanderplas, J. and Passos, A. and
         Cournapeau, D. and Brucher, M. and Perrot, M. and Duchesnay, E.},
 journal={Journal of Machine Learning Research},
 volume={12},
 pages={2825--2830},
 year={2011}
}

@article{simpson2021classo,
  doi       = {10.21105/joss.02844},
  year      = {2021},
  publisher = {The Open Journal},
  volume    = {6},
  number    = {57},
  pages     = {2844},
  author    = {Léo Simpson and Patrick L. Combettes and Christian L. Müller},
  title     = {c-lasso - a Python package for constrained sparse and robust regression and classification},
  journal   = {Journal of Open Source Software}
}

@article{cox1971note,
    author = {Cox, D. R.},
    title = "{A note on polynomial response functions for mixtures}",
    journal = {Biometrika},
    volume = {58},
    number = {1},
    pages = {155-159},
    year = {1971},
    month = {04},
    doi = {10.1093/biomet/58.1.155},
}

@book{cornell2002experiments,
 author = {Cornell, John A.},
 doi = {10.1002/9781118204221},
 isbn = {9781118204221},
 issn = {1940-6347},
 journal = {Wiley Series in Probability and Statistics},
 month = {January},
 publisher = {Wiley},
 title = {Experiments with Mixtures: Designs, Models, and the Analysis of Mixture Data},
 year = {2002}
}

@article{claringbold1955use,
 ISSN = {0006341X, 15410420},
 author = {P. J. Claringbold},
 journal = {Biometrics},
 number = {2},
 pages = {174--185},
 publisher = {[Wiley, International Biometric Society]},
 title = {Use of the Simplex Design in the Study of Joint Action of Related Hormones},
 volume = {11},
 year = {1955}
}

@article{scheffe1958experiments,
    author = {Scheffé, Henry},
    title = "{Experiments with Mixtures}",
    journal = {Journal of the Royal Statistical Society: Series B (Methodological)},
    volume = {20},
    number = {2},
    pages = {344-360},
    year = {1958},
    month = {7},
    issn = {0035-9246},
    doi = {10.1111/j.2517-6161.1958.tb00299.x}
}

@article{becker1968models,
    author = {Becker, N. G.},
    title = "{Models for the Response of a Mixture}",
    journal = {Journal of the Royal Statistical Society: Series B (Methodological)},
    volume = {30},
    number = {2},
    pages = {349-358},
    year = {1968},
    month = {7},
    issn = {0035-9246},
    doi = {10.1111/j.2517-6161.1968.tb00735.x},
}

@article{draper1977mixtures,
 ISSN = {00401706},
 author = {Norman R. Draper and Ralph C. St. John},
 journal = {Technometrics},
 number = {1},
 pages = {37--46},
 publisher = {[Taylor & Francis, Ltd., American Statistical Association, American Society for Quality]},
 title = {A Mixtures Model with Inverse Terms},
 urldate = {2024-05-15},
 volume = {19},
 year = {1977}
}

@article{vansteelandt2022assumption,
    author = {Vansteelandt, Stijn and Dukes, Oliver},
    title = "{Assumption-lean Inference for Generalised Linear Model Parameters}",
    journal = {Journal of the Royal Statistical Society Series B: Statistical Methodology},
    volume = {84},
    number = {3},
    pages = {657-685},
    year = {2022},
    month = {07},
    issn = {1369-7412},
    doi = {10.1111/rssb.12504}
}

@article{berk2021assumption,
  author    = {Richard Berk and Andreas Buja and Lawrence Brown and Edward George and Arun Kumar Kuchibhotla and Weijie Su and Linda Zhao},
  title     = {Assumption Lean Regression},
  journal   = {The American Statistician},
  volume    = {75},
  number    = {1},
  pages     = {76--84},
  year      = {2021},
  publisher = {Taylor \& Francis},
  doi       = {10.1080/00031305.2019.1592781},
}

@article{knight2018best,
 author = {Knight, Rob and Vrbanac, Alison and Taylor, Bryn C. and Aksenov, Alexander and Callewaert, Chris and Debelius, Justine and Gonzalez, Antonio and Kosciolek, Tomasz and McCall, Laura-Isobel and McDonald, Daniel and Melnik, Alexey V. and Morton, James T. and Navas, Jose and Quinn, Robert A. and Sanders, Jon G. and Swafford, Austin D. and Thompson, Luke R. and Tripathi, Anupriya and Xu, Zhenjiang Z. and Zaneveld, Jesse R. and Zhu, Qiyun and Caporaso, J. Gregory and Dorrestein, Pieter C.},
 doi = {10.1038/s41579-018-0029-9},
 issn = {1740-1534},
 journal = {Nature Reviews Microbiology},
 month = {May},
 number = {7},
 pages = {410-422},
 publisher = {Springer Science and Business Media LLC},
 title = {Best practices for analysing microbiomes},
 volume = {16},
 year = {2018}
}

@article{cammarota2020gut,
 author = {Cammarota, Giovanni and Ianiro, Gianluca and Ahern, Anna and Carbone, Carmine and Temko, Andriy and Claesson, Marcus J. and Gasbarrini, Antonio and Tortora, Giampaolo},
 doi = {10.1038/s41575-020-0327-3},
 issn = {1759-5053},
 journal = {Nature Reviews Gastroenterology \& Hepatology},
 month = {July},
 number = {10},
 pages = {635-648},
 publisher = {Springer Science and Business Media LLC},
 title = {Gut microbiome, big data and machine learning to promote precision medicine for cancer},
 volume = {17},
 year = {2020}
}

@misc{passnycc,
author = {PASSNYC},
title = {Data Science for Good: {PASSNYC}},
year={2018},
url={https://www.kaggle.com/datasets/passnyc/data-science-for-good}
}

@InProceedings{beckers2023causal,
  title = 	 {Causal Models with Constraints},
  author =       {Beckers, Sander and Halpern, Joseph and Hitchcock, Christopher},
  booktitle = 	 {Proceedings of the Second Conference on Causal Learning and Reasoning},
  pages = 	 {866--879},
  year = 	 {2023},
  editor = 	 {van der Schaar, Mihaela and Zhang, Cheng and Janzing, Dominik},
  volume = 	 {213},
  series = 	 {Proceedings of Machine Learning Research},
  month = 	 {11--14 Apr},
  publisher =    {PMLR},
}

@article{spirtes2004causal,
  title   = {Causal Inference of Ambiguous Manipulations},
  volume  = {71},
  doi     = {10.1086/425058},
  number  = {5},
  journal = {Philosophy of Science},
  author  = {Spirtes, Peter and Scheines, Richard},
  year    = {2004},
  pages   = {833-845}
}

@article{woodward2015interventionism,
  title={Interventionism and causal exclusion},
  author={Woodward, James},
  journal={Philosophy and Phenomenological Research},
  volume={91},
  number={2},
  pages={303--347},
  year={2015},
  publisher={Wiley Online Library}
}

@article{white1980heteroskedasticity,
author={White, Halbert},
year={1980},
month={May},
title={A HETEROSKEDASTICITY-CONSISTENT COVARIANCE MATRIX ESTIMATOR AND A DIRECT TEST FOR HETEROSKED ASTICITY},
journal={Econometrica},
volume={48},
number={4},
pages={817},
}

@article{aronow2016does,
  title={Does regression produce representative estimates of causal effects?},
  author={Aronow, Peter M and Samii, Cyrus},
  journal={American Journal of Political Science},
  volume={60},
  number={1},
  pages={250--267},
  year={2016},
  publisher={Wiley Online Library}
}

@article{graham2022semiparametrically,
title = {Semiparametrically efficient estimation of the average linear regression function},
journal = {Journal of Econometrics},
volume = {226},
number = {1},
pages = {115-138},
year = {2022},
doi = {10.1016/j.jeconom.2021.07.008},
author = {Bryan S. Graham and Cristine Campos de Xavier Pinto}
}

@article{sloczynski2022interpreting,
  title={Interpreting OLS estimands when treatment effects are heterogeneous: smaller groups get larger weights},
  author={S{\l}oczy{\'n}ski, Tymon},
  journal={Review of Economics and Statistics},
  volume={104},
  number={3},
  pages={501--509},
  year={2022},
  publisher={MIT Press}
}

@book{marron2021object,
 author = {Marron, J.S. and Dryden, Ian L.},
 doi = {10.1201/9781351189675},
 month = {October},
 publisher = {Chapman and Hall/CRC},
 title = {Object Oriented Data Analysis},
 year = {2021}
}

@article{lin2014variable,
    author = {Lin, Wei and Shi, Pixu and Feng, Rui and Li, Hongzhe},
    title = "{Variable selection in regression with compositional covariates}",
    journal = {Biometrika},
    volume = {101},
    number = {4},
    pages = {785-797},
    year = {2014},
    month = {08},
    issn = {0006-3444},
    doi = {10.1093/biomet/asu031}
}

@article{fiksel2022transformation,
author = {Fiksel, Jacob and Zeger, Scott and Datta, Abhirup},
title = {A transformation-free linear regression for compositional outcomes and predictors},
journal = {Biometrics},
volume = {78},
number = {3},
pages = {974-987},
doi = {10.1111/biom.13465},
year = {2022}
}

@article{dai2018principal,
author = {Xiongtao Dai and Hans-Georg M{\"u}ller},
title = {{Principal component analysis for functional data on Riemannian manifolds and spheres}},
volume = {46},
journal = {The Annals of Statistics},
number = {6B},
publisher = {Institute of Mathematical Statistics},
pages = {3334 -- 3361},
year = {2018},
doi = {10.1214/17-AOS1660}
}

@article{dai2021modeling,
author = {Dai, Xiongtao and Lin, Zhenhua and Müller, Hans-Georg},
title = {Modeling sparse longitudinal data on Riemannian manifolds},
journal = {Biometrics},
volume = {77},
number = {4},
pages = {1328-1341},
doi = {10.1111/biom.13385},
year = {2021}
}

@article{petersen2022modeling,
title = {Modeling Probability Density Functions as Data Objects},
journal = {Econometrics and Statistics},
volume = {21},
pages = {159-178},
year = {2022},
doi = {https://doi.org/10.1016/j.ecosta.2021.04.004},
author = {Alexander Petersen and Chao Zhang and Piotr Kokoszka},
}

@book{vanderlaan2011targeted,
 author = {van der Laan, Mark J. and Rose, Sherri},
 doi = {10.1007/978-1-4419-9782-1},
 journal = {Springer Series in Statistics},
 publisher = {Springer New York},
 title = {Targeted Learning: Causal Inference for Observational and Experimental Data},
 year = {2011}
}

@article{xiong2015virus,
author = {Jie Xiong and D. P. Dittmer and J. S. Marron},
title = {{``Virus hunting'' using radial distance weighted discrimination}},
volume = {9},
journal = {The Annals of Applied Statistics},
number = {4},
publisher = {Institute of Mathematical Statistics},
pages = {2090 -- 2109},
year = {2015},
doi = {10.1214/15-AOAS869},
}

@article{diaz2012population,
    author = {Muñoz, Iván Díaz and van der Laan, Mark},
    title = {Population Intervention Causal Effects Based on Stochastic Interventions},
    journal = {Biometrics},
    volume = {68},
    number = {2},
    pages = {541-549},
    year = {2012},
    month = {06},
    doi = {10.1111/j.1541-0420.2011.01685.x},
}

@article{haneuse2013estimation,
author = {Haneuse, S. and Rotnitzky, A.},
title = {Estimation of the effect of interventions that modify the received treatment},
journal = {Statistics in Medicine},
volume = {32},
number = {30},
pages = {5260-5277},
doi = {10.1002/sim.5907},
year = {2013}
}

@article{kennedy2019nonparametric,
author = {Edward H. Kennedy},
title = {Nonparametric Causal Effects Based on Incremental Propensity Score Interventions},
journal = {Journal of the American Statistical Association},
volume = {114},
number = {526},
pages = {645--656},
year = {2019},
publisher = {Taylor \& Francis},
doi = {10.1080/01621459.2017.1422737},
}

@article{diaz2023nonparametric,
author = {Iván Díaz and Nicholas Williams and Katherine L. Hoffman and Edward J. Schenck},
title = {Nonparametric Causal Effects Based on Longitudinal Modified Treatment Policies},
journal = {Journal of the American Statistical Association},
volume = {118},
number = {542},
pages = {846--857},
year = {2023},
publisher = {Taylor \& Francis},
doi = {10.1080/01621459.2021.1955691},
}

\appendix
\renewcommand{\appendixpagename}{Supplementary material}

\renewcommand{\theequation}{S\arabic{equation}}
\renewcommand{\thesection}{S\arabic{section}}
\renewcommand{\thetheorem}{S\arabic{theorem}}
\renewcommand{\thefigure}{S\arabic{figure}}
\renewcommand{\theassumption}{S\arabic{assumption}}
\renewcommand{\thealgorithm}{S\arabic{algorithm}}
\renewcommand{\thedefinition}{S\arabic{definition}}
\renewcommand{\theproposition}{S\arabic{proposition}}
\renewcommand{\thelemma}{S\arabic{lemma}}
\renewcommand{\theexample}{S\arabic{example}}
\renewcommand{\thetable}{S\arabic{table}}
\renewcommand{\theremark}{S\arabic{remark}}
\setcounter{section}{0}
\setcounter{equation}{0}
\setcounter{theorem}{0}
\setcounter{figure}{0}
\setcounter{algorithm}{0}
\setcounter{lemma}{0}
\setcounter{assumption}{0}
\setcounter{theorem}{0}
\setcounter{definition}{0}
\setcounter{proposition}{0}
\setcounter{example}{0}
\setcounter{table}{0}
\setcounter{remark}{0}

\newpage 
\appendixpage

All references to sections, equations and result in the supplementary material are prefixed with an `S'. All other references point to the main text. The supplementary material is structured as follows:
\begin{itemize}
  \item Section~\ref{sec:diff}: Differentiability on the simplex
  \item Section~\ref{sec:perturbations_supplement}: Perturbations: background, theory and examples
  \item Section~\ref{sec:causal_models}: Perturbations as causal quantities
  \item Section~\ref{sec:extra_est_and_alg}: Inference and algorithms for perturbation effects
  \item Section~\ref{sec:extra_num}: Details and additional results from main numerical experiments
\end{itemize}

\section{Differentiability on the simplex}
\label{sec:diff}
In this section we provide the background needed to define
differentiability of a function defined on the simplex, which we use in
Section~\ref{sec:perturbations}, by means of the theory of smooth
manifolds \citep[see][for a comprehensive treatment of this
theory]{lee2012introduction}. The simplex $\Delta^{d-1}$ can be viewed as a
smooth $(d-1)$-dimensional manifold with corners \citep[][p.\
415]{lee2012introduction}, that is, the simplex locally behaves like a Euclidean
half-space $\mathcal{H}^{d-1}:=\{x \in \mathbb{R}^{d-1} \mid \forall j \in
[d-1]: x^j \geq 0\}$.

To apply the tools of differential geometry, we need to equip the
simplex with a smooth atlas, that is, a set of mappings from subsets
of the simplex to subsets of the half-space satisfying certain
smoothness conditions. In particular, the mappings need to be
homeomorphisms from subsets of $\Delta^{d-1}$ to relatively open
subsets of $\mathcal{H}^{d-1}$ with the additional property that if
$\varphi_1 : \mathcal{Z}_1 \to \mathcal{U}_1$ and
$\varphi_2: \mathcal{Z}_2 \to \mathcal{U}_2$ are both elements of the
atlas, the composition $\varphi_2 \circ \varphi_1^{-1} : \varphi_1(\mathcal{Z}_1 \cap \mathcal{Z}_2) \to \varphi_2(\mathcal{Z}_1 \cap \mathcal{Z}_2)$ is a
diffeomorphism. A
mapping is differentiable at a boundary point if there exists an
extension of the mapping defined in a neighborhood of the boundary
point that is differentiable at this point.

\begin{definition}[the simplex as a smooth manifold]
  \label{def:simplex_manifold}
  Define
  $U := \{x \in \mathcal{H}^{d-1} \mid \sum_{j=1}^{d-1} x^j < 1\}$,
  which is relatively open in $\mathcal{H}^{d-1}$. For each
  $j \in [d]$, define
  $\mathcal{E}_j := \{e_1, \dots, e_{j-1}, e_{j+1}, \dots, e_d\}$,
  $\mathcal{Z}_j := \Delta^{d-1} \setminus \mathcal{E}_j$ and
  $\varphi_j : \mathcal{Z}_j \to U$ by
  $\varphi_j(z) := (z^1, \dots, z^{j-1}, z^{j+1}, \dots, z^d)$. We
  call the set
  $\mathcal{A}:= \{(\mathcal{Z}_j, \varphi_j) \mid j\in[d]\}$ the
  \emph{canonical smooth atlas on $\Delta^{d-1}$} and each element of
  $\mathcal{A}$ is a \emph{canonical chart on $\Delta^{d-1}$}.
\end{definition}

For a canonical chart $(\mathcal{Z}_j, \varphi_j)$ the mapping
  $\varphi_j$ is indeed a homomorphism since both $\varphi_j$ and
  $\varphi_{j}^{-1}$ are continuous.

\begin{definition}[differentiability on the simplex]
  Let $\mathcal{Z} \subseteq \Delta^{d-1}$ and
  $f: \mathcal{Z} \to \mathbb{R}$. We say that $f$ is \emph{differentiable at $z \in \mathcal{Z}$} if there exists a canonical chart $(\mathcal{Z}_j, \varphi_j)$ with $z \in \mathcal{Z}_j$, an open set $O \subseteq \mathbb{R}^{d-1}$ with $\varphi_j(z) \in O$ and a function $g: O \to \mathbb{R}$ such that
  \begin{enumerate}[(a)]
    \item $g$ is differentiable at $\varphi_j(z)$,
    \item for all $u \in U \cap O$, it holds that $g(u) = f(\varphi_j^{-1}(u))$.
  \end{enumerate}
\end{definition}
If $\varphi_j(z)$ is not a boundary point of $U$ in
$\mathbb{R}^{d-1}$, then we can take $g=f \circ \varphi_j^{-1}$. It is
easy to check that if $f: \Delta^{d-1} \to \mathbb{R}$ is
differentiable at $z$ then for any $j \in [d]$ with
$z \in \mathcal{Z}_j$ the function $f \circ \varphi_j^{-1}$ must be
differentiable at $\varphi_j(z)$. We will use these tools to 
prove results about perturbations below.

\section{Perturbations: background, theory and examples}
\label{sec:perturbations_supplement}
In this section we provide additional details on perturbations and perturbation
effects. Section~\ref{sec:related_literature} provides background on the history
of effect measures and compositional data. Section~\ref{sec:toy_examples}
contains two examples to motivate why applying standard methods to compositional
predictors can be misleading. Section~\ref{sec:well-defined_perturbations}
discusses well-definedness of the perturbation effects defined in
Definition~\ref{def:avg_per_eff}. Section~\ref{sec:deriv-iso_proof} contains a
proof of Proposition~\ref{prop:derivative-isolating}.
Section~\ref{sec:derivative-isolating} contains a general result on the
construction of a derivative-isolating reparametrization and examples of the use
of the result. Section~\ref{sec:amalgamations} includes a generalization of the
`individual component' effect measures that we presented in
Section~\ref{sec:perturbations}. Section~\ref{sec:other_geometries} investigates
how perturbations can be derived starting from the Aitchison geometry and how
such perturbations compare to the ones we construct in
Section~\ref{sec:perturbations}. Section~\ref{sec:misc_results} contains two
results relating our average perturbation effects to well-known classical effect
measure.

\subsection{Historical review of effect measures for compositional data}
\label{sec:related_literature} 
In this section we provide historical context for our proposal by
summarizing relevant developments from several different statistical fields
related to our proposal. Our proposed perturbation-based framework is
perhaps most closely related to the literature on causal inference,
semiparametric inference and targeted learning
\citep{bickel1993efficient,tsiatis2006semiparametric,vanderlaan2011targeted}.
\citet{spirtes2004causal} discuss some of the problems of modelling causal
effects when certain variables are deterministic functions of other variables.
This discussion was later continued by \citet{woodward2015interventionism} and
new causal models have been developed to address this problem
\citep{beckers2023causal}. We contribute to these developments by providing a
novel framework for causal effect estimation when some covariates are
compositional (see Section~\ref{sec:causal_models}).

An important aspect of causal effect estimation more generally is the need for
an interpretable summary measure of the effect. Historically, such interpretable
measures originated from parametric models that were easy to interpret. Much
work has since gone into the interpretation of and uncertainty quantification
for such parameters when the model is misspecified. Perhaps one of the most
impactful such analyses was given by \citet{white1980heteroskedasticity} who
provided an estimation procedure for the coefficients in a linear model that did
not require the model to be correctly specified. To what extent parameters based
on linear models are useful and interpretable when the model is misspecified
remains an active research topic particularly in econometrics
\citep{aronow2016does,graham2022semiparametrically,sloczynski2022interpreting}.
More recently, interpretable parameters based on partially linear models have
been proposed and analyzed extensively
\citep{chernozhukov2018double,vansteelandt2022assumption}. Despite these and
similar parameters depending on potentially complicated nuisance functions, it
is often possible to estimate them efficiently under relatively mild conditions,
as we also do here, by utilizing semiparametric theory
\citep{kennedy2023semiparametric}.

Model-based effect estimation with compositional variables has a long history
and was perhaps first thoroughly investigated under the name `experimental
design with mixtures' \citep{claringbold1955use, cornell2002experiments}. While
parts of the literature focused on the problem of selecting design points $z_1,
\dots, z_n \in \Delta^{d-1}$ in optimal ways, a key aspect was also the problem
of constructing reasonable regression models in this setting.
\citet{scheffe1958experiments} pointed out that polynomial models are
overidentified with compositional data due to the sum-to-one constraint and
provided a canonical choice of polynomials. \citet{becker1968models} critiqued
the use of polynomials for compositional features and instead provided
parametric models with homogenous functions of degree one as building blocks.
\citet{cox1971note} argued that the polynomials of Scheffé would often not
result in interpretable coefficients but that one should instead define the
coefficients differently, e.g.\ so that the $j$th coefficient in a linear model
corresponds to increasing the value of $Z^j$ while scaling the remaining
components towards a fixed and problem-dependent `standard mixture' $z_* \in
\Delta^{d-1}$. It is possible to define a model-agnostic version of Cox's linear
effects using our perturbation-based framework (see Proposition~\ref{prop:cox}
in Section~\ref{sec:misc_results} of the supplementary material).
\citet{draper1977mixtures} developed predictive polynomial models based on the
reciprocal of the components of $Z$ to cope with cases where the expected
response diverges at the boundary. \citet{aitchison1984log} introduced the
log-contrast regression model where the linear version corresponds to a linear
regression on the log-transformed $Z$ with a sum-to-zero constraint on the
coefficients. Log-contrast models have seen extensive use in practice and have
since been generalized to high-dimensional $Z$
\citep{lin2014variable,simpson2021classo}. One of our proposed feature influence
measures, the compositional feature influence with multiplicative speed (see
Example~\ref{ex:cfi}), corresponds exactly to the log-contrast coefficients when
the model is correctly specified (see Proposition~\ref{prop:log-contrast} in
Section~\ref{sec:misc_results} of the supplementary material or
\citet{huang2022supervised} where this was first shown). Defining practical
predictive models for compositional data remains an active topic of research
\citep{xiong2015virus,fiksel2022transformation,
li2023its,bhaduri2025compositional}.

Naturally, our proposed average perturbation effects are also connected to the
broader literature on compositional data analysis.  It is well-known that many
conventional statistical procedures become invalid when applied to compositional
data. For example, already \citet{pearson1897mathematical} drew attention to the
fact that correlations between individual compositional components can be
spurious due to the sum-to-one constraint. The
most prevalent class of general methods is based on the seminal paper by
\citet{aitchison1982statistical} which introduces the log-ratio approach. In its
most basic form it consists of transforming the simplex data via the additive,
centered or isometric log-ratio transforms
\citep{aitchison1982statistical,aitchisonbook,egozcue2003isometric}
and then performing the statistical analysis in the transformed space. Our
approach is not based on log-ratios, but it is possible to construct
perturbations from this perspective (see Section~\ref{sec:other_geometries}).
The idea of transforming compositional data via the log-ratio transform to a
simpler sample space has been applied successfully in many applications, however
this in principle requires $Z$ to take values in the open simplex. This is known
as the zero problem and has sparked a large body of research developing various
types of zero-imputation schemes \citep[e.g.,][]{martin2000zero,
martin2003dealing, martin2015bayesian, palarea2008modified, lubbe2021comparison}
that can be used before applying the log-ratio approach.
\citet{greenacre2023aitchison} provide an overview of current state-of-the-art
compositional data analysis and highlights some of the shortcomings in existing
transformation-based approaches which are further highlighted in
\citet{scealy2014colours}. While many of these approaches are sensible in
specific applications, we believe that it is generally preferable to modify the
statistical procedure to fit the data rather than vice versa.

There are other approaches to the analysis of data on the simplex, both in the
form of alternative transformations, for example power-transformations
\citep{greenacre2009} or spherical transformations \citep{scealy2011,
li2022reproducing}, or by appealing to more general notions of object-oriented
data analysis \citep{marron2021object}. There has also been much recent work
on manifold-valued data where compositional data is a special case
\citep{dai2018principal,dai2021modeling,petersen2022modeling,kurisu2024geodesic}.

\subsection{Illustrative examples}
\label{sec:toy_examples}
In this section we provide two helpful examples meant to motivate the need to
account for the compositional nature of covariates. The examples are meant to 
supplement the development in Section~\ref{sec:intro}.

\begin{example}[Effects of presence or absence of microbes]
  \label{ex:absence_microbes}
  Consider the problem of determining whether the
  presence or absence of certain microbes in the human gut is affecting a
  disease. Microbial abundances are commonly measured using genetic sequencing
  (e.g., targeted amplicon or shotgun metagenomic sequencing) which in most
  cases only measures the composition and not the absolute abundances. Assume we
  have measurements from a cohort of patients, where for each patient the
  relative abundance of all microbes in their gut $Z\in\Delta^{d-1}$ and a
  disease indicator $Y\in\{0,1\}$ ($1$ encoding disease) has been
  recorded. While estimating causal effects in such a setting is often too
  ambitious, one could hope to screen for relevant microbes that could then be
  considered for follow-up randomized studies.  One possible approach
  \citep[e.g.,][]{hu2021rarefaction} would be to consider a fixed microbe
  $j$, define the treatment variable $T:= \ind_{\{Z^j = 0\}}$
  corresponding to the presence or absence of microbe $j$ and then estimate a
  marginal treatment effect given by

    \[
      \mathrm{MTE} := \E[Y \given T = 1] - \E[Y \given T = 0] =
      \mathbb{P}(Y = 1 \given T = 1) - \mathbb{P}(Y = 1 \given T = 0),
    \]
  where we implicitly assume $\mathbb{P}(T=0)>0$ and $\mathbb{P}(T=1)>0$. It is
  well-known that this type of marginal dependence analysis can be misleading if
  dependencies exist between the individual components.\footnote{This is true
  when $Z\in \mathbb{R}^d$ is unconstrained, but the
  problem is even more severe when $Z$ is compositional as the sum-to-one
  constraint enforces additional dependency between the components of $Z$.} To
  illustrate this, assume the disease indicator $Y$ given the composition
  $Z$ has conditional distribution for all $z \in \Delta^{d-1}$ given by
  a Bernoulli distribution with 
    \[
      \mathbb{P}(Y = 1 \given Z = z) = \tfrac{3}{4}+\tfrac{1}{8}\ind_{\{z^1=0\}}-\tfrac{1}{2}\ind_{\{z^2 = 0\}}
    \]
  that is, the absence of $Z^1$ increases while the absence of $Z^2$
  reduces the probability of having the disease. In a
  numerical simulation of this generative model with dependence between $Z^1$
  and $Z^2$, the effect of $Z^1$ is instead estimated to decrease the
  probability of disease as shown in Table~\ref{tab:toy_illustration}
  (details are provided in Section~\ref{sec:intro_sim} of the supplementary
  material). If $Z \in \mathbb{R}^d$ was
  unconstrained, we could adjust for the dependence in the remaining
  coordinates by estimating a conditional treatment effect such as 
  \[
    h: z^{-j} \mapsto \E[Y \given T = 1, Z^{-j} = z^{-j}] -  \E[Y \given T = 0, Z^{-j} = z^{-j}]
  \]
  and marginalizing it to produce a single effect measure
  $\E[h(Z^{-j})]$.\footnote{In the literature on causal inference (and under
  appropriate causal assumptions), $h$ corresponds to the conditional average
  treatment effect (CATE).} Estimating this quantity is a well-studied problem,
  e.g.\ using the AIPW estimator \citep{robins1994estimation} or using a
  partially linear model \citep{chernozhukov2018double}. Unfortunately, neither
  of the estimators is well defined when $Z \in \Delta^{d-1}$ as it is not
  possible to change $T$ and $Z^{-j}$ independently of each other and, in fact,
  for all $z^{-j}$, we have $h(z^{-j}) = 0$ as $Y$ is conditionally independent
  of $T$ given $Z^{-j}$ due to the sum-to-one constraint which makes $T$ fully
  determined by $Z^{-j}$. In the Table~\ref{tab:toy_illustration}, we provide
  the results from applying our proposed marginalized conditional treatment
  effect the \emph{compositional knockout effect} to $Z^1$ $(\cke^1)$ which we
  introduce in Section~\ref{sec:simplex_effects}.
\end{example}

\begin{table}[t]
  \begin{minipage}[t]{.5\linewidth}
    \vspace{0pt}
    \centering
    \begin{tabular}{@{}lcc@{}}
      \multicolumn{3}{c}{Example~\ref{ex:absence_microbes}}\\
      \toprule
      method        & estimate & 95\%-CI \\ \midrule
      $\mathrm{MTE}$ & $\textcolor{LKred}{\mathbf{-0.12}}$  & $(-0.17, -0.06)$ \\
      $\cke^1$ using \texttt{rf} & $\textcolor{LKgreen}{\mathbf{0.15}}$ & $(0.10, 0.21)$\\
      \bottomrule
    \end{tabular}
  \end{minipage}%
  \begin{minipage}[t]{.5\linewidth}
    \vspace{0pt}
    \centering
    \begin{tabular}{@{}lcc@{}}
      \multicolumn{3}{c}{Example~\ref{ex:increase_diversity}}\\
      \toprule
      method        & estimate & 95\%-CI\\ \midrule
      OLS of $Y$ on $-G(Z)$    & $\textcolor{LKred}{\mathbf{-1.57}}$         & $(-2.18, -0.97)$              \\
      $\cdi_{\gini}$ using \texttt{rf} & $\textcolor{LKgreen}{\mathbf{0.63}}$ & $(0.14, 1.12)$ \\\bottomrule
    \end{tabular}
  \end{minipage}
  \caption{Numerical illustration of Examples~\ref{ex:absence_microbes} and
    \ref{ex:increase_diversity}. The first row in each table corresponds to the
    marginal treatment effect (MTE or OLS of $Y$ on $-G(Z)$) estimate based on a
    correctly specified linear model. The second row corresponds to estimates of
    two of the target parameters given in Table~\ref{tab:targets}
    with random forest (\texttt{rf}) as nuisance function estimators. Bold red
    numbers indicate a significant negative effect at a $5\%$ level, while bold
    green indicates a significant positive effect. Left: The estimated MTE is
    negative, even though the true effect is positive. The compositional
    knock-out effect ($\cke^1$), which adjust for the simplex structure, results
    in the correct sign. Right: The estimated marginal effect of an increase in
    diversity (measured by the negative Gini coefficient, i.e., $-G(Z)$) is
    negative even though the true effect is equal to $1$ by construction. The
    compositional diversity effect $\cdi_{\gini}$ correctly recovers the effect.
    Details on both experiments are provided in Section~\ref{sec:intro_sim}.}
  \label{tab:toy_illustration}
\end{table}

\begin{example}[Effects of increased diversity]
  \label{ex:increase_diversity}
  Consider the problem of determining the effect of an increase in diversity of
  $Z$ on a response $Y$. Such a question is of immediate interest in many
  applied fields, e.g.\ when investigating whether more racially diverse
  university programs improve educational outcomes \citep{antonio2004effects} or
  whether increased income equality boosts economic output
  \citep{shin2012income}. The prevalent approach is to summarize diversity via a
  one-dimensional summary statistic $G(Z)$ (e.g., the Gini coefficient,
  generalized entropy or Gini-Simpson index) and analyze the correlation between
  $G(Z)$ and $Y$. As in Example~\ref{ex:absence_microbes} this type of marginal
  dependence analysis is potentially misleading in practice. In
  Table~\ref{tab:toy_illustration} we illustrate this on a specific numeric
  example, where the true causal effect is positive, but the marginal method
  leads to a negative effect of diversity (details are provided in
  Section~\ref{sec:intro_sim} of the supplementary material). To avoid this type
  of confounding bias, one would again like to adjust for the remaining
  dependence within $Z$. However, since $G(Z)$ depends on all of $Z$ existing
  adjustment approaches are not applicable. The \emph{compositional diversity
  effect} (with Gini-speed) which we introduce in
  Section~\ref{sec:simplex_effects} appropriately adjusts for the confounding as
  shown in Table~\ref{tab:toy_illustration}. A more in-depth example in which
  the marginal approach is misleading is given in
  Section~\ref{sec:diversity_income} where we estimate the effect of racial
  diversity on income on a semi-synthetic
  dataset based on US census data. A real application is given in
  Section~\ref{sec:american_schools} where the effect of racial diversity on
  school grades is analyzed using data from schools in New York. Unlike in
  Example~\ref{ex:absence_microbes}, the problem of not being able to adjust
  with existing techniques is not solely due to the simplex-constraint and
  similarly appears for summary statistics of unconstrained data. Our
  perturbation-based framework can also be extended to those settings, but as
  the use of this type of analysis is particularly prevalent in compositional
  data applications we restrict our attention to this setting.
\end{example}

\subsection{Well-definedness of perturbation effects}
\label{sec:well-defined_perturbations}
In this section, we argue that the perturbation effects we define in
Definition~\ref{def:avg_per_eff} are well-defined. This is immediate for the
directional effect and similarly it is clear why we need $\mathbb{P}(Z \neq
\psi(Z, 1)) > 0$ for the binary effect. To see why we need that $z \mapsto
\psi(z, 1)$ preserves null sets, we can argue as follows.

If $z \mapsto \psi(z, 1)$ does not preserve null sets, there exists a set $A
\subseteq \mathcal{Z}$ such that $\mathbb{P}(Z \in A) = 0$ but
$\mathbb{P}(\psi(Z, 1) \in A) > 0$. Let $f_1$ be any version of $z \mapsto \E[Y
\given Z =z]$ and let $f_2: z \mapsto f_1(z) + \ind_{A}(z)$. Since $A$ is a null
set for the distribution of $Z$, $f_2(Z)$ is a valid version of the conditional
expectation of $Y$ given $Z$. However, 
\[
  \E[f_1(\psi(Z, 1))] \neq \E[f_1(\psi(Z, 1))] + \mathbb{P}(\psi(Z, 1) \in A) = \E[f_2(\psi(Z, 1))],
\]
so it is necessary for $z \mapsto \psi(z, 1)$ to preserve null sets for the
binary perturbation effect to be well-defined.

On the contrary, if $z \mapsto \psi(z, 1)$ preserves null sets, then for any two
version $f_1$ and $f_2$ of $z \mapsto \E[Y \given Z = z]$ we have that
$\mathbb{P}(f_1(Z) = f_2(Z)) = 1$ and thus also $\mathbb{P}(f_1(\psi(Z, 1)) =
f_2(\psi(Z, 1))) = 1$. This immediately yields that $\E[f_1(\psi(Z, 1))] =
\E[f_2(\psi(Z, 1))]$, so the binary perturbation effect is well-defined.

\subsection{Proof of Proposition~\ref{prop:derivative-isolating}}
\label{sec:deriv-iso_proof}
In this section we prove Proposition~\ref{prop:derivative-isolating} using the
tools from Section~\ref{sec:diff}.
\begin{proof}
It suffices to show that for all $z \in \mathcal{Z}$
\begin{equation}
  \label{eq:deriv-iso-z}
  \partial_\gamma f(\psi(z, \gamma)) \bigm|_{\gamma = 0} = \partial_\ell g(\ell, \phi^W(z)) \bigm|_{\ell = \phi^L(z)},
\end{equation}
or equivalently for all $w \in \mathcal{W}$ and $\ell \in \mathcal{L}_w$
\begin{equation}
  \label{eq:deriv-iso-lw}
  \partial_\gamma f(\psi(\phi^{-1}(\ell, w), \gamma)) \bigm|_{\gamma = 0} = \partial_{\ell'} g(\ell', w) \bigm|_{\ell' = \ell}.
\end{equation}
Fix $z \in \mathcal{Z}$, set $(\ell, w) := \phi(z)$ and define $\psi_z : \gamma
\mapsto \psi(z, \gamma)$. Choose a canonical chart from the canonical smooth
atlas $(\mathcal{Z}_j, \varphi_j) \in \mathcal{A}$ with $z \in \mathcal{Z}_j$
and note that
\[
  f \circ \psi_z = \underbrace{f \circ \varphi_j^{-1}}_{\widetilde{f}} \circ \underbrace{\varphi_j \circ \psi_z}_{\widetilde{\psi}_z},
\]
where $\widetilde{f}: U \to \mathbb{R}$ and $\widetilde{\psi}_z :
\mathcal{I}_\psi(z) \to U$ with $U$ defined in
Definition~\ref{def:simplex_manifold}. We can now apply the chain rule to
compute the left-hand side of \eqref{eq:deriv-iso-z}: 
\[
  \partial_\gamma f(\psi(z, \gamma)) \bigm|_{\gamma = 0} = \nabla \widetilde{f}(z)^\top \partial_{\gamma} \widetilde{\psi}_z(\gamma) \bigm|_{\gamma = 0}.
\]
Similarly, 
\[
  g = f \circ \phi^{-1} =  \underbrace{f \circ \varphi_j^{-1}}_{\widetilde{f}} \circ \underbrace{\varphi_j \circ \phi^{-1}}_{\widetilde{\phi}^{-1}},
\]
where $\widetilde{\phi}^{-1}: \mathrm{Im}(\phi) \to U$. We can again apply the
chain-rule to compute the right-hand side of \eqref{eq:deriv-iso-lw}
\[
  \partial_{\ell'} g(\ell', w) \bigm|_{\ell' = \ell} = \nabla \widetilde{f}(\widetilde{\phi}^{-1}(\ell, w))^\top \partial_{\ell'} \widetilde{\phi}^{-1}(\ell', w) \bigm|_{\ell' = \ell}.
\]
Combining the above, we conclude that \eqref{eq:deriv-iso-lw} is satisfied if
\[
  \partial_{\gamma} \widetilde{\psi}_z(\gamma) \bigm|_{\gamma = 0} = \partial_{\ell'} \widetilde{\phi}^{-1}(\ell', w) \bigm|_{\ell' = \ell}.
\]
This is true by \citet[][Proposition 3.24]{lee2012introduction} under
the assumption that $\phi$ is derivative-isolating, hence proving the
result. We can therefore write
\[
  \tau_\psi = \E[\partial_{\gamma} f(\psi(Z, \gamma)) \bigm|_{\gamma = 0}] = \E[\partial_{\ell} g(\ell, W) \bigm|_{\ell = L}]
\]
as desired.

\end{proof}

\subsection{Finding derivative-isolating reparametrizations}
\label{sec:derivative-isolating}
In this section we provide a general method for finding derivative-isolating 
reparametrizations based on the development in Section~\ref{sec:deriv-iso}. 
Our main result is as follows.

\begin{proposition}[Finding derivative-isolating reparametrizations]
\label{prop:find_deriv-iso}
For $\mathcal{Z} \subseteq \Delta^{d-1}$, let
$\mathsf{E} : \mathcal{Z} \to \Delta^{d-1}$ be an endpoint function satisfying
for all $z\in\mathcal{Z}$ that $\mathsf{E}(z)\neq z$. Denote by
$\mathcal{E}:= \mathsf{E}(\mathcal{Z})$ the set of endpoints and define for all
$z^* \in \mathcal{E}$ and all $u \in \mathbb{R}^d$ the sets
$\mathcal{I}_{(z^*, u)} := \{\delta \in (0, \infty) \mid z^* + \delta
u \in \Delta^{d-1}\}$ and
$\mathcal{W} := \{(w_{\mathsf{E}}, w_v) \in \mathcal{E} \times \mathbb{R}^{d}
\mid \|w_v\|_1 = 1,\, \mathcal{I}_{w} \neq \emptyset\}$. Finally,
define the unit speed reparametrization
$\phi_{\unit} = (\phi_{\unit}^L,
\phi_{\unit}^W):\mathcal{Z}\rightarrow \mathbb{R} \times \mathcal{W}$
for all $z\in\mathcal{Z}$ by $\phi_{\unit}^L(z) := -\|\mathsf{E}(z)- z\|_1$ and
$\phi_{\unit}^W(z):=(\mathsf{E}(z), v_{\mathsf{E}}(z))$ with $v_{\mathsf{E}}$ defined in \eqref{eq:endpoint_dir}. Then the following approaches lead
to perturbations with corresponding derivative-isolating
reparametrizations.

\begin{enumerate}[(a)]
\item \textit{(Speed)} Let $s:\mathcal{Z}\rightarrow (0, \infty)$ be a given
  speed function. For all $w \in \mathcal{W}$ define
  $s_w:\mathcal{I}_w\rightarrow (0,\infty)$ for all $\delta\in \mathcal{I}_w$ by
  $s_w(\delta):=s(\phi_{\unit}^{-1}(-\delta, w))$ and let $t_w : \mathcal{I}_w
  \to \mathbb{R}$ denote a solution to
  \begin{equation}
    \label{eq:diff_eq}
    -s_{w}(\delta)^{-1} = \partial_{\delta}t_{w}(\delta).
  \end{equation}
  Then, the reparametrization
  $\phi:\mathcal{Z}\rightarrow\mathbb{R}\times\mathcal{W}$ defined for all
  $z\in\mathcal{Z}$ by $\phi(z):=(t_{\phi_{\unit}^W(z)}(\|\mathsf{E}(z)-z\|_1),
  \phi_{\unit}^W(z))$ is derivative-isolating for the perturbation $\psi:
  z\mapsto z + \gamma s(z)v_{\mathsf{E}}(z)$ and any other perturbation with the same speed
  and direction.

\item \textit{(Summary statistic)} Let
  $t:\mathcal{Z}\rightarrow\mathbb{R}$ be a given summary
  statistic. For all $w \in \mathcal{W}$ define
  $t_w: \mathcal{I}_w\rightarrow \mathbb{R}$ for all
    $\delta\in \mathcal{I}_w$ by
    $t_w(\delta):=t(\phi_{\unit}^{-1}(-\delta, w))$ and
    assume that $t_w$ is differentiable and strictly
  decreasing. Then, the reparametrization
  $\phi:\mathcal{Z}\rightarrow\mathbb{R}\times\mathcal{W}$ defined for
  all $z\in\mathcal{Z}$ by $\phi(z):=(t(z), \phi_{\unit}^W(z))$ is
  derivative-isolating for the perturbation
  $\psi:z\mapsto z + \gamma s(z)v_{\mathsf{E}}(z)$, where the speed
  $s:\mathcal{Z}\rightarrow\mathbb{R}$ is defined for all
  $z\in\mathcal{Z}$ as 
  \[
    s(z):= -\left(\partial_{\delta}t_{\phi^W(z)}(\delta)\big\vert_{\delta=\|\mathsf{E}(z)-z\|_1}\right)^{-1},
  \]
  and any other perturbation with the same speed and direction.
\end{enumerate}
\end{proposition}
\begin{proof}
\emph{(a)}
For each $w \in \mathcal{W}$, the speed $s_w$ is a positive function, so by \eqref{eq:diff_eq} $t_w$ must be strictly decreasing and hence invertible on $\mathcal{L}_w := t_w(\mathcal{I}_w)$. This lets us compute the inverse of $\phi$ for each $w = (w_{\mathsf{E}}, w_v) \in \mathcal{W}$ and $\ell \in \mathcal{L}_{w}$ as
\begin{equation}
  \label{eq:phi_inv}
  \phi^{-1}(\ell, w) = -t_{w}^{-1}(\ell) w_v + w_{\mathsf{E}}
\end{equation}
which shows that $\phi$ is indeed a bijection on the image of $\phi$. For all $w \in \mathcal{W}$ and $\ell \in \mathcal{L}_{w}$ we have by \eqref{eq:diff_eq} and the inverse function rule for derivatives that
\[
  s(\phi^{-1}(\ell, w)) = s_{w}(t_{w}^{-1}(\ell)) = - \frac{1}{(\partial_{\delta} t_{w})(t_{w}^{-1}(\ell))} =  - \partial_{\ell}t_{w}^{-1}(\ell).
\]
We conclude that \eqref{eq:derivative-isolating} is satisfied and hence $\phi$ is derivative-isolating for $\psi$ since
\[
  \partial_{\ell} \phi^{-1}(\ell, (w_{\mathsf{E}}, w_v)) = - \partial_{\ell} t_{w}^{-1}(\ell) w_v.
\]

\noindent \emph{(b)} By the same argument as in \emph{(a)}, we have that $\phi$ is invertible with inverse as computed in \eqref{eq:phi_inv}. This reparametrization is derivative-isolating for perturbations with direction $v_{\mathsf{E}}$ given by \eqref{eq:endpoint_dir} and speed given by (reversing the computation in \emph{(a)})
\[
  s: z \mapsto -\partial_{\ell} t_{(\mathsf{E}(z), v_{\mathsf{E}}(z))}^{-1}(t(z)) = -(\partial_{\delta} t_{(\mathsf{E}(z), v_{\mathsf{E}}(z))}(\|\mathsf{E}(z)-z\|_1))^{-1}.
\]
\end{proof}
Part (a) starts from a given speed and finds a corresponding
derivative-isolating reparametrization by solving the differential equation in
\eqref{eq:diff_eq}. Part (b) starts from a one-dimensional summary statistic,
which is assumed to be smoothly decreasing in the distance from $z$ to the
endpoint.  Using the statistic as one component in a derivative-isolating
reparametrization, the result provides the speed of the perturbation that is
being parametrized. The conditions in (a) and (b) are not met by every
perturbation one might wish to consider and Section~\ref{sec:simplex_effects}
discusses the violations that may arise in practice. Note that requiring $t_w$
to be strictly decreasing is the same as requiring the endpoint to maximise the
statistic along every ray. In the following example, we give the missing
details from Examples~\ref{ex:cfi} and \ref{ex:gini}.

\begin{example}
\label{ex:feature_influence_details}
Recall the setting of Example~\ref{ex:cfi}. To find a derivative-isolating
reparametrization which corresponds to the speed $s_\psi(z) = 2 z^j (1-z^j)$,
we apply Proposition~\ref{prop:find_deriv-iso}~(a). Since for all $w \in
\mathcal{W}$ it holds that $w_{\mathsf{E}}^j =1$ (component $w_{\mathsf{E}}$ corresponds to the
endpoint $\mathsf{E}(z)$) and $w_v^j =\frac{1}{2}$ (component $w_v$ corresponds to the
direction $v_{\mathsf{E}}(z)$ and we used that
$v_{\mathsf{E}}(z)^j=\frac{1-z^j}{\|e_j-z^j\|_1}=\frac{1-z^j}{2(1-z^j)}=\frac{1}{2}$), we
have that $\phi_{\unit}^{-1}(\ell, w)^j = \frac{\ell}{2} + 1$. Therefore, 
\[
  s_w: \delta \mapsto s(\phi_{\unit}^{-1}(-\delta, w)) = \frac{(2-\delta)\delta}{2}.
\]
Hence, solving \eqref{eq:diff_eq} amounts to solving
\[
  -\frac{2}{(2-\delta)\delta} = \partial_{\delta} t_w(\delta)
\]
which has solutions of the form
$t_w(\delta) = \log\left(\frac{2-\delta}{\delta}\right) + C$ for
$C \in \mathbb{R}$. For simplicity, we choose $C=0$ and note that as
$\|e_j - z\|_1 = 2 (1-z^j)$, we have 
\[
  \phi^W(z)=(e_1, v_{\mathsf{E}}(z)) \quad\text{and}\quad \phi^L(z) = t_{\phi^W(z)}(\|e_j - z\|_1)=\log\left(\frac{z^j}{1-z^j}\right)
\]
We can therefore conclude (dropping the constant
endpoint from the reparametrization) that
\[
  \phi = (\phi^L, \phi^W): z \mapsto \left(\log\left(\frac{z^j}{1-z^j}\right), v_\psi(z) \right)
\]
is a derivative-isolating reparametrization for the $([d] \setminus \{j\} \to
  \{j\})$-amalgamating perturbation $\psi$ with multiplicative speed.
\end{example}

\begin{example}
  \label{ex:gini_details}
Recall the setting of Example~\ref{ex:gini}. Using that $z = z_{\cen} -
v_{\psi}(z)\|z_{\cen}-z\|_1$, we can express
\[
  G(z) = \frac{\|z_{\cen}-z\|_1}{2d} \sum_{j=1}^d \sum_{k=1}^d |v_{\psi}(z)^j - v_{\psi}(z)^k|.
\]
Hence, defining for all $w\in\mathcal{W}$ the strictly decreasing and
differentiable function $t_w: \mathcal{I}_w\rightarrow\mathbb{R}$
($\mathcal{I}_w$ is defined in Proposition~\ref{prop:find_deriv-iso}) for all
$\delta\in\mathcal{I}_w$ by
\[
  t_w(\delta) :=  1-\frac{\delta}{2d} \sum_{j=1}^d \sum_{k=1}^d |w^j - w^k|, 
\] 
we have that $t_{v_\psi(z)}(\|z_{\cen}-z\|_1) = 1-G(z)$. Applying
Proposition~\ref{prop:find_deriv-iso}~(b), we obtain that
\[
  \phi = (\phi^L, \phi^W) : z \mapsto \left(1-G(z), \frac{z_{\cen}-z}{\|z_{\cen}- z\|_1} \right)
\]
is a derivative-isolating reparametrization for perturbations with direction
$v_\psi$ and speed given by 
\begin{equation*}
  s_\psi(z) := \frac{2d}{ \sum_{j=1}^d \sum_{k=1}^d |v_\psi(z)^j - v_\psi(z)^k|}.
\end{equation*}
\end{example}

\subsection{Amalgamation effects}
\label{sec:amalgamations}
In this section we consider general amalgamating perturbations as defined in
Section~\ref{sec:endpoints} rather than only those which push towards a vertex
of the simplex, i.e., those we describe in
Section~\ref{sec:simplex_effects}. We first generalize
Example~\ref{ex:cfi}.

\begin{example}[Amalgamation influence]
\label{ex:cai}
For disjoint $A, B \subseteq [d]$, consider an $(A\to B)$-amalgamating
perturbation. Suppose as in Example~\ref{ex:cfi} that $Z := \C(X)$ for some $X
\in \mathbb{R}^d_+ \setminus \{0\}$ on an absolute scale. We can omit components
of $Z$ and $X$ that are not in $A \cup B$ from the analysis as these do not
change under the perturbation, so we suppose now that $A \cup B = [d]$. The
multiplicative operation described in Example~\ref{ex:cfi} can be generalized by
letting $u_B = \sum_{j \in B} e_j$ and considering
\[
  c \mapsto \C(X \odot (1 + c u_B)),
\]
which similarly corresponds to a well-defined perturbation on the simplex.
Computing the speed, we obtain that
\[
  \| \partial_c \C(X \odot (1 + c u_B)) \bigm|_{c = 0} \|_1 = \frac{2 \|X^A\|_1 \|X^B\|_1}{(\|X^A\|_1 + \|X^B\|_1)^2}.
\]
We can use this speed even when
$A \cup B \neq [d]$ resulting in the general multiplicative speed
\begin{equation}
  \label{eq:multiplicative_speed}
  s_\psi(z) = \frac{2 \|z^A\|_1 \|z^B\|_1}{(\|z^A\|_1 + \|z^B\|_1)^2}.
\end{equation}
Using Proposition~\ref{prop:find_deriv-iso}~(a) as in Example~\ref{ex:cfi}, we
further obtain that
\[
  \phi = (\phi^L, \phi^W): z \mapsto \left(\log\left(\frac{\|z^B\|_1}{\|z^A\|_1}\right), (\A_{A \to B}(z), v_\psi(z)) \right)
\]
is a derivative-isolating reparametrization for $\psi$ with multiplicative speed.
\end{example}

With this generalization, we can consider perturbation effects for general
amalgmating perturbations. As an example, consider a study where data is
collected on how individuals spend their time each day (e.g., hours spent
sleeping, eating, working, doing sports, etc.) as well as a response variable
(e.g., general happiness). Assume now that we want to use this data to
investigate the relationship between time spent on sports and happiness. As
in Section~\ref{sec:simplex_effects}, we are again
interested in the effect of an individual component, so we might consider
estimating a $\cfi^j$. However, this would capture the effect of increasing time
spent on sports while all other activities are down-scaled proportionally. For
example, time spent sleeping and eating would both be decreased if time doing
sports is increased which is not realistic. To address this, we can instead
consider a general $(A\to B)$-amalgamating perturbation (as defined in
Section~\ref{sec:endpoints}), where $B=\{\text{sport}\}$ and $A$ corresponds to
activities which are reduced to make increased time for sports. Perturbing in
this fashion keeps the components in $[d]\setminus(A\cup B)$ fixed (e.g., eating
and drinking) while down-scaling the activities in $A$. We call directional
perturbation effects of this form \emph{compositional amalgamation influences},
or $\cai^{A \to B}$ for short. As with the $\cfi^j$, we use subscripts to denote
the speed of the perturbation, so we have $\cai_{\unit}^{A \to B}$ and $\cai^{A
\to B}_{\mult}$ (see Example~\ref{ex:cai}).

Analogously, we could be interested in the effect of removing particular
activities (e.g.\ gambling) entirely and spending the time on desirable
activities (e.g.\ sports) while fixing time spent on a third group of activities
(e.g.\ eating, sleeping). We can capture this effect using a binary amalgamating
perturbation $\psi$. We call the resulting average binary perturbation effect
for $\psi$ an $(A\to B)$-\emph{compositional amalgamation effect}, or $\cae^{A
\to B}$ for short.

\subsection{Perturbations arising from alternative simplex geometries}
\label{sec:other_geometries}
In this section we describe an alternative way to construct perturbations
than the one considered in Sections~\ref{sec:endpoints} and
\ref{sec:deriv-iso}. When constructing the directional perturbations in
Section~\ref{sec:deriv-iso}, we considered (Euclidean) straight lines to the
endpoint to construct directions, that is, $v_{\mathsf{E}}: z \mapsto
(\mathsf{E}(z)-z)/\|\mathsf{E}(z)-z\|_1$.  While we believe that this is natural in many
applications, one could also imagine starting from alternative geometries on the
simplex and construct other types of directions. We illustrate how this would
work for the Aitchison geometry and show that it provides an interesting
connection to the log-ratio based literature. The Aitchison geometry is a
popular alternative geometry that goes back to the work of
\citet{aitchison1982statistical} who proposed to analyze the interior of the
simplex by transforming to a subset of Euclidean space and applying conventional
methods to the transformed data. A popular variant of this transformation is the
\emph{centered log-ratio transform} (CLR) defined for all
$z\in(\Delta^{d-1})^{\circ}$ (the interior of the simplex) by
\[
\clr(z)^j := \log(z^j) - \frac{1}{d} \sum_{j=1}^d \log(z^j), \quad \forall j \in [d].
\]
This transformation maps $(\Delta^{d-1})^{\circ}$ to the subspace of $\mathbb{R}^d$
given by $\{ x \in \mathbb{R}^d \mid \sum_{j=1}^d x^j = 0\}$. The inverse CLR is defined for all $x\in\operatorname{Im}(\clr)$
by
\[
  \clr^{-1}(x)^j := \frac{\exp(x^j)}{\sum_{j=1}^d \exp(x^j)}, \quad
  \forall j \in [d],
\]
or more simply as $\clr^{-1}(x)=\C(\exp(x))$ where $\exp$ is applied
componentwise. The inverse of the $\clr$ transform can in fact be defined on all
of $\mathbb{R}^d$ rather than just the range of the $\clr$-function as is clear
from the definition. The $\clr$-transformation induces what is called the
Aitchison geometry \citep{egozcue2006simplicial} on the simplex. The Aitchison
geometry induces a different notion of direction than the Euclidean notion
considered above and, in particular, the direction of the line connecting $z$ to
an endpoint may differ from the direction in \eqref{eq:endpoint_dir} (see, for
example, Figure~\ref{fig:perturbation_illustration}). However, as $\clr(z)$ is a
bijection on the open simplex, we can just as well consider perturbations on
$\mathbb{R}^d$ directly and use arguments similar to those in
Proposition~\ref{prop:find_deriv-iso} but now for the perturbation on the
log-ratio scale. We illustrate this procedure in the following two examples and
discuss connections to the centering and amalgamating perturbations defined
in Examples~\ref{ex:cfi} and \ref{ex:gini}.

\begin{example}[Diversity in the Aitchison geometry]
  Suppose we are interested in the effect of diversifying $z \in
  (\Delta^{d-1})^{\circ}$. As in Example~\ref{ex:gini}, we can think of such
  effects as moving towards the center $z_{\cen}$, but now using the Aitchison
  geometry this corresponds to moving $\clr(z)$ towards $\clr(z_{\cen})=0$ along
  a straight line in CLR-space. Using the Aitchison geometry, this 
  defines a perturbation via 
  \[
    \psi_{\mathrm{Ait}}: (z, \gamma) \mapsto \clr^{-1}((1-\gamma)\clr(z)).
  \]
  We can compute $\partial_{\gamma} \psi_{\mathrm{Ait}}(z, \gamma)
  \bigm|_{\gamma=0}$ to find a perturbation $\psi$ of the form
  \eqref{eq:lin_pert} for which $\tau_{\psi} = \tau_{\psi_{\mathrm{Ait}}}$.
  Applying the chain rule, we obtain
  \[
    \partial_{\gamma} \psi_{\mathrm{Ait}}(z, \gamma)^j \bigm|_{\gamma=0}  = \sum_{k \in [d]: k \neq j} z^j z^k \log\left( \frac{z^k}{z^j}\right),
  \]
  which can be decomposed into a speed and direction as in Definition~\ref{def:perturbation} to obtain $\psi$. The direction of this
  derivative can be vastly different from the centering direction used in
  Example~\ref{ex:gini}. To obtain a derivative-isolating reparametrization, since
  $\psi_{\mathrm{Ait}}$ shrinks towards $0$ in CLR-space, it is straightforward
  to show that $\phi^L: z \mapsto -\|\clr(z)\|_{2}$ and $\phi^W: z \mapsto
  \clr(z)/\|\clr(z)\|_2$ are derivative-isolating reparametrizations.
\end{example}

\begin{example}[Feature influence in the Aitchison geometry]
  Suppose, as in Example~\ref{ex:cfi}, we are interested in the effect of
  increasing the value of a particular component $z^j$ of $z \in
  (\Delta^{d-1})^{\circ}$. There is a natural notion of increasing $z^j$
  inherited from CLR-space, namely adding $\gamma$ to $\clr(z)^j$ and
  transforming back to the simplex. The resulting perturbation is given by 
  \[
    \psi:(z, \gamma) \mapsto \clr^{-1}(\clr(z) + \gamma e_j).
  \]
  The derivative of the resulting perturbation is for all $k \in [d]$ given by
  \[
    \partial_{\gamma} \psi(z, \gamma)^k \bigm|_{\gamma = 0} = \ind_{\{k = j\}}z^j - z^jz^k.
  \]
  It is straight-forward to check that the speed and direction of this
  perturbation match exactly the multiplicative perturbation derived
  in Example~\ref{ex:cfi}. One can therefore view an
  additive log-ratio perturbation as a multiplicative speed
  perturbation.
\end{example}
We summarize the types of perturbations considered in this work in
Figure~\ref{fig:perturbation_illustration}.

\begin{figure}[ht!]
  \centering
  \includegraphics*[width=\textwidth]{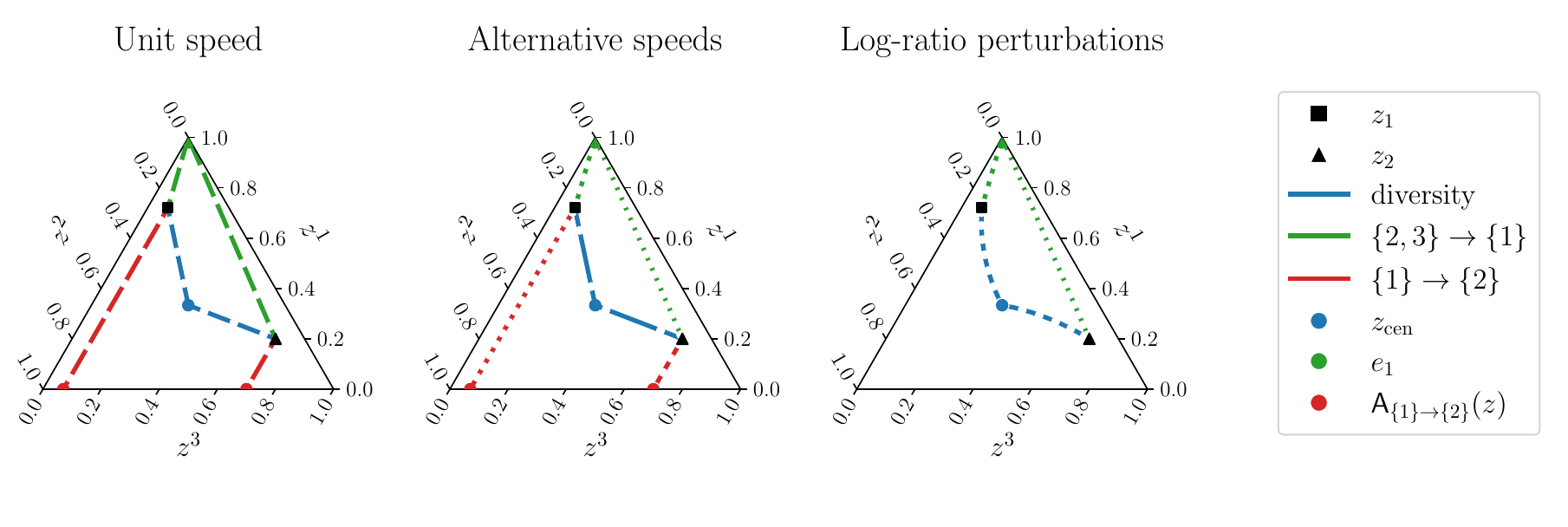}
  \caption{Plots of (left) unit speed perturbations, (middle) alternative
  speed parametrizations and (right) log-ratio perturbations for two different
  initial points $z_1$ and $z_2$. The types of perturbation are distinguished by
  color. The alternative speed for the diversity perturbation in blue is Gini-speed
  (see \eqref{eq:gini_speed}) and for the amalgamating perturbations in red and
  green the speed is multiplicative speed (see \eqref{eq:multiplicative_speed}).
  The dashes represent the speed of the perturbation and are equally spaced in the
  $\gamma$-domain, so that longer lines correspond to faster speeds. The log-ratio
  diversity perturbation is not a straight line and does not have the same
  direction as the centering perturbations.} \label{fig:perturbation_illustration}
\end{figure}

\subsection{Perturbation effects and classical effect measures}
\label{sec:misc_results}
In this section we present two results which connect our perturbation effects to
classical effect measures in the literature. The first result provides a
relationship between the $\cfi^j$ and coefficients in a well-specified
log-contrast regression model. The result has previously been proven in
\citet{huang2022supervised}.
\begin{proposition}
  \label{prop:log-contrast}
Let $(Y, Z)$ denote random variables taking values in $\mathbb{R} \times
\Delta^{d-1}$ and let $f: \Delta^{d-1} \to \mathbb{R}$ be given by $f: z \mapsto
\E[Y \given Z=z]$. Suppose that there exists $\beta = (\beta^1, \dots, \beta^d)
\in \mathbb{R}^d$ with $\sum_{j=1}^d \beta^j = 0$ and $\mu \in \mathbb{R}$ such
that for all $z = (z^1, \dots, z^d) \in \Delta^{d-1}$ we have $f(z) =
\sum_{j=1}^d \beta^j \log(z^j) + \mu$. Then, for all $j \in [d]$, we have
  \[
    \cfi^j_{\mult} = \beta^j .
  \]
\begin{proof}
  The perturbation corresponding to $\cfi^j_{\mult}$ is given by
  \[
    \psi(z, \gamma) = z + \gamma \frac{2 z^j (1-z^j)}{\|e_j - z\|_1} (e_j - z) = z + \gamma z^j (e_j - z). 
  \]
  We thus have
  \[
    f(\psi(z, \gamma)) = \sum_{k=1}^d \beta^k \log(z^k + \gamma z^j (\ind_{\{k=j\}} - z^k)) + \mu
  \]
  and therefore
  \[
    \partial_{\gamma} f(\psi(z, \gamma)) \bigm|_{\gamma = 0} = \sum_{k=1}^d \frac{\beta^k z^j (\ind_{\{k=j\}} - z^k)}{z^k} = \beta^j - z^j \sum_{k=1}^d \beta^k = \beta^j.
  \]
  We conclude that 
  \[
    \cfi^j_{\mult} = \tau_{\psi} = \E\left[\partial_{\gamma} f(\psi(z, \gamma)) \bigm|_{\gamma = 0} \right] = \beta^j.
  \]
\end{proof}
\end{proposition}
The second result shows that the effects proposed in
\citet{cox1971note} are special cases of average perturbation effects in a
parametric model.
\begin{proposition}
  \label{prop:cox}
  Let $(Y, Z)$ denote random variables taking values in
  $\mathbb{R} \times \Delta^{d-1}$ and let $z_* \in \Delta^{d-1}$. Let
  $f: \Delta^{d-1} \to \mathbb{R}$ be given by
  $f: z \mapsto \E[Y \given Z=z]$ and suppose that there exists
  $\beta = (\beta^1, \dots, \beta^d) \in \mathbb{R}^d$ with
  $\sum_{j=1}^d \beta^j z_*^j = 0$ and $\mu \in \mathbb{R}$ such that
  for all $z = (z^1, \dots, z^d) \in \Delta^{d-1}$ we have
  $f(z) = \sum_{j=1}^d \beta^j z^j + \mu$. For
  $k\in [d]$, define the perturbation
  $\psi_{k}$ such that for all
    components $j\in [d]$ it holds
  \begin{equation*}
    \psi_k^j(z, \gamma)=
    \begin{cases}
      z^j + \gamma\quad &\text{if } j=k\\
      z^j - \gamma \frac{z^j_*}{1-z^k_*}\quad &\text{else}.
    \end{cases}
  \end{equation*}
    Then
  \[
    \tau_{\psi_k} = \frac{\beta^k}{1-z^k_*}.
  \]
\begin{proof}
  We have
  \[
    f(\psi_k(z, \gamma)) = \beta^k(z^k+\gamma) + \sum\limits_{\substack{j=1\\j \neq k}}^d  \beta^j \left( z^j - \gamma \frac{z^j_*}{1-z^k_*} \right) + \mu
  \]
  and therefore
  \[
    \partial_{\gamma} f(\psi_k(z, \gamma)) \bigm|_{\gamma = 0} = \beta^k - \sum\limits_{\substack{j=1\\j \neq k}}^d  \beta^j  \frac{z^j_*}{1-z^k_*} = \beta^k \left( 1 + \frac{z^k_*}{1-z^k_*} \right) - \sum\limits_{j=1}^d  \beta^j  \frac{z^j_*}{1-z^k_*} = \frac{\beta^k}{1-z^k_*}.
  \]
  We conclude that 
  \[
    \tau_{\psi_k} = \E\left[\partial_{\gamma} f(\psi_k(z, \gamma)) \bigm|_{\gamma = 0} \right] = \frac{\beta^k}{1-z^k_*}.
  \]
  \end{proof}
\end{proposition}

\section{Perturbations as causal quantities}\label{sec:causal_models}

Perturbations are a priori non-causal quantities and should not be confused with
(causal) interventions. Causal models \citep{pearl2009causality,
rubin2005causal} formalize interventions as changes to the data-generating
mechanism of individual variables. In contrast to an (observational) statistical
model, which assumes a single distribution that generated the data, a causal
model additionally postulates a collection of interventional distributions
describing the system under such changes. For a compositional predictor $Z$ this
formalism does not directly apply at the level of individual components. The
coordinates of $Z$ are deterministically tied together by the sum-to-one
constraint, so there is no free-standing variable $Z^j$, nor one corresponding
to a summary statistic such as the diversity of $Z$, whose assignment could be
changed while the rest of the system is left as it was. This is an instance of
the general difficulty of defining interventions on deterministically related
variables \citep{spirtes2004causal, woodward2015interventionism,
beckers2023causal}. A causal model over $(Y, Z)$ can therefore only assign
interventional meaning to changes of the composition as a whole. Perturbations
are precisely a way of encoding a component-level or statistic-level question
(`increase $Z^j$', `increase diversity') as a change of the entire composition,
and are thus the device that connects such questions to well-defined
interventions.

More specifically, let $\mathcal{D} \subseteq \Delta^{d-1} \times [0, \infty)$
and let $\psi : \mathcal{D} \to \Delta^{d-1}$ be a perturbation as defined in
Definition~\ref{def:perturbation}. For all $\gamma$ sufficiently small, the
\emph{perturbation-induced intervention} replaces, for each unit, the
composition it would otherwise have had, $Z$, its natural value, by the
perturbed composition $\psi(Z, \gamma)$. We write $\operatorname{do}(Z = \psi(Z,
\gamma))$ for this intervention and note that the intervention stochastically
modifies the value of $Z$ rather than setting it to a fixed point of the
simplex. Since we only intervene on the entire compositional vector, it is not
possible to model how causal effects propagate within the composition itself.
The corresponding causal versions of the target parameters in
Definition~\ref{def:avg_per_eff} are the causal average binary perturbation
effect
\begin{equation*}
  \lambda_{\psi}^{\operatorname{causal}}
  := \frac{\E\left[Y \given \operatorname{do}(Z = \psi(Z, 1))\right]
      - \E[Y]}{\mathbb{P}(Z \neq \psi(Z, 1))}
\end{equation*}
and the causal average directional perturbation effect
\begin{equation*}
  \tau_{\psi}^{\operatorname{causal}}
  := \partial_{\gamma} \E\left[Y \given
      \operatorname{do}(Z = \psi(Z, \gamma))\right] \big|_{\gamma = 0}.
\end{equation*}
These definitions only presuppose that the expectation of $Y$ under the
intervention is well-defined and can be formalized in any causal framework that
permits interventions on the whole vector $Z$; we make them precise in the
structural causal model (SCM) framework \citep{pearl2009causality,
bongers2021foundations} below and return to the potential outcome framework at
the end of this section. Unlike their non-causal counterparts, the causal
average perturbation effects may not be identifiable from the observational
distribution, and whether they coincide with the effects of
Definition~\ref{def:avg_per_eff} is settled by assumptions on the
data-generating mechanism, which we also discuss below.

Before turning to identification, we discuss what role the reparametrizations of
Section~\ref{sec:deriv-iso} play in the causal analysis. Given a perturbation
$\psi$, a derivative-isolating reparametrization $\phi = (\phi^L, \phi^W) :
\Delta^{d-1} \to \mathbb{R} \times \mathcal{W}$ which additionally satisfies for
all $(z, \gamma) \in \mathcal{D}$ that $\phi^W(\psi(z, \gamma)) = \phi^W(z)$, as
each of the reparametrizations constructed in
Section~\ref{sec:simplex_perturbations} does, maps $Z$ to $(L, W)$ such that the
perturbation-induced intervention leaves $W$ unchanged and replaces $L$ by
$\phi^L(\psi(Z, \gamma))$. On the $(L, W)$-scale the intervention therefore
modifies the assignment of $L$ only, though not by setting $L$ to a fixed value
(see Figure~\ref{fig:causal_graph}). The reparametrization does not by itself
determine an SCM over $(Y, L, W)$: $L$ and $W$ are deterministic functions of
the same variable $Z$ and cannot be written as two root variables with
independent errors. Figure~\ref{fig:causal_graph}~(b) displays one factorization
of their joint distribution, in which $W$ is exogenous and $L$ is assigned from
it. The opposite direction is unavailable, since intervening on $L$ would then
also change $W$, and the choice does not affect the effects we consider, which
modify the assignment of $L$ only. This factorization suggests a useful reading:
after reparametrizing, the simplex constraint reappears as confounding of the
$L$--$Y$ relationship by $W$. We emphasize that this is an interpretation
induced by the reparametrization and not a causal claim, as $W$ is not a cause
of $L$. The interpretation is nevertheless valuable in two ways: it explains why
marginal analyses of $L$ (e.g., regressing $Y$ on a diversity index) are biased
by the remaining variation in $Z$ (see Sections~\ref{sec:diversity_exps}
and~\ref{sec:toy_examples}), and it makes the adjustment machinery of
Section~\ref{sec:estimation} applicable. At the same time, the $(L,
W)$-representation is a tool tied to the specific perturbation $\psi$ for which
$\phi$ was constructed: interventions on $W$, or changes of $L$ other than along
the perturbation trajectory, have no compositional meaning, and a different
perturbation requires its own reparametrization.

\begin{figure}[ht!]
\centering
\begin{minipage}{1.0\linewidth}
  \hspace{1em}\textbf{(a) SCM on $Z$-scale} \hfill \textbf{(b) SCM
    on $(L,W)$-scale}\hspace{6em}
\end{minipage}

\begin{minipage}{0.3\textwidth}
  \centering
  \begin{align*}
    Y&\leftarrow f(Z, \eta)
  \end{align*}
  \begin{tikzpicture}[scale=1.3]
    \tikzset{round/.style={circle, draw, fill=white!80!black, minimum
        width=2em}}
    \node[round] at (1,0.75) (Y){$Y$};
    \node[round] at (-1,0.75) (Z){$Z$};
    \node[red] at (-1.25,1.75)(P){$\psi$};
    \draw [-{Latex}{Latex}, red, thick] (-1.25,1.6)--(-1.1,1.1); 
    \path[-Latex,draw,thick] (Z) edge node[fill=white] {$f$} (Y);
  \end{tikzpicture}
\end{minipage}%
\begin{minipage}{0.3\textwidth}
  \centering
  \begin{tikzpicture}
    \tikzset{round/.style={circle, draw, fill=white!80!black, minimum
        width=2em}}
    \path[-Latex, thick] (0.5, 4) edge[bend left] (5,4);
    \node[align=center] at (2.75, 5.25) {bijection $\phi$};
    \node[align=center] at (2.75, 4) {$\phi(Z)=(L, W)$};
    \path[-Latex, thick] (5, 3) edge[bend left] (0.5,3);
    \node[align=center] at (2.75, 3) {$Z=\phi^{-1}(L, W)$};
  \end{tikzpicture}
\end{minipage}%
\begin{minipage}{0.25\textwidth}
  \centering
  \begin{tikzpicture}[scale=1.3]
    \tikzset{round/.style={circle, draw, fill=white!80!black, minimum
        width=2em}}
    \node[round] at (6,0) (Y2){$Y$};
    \node[round] at (4,0) (L){$L$};
    \node[round] at (5,1.5) (W){$W$};
    \node[red] at (3.7,0.95)(P){$\psi$};
    \draw [-{Latex}{Latex}, red, thick] (3.7,0.8)--(3.85,0.3);
    \path[-Latex,draw,thick] (L) edge node[fill=white] {$h$} (Y2);
    \path[-Latex,draw,thick] (W) edge node[fill=white] {$m$} (L);
    \path[-Latex,draw,thick] (W) edge node[fill=white] {$h$} (Y2);
  \end{tikzpicture}
\end{minipage}%
\begin{minipage}{0.15\textwidth}
  \begin{align*}
    L&\leftarrow m(W, \xi)\\
    Y&\leftarrow h(L, W,\varepsilon)
  \end{align*}
\end{minipage}

\caption{(a) Data-generating SCM on the
  $Z$-scale. $Z\in\Delta^{d-1}$ is a compositional vector, $Y$ is a
  response variable and $\eta$ is exogenous noise. (b)
  Reparametrized SCM on the $(L,W)$-scale. $(L, W)=\phi(Z)$ are the
  reparametrized predictors, $Y$ is the same response as in (a) and
  $\xi$ and $\varepsilon$ are exogenous noise variables. The
  reparametrization needs to disentangle the perturbation effect on
  $Z$ in a way such that intervening on $Z$ by setting it to
  $\psi(Z, \gamma)$ corresponds to intervening on only $L$ (red
  double head arrows).  }\label{fig:causal_graph}
\end{figure}

Relatedly, the derivative-isolating condition of
Proposition~\ref{prop:derivative-isolating} is not a causal assumption. It
restricts only the pair $(\psi, \phi)$, involves no unobserved variables and no
counterfactuals, and can be checked by differentiation. A perturbation may admit
a derivative-isolating reparametrization while its causal effect remains
unidentified, and a causal effect may be identified for a perturbation that
admits none. The one loose parallel is with positivity, in that the construction
is stated for $\omega_{\psi}(z) \neq 0$, so that an effect is defined only where
the perturbation actually moves the composition. This is not an additional
assumption: by Remark~\ref{rmk:zeros}, points with $\omega_{\psi}(z) = 0$ do not
contribute to $\tau_\psi$, and Section~\ref{sec:zeros} describes how to adapt
the estimation to allow them.

We now formalize the above in the SCM framework. Assume first that the SCM in
Figure~\ref{fig:causal_graph}~(a) generates the data, that is, $Z$ is exogenous
(it does not depend on other variables in its structural assignment). The
perturbation-induced intervention replaces the structural assignment of $Z$, say
$Z \leftarrow q$, by $Z \leftarrow \psi(q, \gamma)$. Since $Z$ is exogenous,
intervening on $Z$ and conditioning on it agree, so the causal average
perturbation effects coincide with the effects of
Definition~\ref{def:avg_per_eff} and the estimators of
Section~\ref{sec:estimation} estimate them directly. As a second example, in
which this is no longer the case, assume the data are generated by the SCM over
$(H, Z, Y) \in \mathbb{R}^q \times \Delta^{d-1} \times \mathbb{R}$ given by
\begin{equation*}
  \mathcal{M} :
  \begin{cases}
    \, H \leftarrow \varepsilon_H \\
    \, Z \leftarrow h(H, \varepsilon_Z) \\
    \, Y \leftarrow g(H, Z) + \varepsilon_Y,
  \end{cases}
\end{equation*}
where $h : \mathbb{R}^q \times \mathbb{R} \to \Delta^{d-1}$ and $g :
\mathbb{R}^q \times \Delta^{d-1} \to \mathbb{R}$ are measurable functions and
$(\varepsilon_H, \varepsilon_Y, \varepsilon_Z)$ are jointly independent noise
terms with expectation zero. Here $Z$ causally precedes the response but is
confounded by $H$. The perturbation-induced intervention $\operatorname{do}(Z
= \psi(Z, \gamma))$ in $\mathcal{M}$ replaces the $Z$-structural assignment by
$Z \leftarrow \psi(h(H, \varepsilon_Z), \gamma)$, so the value $Z$ would
otherwise have taken retains its observational joint distribution with $H$ and
$\E[Y \,|\, \operatorname{do}(Z = \psi(Z, \gamma))] =
\E[g(H, \psi(Z, \gamma))]$. If $H$ is unobserved, this quantity is in
general not identifiable from the observational distribution. If $H$ is
observed, it can be expressed as
\begin{equation}
  \label{eq:adjustment}
  \E\left[Y \given \operatorname{do}(Z = \psi(Z, \gamma))\right]
  = \E\left[\E\left[Y \given Z = \psi(Z, \gamma),
    H\right]\right]
\end{equation}
and hence is identifiable. The formula in \eqref{eq:adjustment} is a special
case of the adjustment formula \citep[e.g.,][]{pearl2009causality}, but other
ways of identifying this effect also exist (e.g., propensity score matching).
It is easy to adapt the estimators proposed in this work to such settings:
starting from the reparametrization into $(L, W)$ one only needs to add the
observed confounders alongside $W$. This is for example done in the numerical
experiments in Section~\ref{sec:diversity_exps}.

There are several possible assumptions which permit us to conclude that the
causal and non-causal perturbation effect coincide. Obviously, one sufficient
condition is to assume that
\begin{equation}
  \label{eq:coincide}
  \E\left[Y \given \operatorname{do}(Z = \psi(Z, \gamma))\right]
  = \E\left[Y \given Z = \psi(Z, \gamma)\right].
\end{equation}
Such a condition holds by definition whenever $Z$ is exogenous, as in the
first example above, but can also hold more generally. A common assumption to
ensure \eqref{eq:coincide} is that $Z$ and $Y$ are not confounded, that is,
have no common causes. In $\mathcal{M}$ this would for example mean that
either $g$ or $h$ does not depend on $H$. In many cases this is not satisfied
and one should instead carefully argue for which causal assumptions are
reasonable and adjust the estimation accordingly, e.g., by conditioning on
additional covariates as in \eqref{eq:adjustment}.

Finally, the same target parameters can be stated in the potential outcome
framework. Consistent with the discussion at the beginning of this section,
potential outcomes are indexed by entire compositions: we write $Y(z)$ for the
potential outcome under an intervention setting $Z$ to $z \in \Delta^{d-1}$, and
there is no component-level potential outcome corresponding to $Z^j$ alone.
Since a perturbation modifies the observed composition, the relevant potential
outcome is $Y(\psi(Z, \gamma))$, which depends on the natural value of $Z$. The
causal average binary perturbation effect is then the contrast $\E[Y(\psi(Z, 1))
- Y(Z)]$, scaled as above, and the causal average directional perturbation
effect is the derivative at $\gamma = 0$ of $\gamma \mapsto \E[Y(\psi(Z,
\gamma))]$. Identification requires the usual assumptions of consistency,
exchangeability given the observed covariates and positivity; these play the
same roles as the causal assumptions in the SCM formulation above, and none of
them are specific to the simplex or supplied by the perturbation. Because the
intervened value $\psi(Z, \gamma)$ depends on the observed value of $Z$, the
interventions considered here are modified treatment policies, also studied
under the names stochastic and natural-value interventions
\citep{diaz2012population, haneuse2013estimation, kennedy2019nonparametric,
diaz2023nonparametric}. That literature develops identification and efficient
estimation for such interventions when the treatment is unconstrained, and those
results apply once a perturbation has been specified. What it does not address
is which modifications of the treatment are admissible when the treatment is
confined to a simplex, which is the question the constructions of
Section~\ref{sec:simplex_perturbations} answer.

\section{Inference and algorithms for perturbation effects}
\label{sec:extra_est_and_alg}
In this section we provide details on the theory and algorithms for the
estimation of perturbation effects. Section~\ref{sec:np_estimation} introduces
the nonparametric estimators of perturbation effects. Section~\ref{sec:theory}
contains a summary of the theoretical results and algorithms for the main
estimation procedures. Section~\ref{sec:main_proofs} contains a proof of our
main theoretical result. Section~\ref{sec:estimator_robustness} investigates the
robustness of the different estimators. Section~\ref{sec:zeros} describes how to
modify the estimation of $\tau_\psi$ using a derivative-isolating
reparametrization in the presence of observations with $\omega_\psi(Z) = 0$.
Section~\ref{sec:tree-derivative} describes our procedure for the estimation of
the derivative of a generic regression function. Section~\ref{sec:loc-scale}
describes our procedure for the estimation of the conditional score function
$\rho$ that is required for the estimation of an average partial effect.

\subsection{Nonparametric estimation of perturbation effects}
\label{sec:np_estimation}
\subsubsection{Binary perturbation effect}
\label{sec:estimation_np_binary}
Consider a binary perturbation $\psi$ with corresponding perturbation effect
$\lambda_{\psi}$ and recall from the arguments of
Section~\ref{sec:perturbations} that with $L := \ind_{\{Z = \psi(Z, 1)\}}$ and
$W:= \psi(Z, 1)$, we have that
  \begin{equation*}
    \lambda_{\psi}
    = \frac{\E\big[\E\left[Y \given L=1, W \right]]
    - \E\left[Y \right] }{\mathbb{P}(L = 0)}.
  \end{equation*}
The denominator and $\E[Y]$ are both easy to estimate efficiently using their
empirical counterparts (the one-step correction term is $0$ for these
quantities). The term $\E\big[\E\left[Y \given L=1, W \right]]$ in the numerator
is trickier and is related to the notion of the \emph{average predictive effect}
\citep[e.g.,][]{robins1994estimation,chernozhukov2018double,kennedy2023semiparametric}
that has been analyzed extensively in the causal inference literature. The
one-step corrected estimator for $\E\big[\E\left[Y \given L=1, W \right]]$
corresponds to the \emph{augmented inverse propensity weighted} (AIPW) estimator
and is given by
\[
  \frac{1}{n} \sum_{i=1}^n \widehat{f}(1, W_i) + \frac{Y_i - \widehat{f}(L_i, W_i)}{\widehat{\pi}(W_i)} L_i,
\]
where $\widehat{f}$ is an estimator of $f: (\ell, w) \mapsto \E[Y \given L=\ell,
W=w]$ and $\widehat{\pi}$ is an estimator of the propensity score $\pi : w
\mapsto \mathbb{P}(L = 1 \given W=w)$. An efficient estimator for
$\widehat{\lambda}_\psi$ is thus
\begin{equation}
  \label{eq:nonparametric_discPE}
  \widehat{\lambda}_\psi := \frac{\frac{1}{n} \sum_{i=1}^n \widehat{f}(1, W_i) - Y_i + \frac{Y_i - \widehat{f}(L_i, W_i)}{\widehat{\pi}(W_i)} L_i}{\frac{1}{n} \sum_{i=1}^n (1-L_i)}.
\end{equation}

\subsubsection{Directional perturbation effect}
\label{sec:estimation_np_directional} 

Consider a directional perturbation $\psi$ with a derivative-isolating
reparametrization $\phi = (\phi^L, \phi^W)$ and the corresponding perturbation
effect $\tau_{\psi}$. By the arguments in Section~\ref{sec:perturbations},
setting $L := \phi^L(Z)$ and $W:= \phi^W(Z)$, we have that
\begin{equation*}
  \tau_\psi =  \E[\partial_\ell \E[Y \given L=\ell, W] \big|_{L=\ell}].
\end{equation*}
This quantity is another well-known target of inference; the \emph{average
partial effect} \citep{newey1993efficiency,chernozhukov2022debiased,
klyne2023average}. The one-step corrected estimator in this setting is given by
\begin{equation}
  \label{eq:nonparametric_contPE}
  \widehat{\tau}_{\psi} := \frac{1}{n} \sum_{i=1}^n \partial_\ell \widehat{f}(L_i, W_i) - \widehat{\rho}(L_i, W_i)(Y_i - \widehat{f}(L_i, W_i))
\end{equation}
where $\widehat{f}$ is a differentiable estimator of $f: (\ell, w) \mapsto \E[Y
\given L=\ell, W=w]$ and $\widehat{\rho}$ is an estimate of the conditional
score function $\rho$ of $L$ given $W$, that is, an estimate of the function
$\rho: (\ell, w)\mapsto \partial_{\ell}\log(p(\ell\vert w))$ where $p$ denotes
the conditional density of $L$ given $W$.

For the estimation of $f$ we require that the estimator $\widehat{f}$ is
differentiable and with computable derivatives (which additionally need to
approximate the true derivative at a sufficiently fast rate). This excludes the
direct use of many popular machine learning methods such as random forests or
boosted trees.  One solution is to start from a potentially non-differentiable
pilot estimate of $f$ and smooth it to get a differentiable estimate for which
derivatives can be computed. Such a procedure based on kernel-smoothing has been
proposed by \citet{klyne2023average}. Here, we use an alternative smoothing
method which uses random forests and local polynomials. Finally, for the
estimation of $\rho$, we follow the proposal in \citet{klyne2023average} and
assume a location-scale model. This requires the estimation of the conditional
mean and variance of $L$ given $W$ (referred to as $\eta$ in
Algorithm~\ref{alg:nonparametric_tau}). Further details on our smoothing method
and the location-scale based estimate of $\rho$ are given in
Sections~\ref{sec:tree-derivative} and \ref{sec:loc-scale}, respectively.

\subsection{Theoretical analysis of estimators}\label{sec:theory}

In this section we present a result on asymptotic normality
uniformly over large classes of distributions for each of the three estimation
methods described in Section~\ref{sec:estimation}; nonparametric estimation of
$\tau_\psi$, nonparametric estimation of $\lambda_\psi$ and by means of
estimating $\theta$ in the partially linear model. These results
are derived using standard proof techniques from the semiparametric literature
and likely exist individually in similar form elsewhere
\citep[e.g.,][]{klyne2023average,chernozhukov2018double}. We do not claim that
any of the derived asymptotic results are novel or surprising but include them
for completeness and to motivate the precise form of the estimators and their
corresponding variance estimates.

We start from a given perturbation $\psi$ and a corresponding reparametrization
$\phi:\mathcal{Z}\rightarrow\mathbb{R}\times\mathcal{W}$ which reparametrizes
$Z$ into $(L, W):=(\phi^L(Z), \phi^W(Z))$ and assume that $(L, W)$ satisfy
\eqref{eq:simple_dape} if $\psi$ is a binary perturbation and that $\phi$ is
derivative-isolating if $\psi$ is a directional perturbation. Our goal is now to
use the estimation procedures discussed in Section~\ref{sec:estimation} to
estimate the average perturbation effects (i.e., $\tau_\psi$ or $\lambda_\psi$)
consistently and construct asymptotically valid confidence intervals for the
estimates. To construct asymptotically valid confidence intervals, we estimate
both the target parameter and its asymptotic variance and then show that
the resulting standardized estimators are asymptotically standard normal
with a uniform convergence rate over the class of possible distributions of the
data. Distributionally uniform results are desirable here as they ensure uniform
coverage guarantees of the confidence intervals. The following result shows
uniform asymptotic normality of each of the three estimation procedures.
\begin{theorem}[Uniform asymptotic normality]
  \label{theorem:normality}
  Let $\mathcal{P}$ denote a class of distributions of
  $(Y, L, W)\in\mathbb{R}\times\mathbb{R}\times\mathcal{W}$ and assume
  we are in one of the following settings.
  \begin{enumerate}[(a)]
  \item \label{theorem:normality_directional} \textit{Nonparametric directional
    effect estimation:} Assumption~\ref{ass:tau_estimation} and
    $(\widehat{\beta}, \widehat{\sigma}^2_{\beta}):=(\widehat{\tau},
    \widehat{\sigma}^2_{\tau})$ from Algorithm~\ref{alg:nonparametric_tau}.
  \item \label{theorem:normality_binary} \textit{Nonparametric binary effect
    estimation:} Assumption~\ref{ass:lambda_estimation} and $(\widehat{\beta},
    \widehat{\sigma}^2_{\beta}):=(\widehat{\lambda},
    \widehat{\sigma}^2_{\lambda})$ from
    Algorithm~\ref{alg:nonparametric_lambda}.
  \item \label{theorem:normality_plm} \textit{Effect estimation using a
    partially linear model:} Assumption~\ref{ass:dml_estimation} and
    $(\widehat{\beta}, \widehat{\sigma}^2_{\beta}):=(\widehat{\theta},
    \widehat{\sigma}^2_{\theta})$ from Algorithm~\ref{alg:dml}.
  \end{enumerate}
  Then, letting $\Phi$ denote the cumulative distribution function of the
  standard normal distribution, it holds that
  \[
    \lim_{n \to \infty} \sup_{P \in \mathcal{P}} \sup_{x
      \in\mathbb{R}} \left|\mathbb{P}_P\left( \frac{\sqrt{n}}{
          \widehat{\sigma}_{\beta}}(\widehat{\beta}-\beta_P) \leq x
      \right) - \Phi\left(x \right)\right| = 0.
  \]
\end{theorem}
The result is proven in Section~\ref{sec:main_proofs} below. In the following
three sections, we go over each of the settings (a), (b) and (c) individually
and provide the required assumptions and detailed algorithms for the target
parameter and variance. For the remainder of this section, we omit
$\psi$-subscripts for simplicity and use the following notation to help us
express the uniformly asymptotic results rigorously. Let $(X_{P, n})_{n \in
\mathbb{N}, P \in \mathcal{P}}$ denote a family of sequences of random variables
indexed by $\mathcal{P}$. For a nonnegative sequence $(a_n)_{n \in \mathbb{N}}$
we write $X_{P, n} = o_{\mathcal{P}}(a_n)$ if for all $\epsilon > 0$ it holds
that
\[
  \lim_{n \to \infty} \sup_{P \in \mathcal{P}} \mathbb{P}_P(|X_{P,n}| \geq a_n \epsilon) = 0,
\]
that is, $X_{P, n}/a_n$ converges to $0$ in probability uniformly over
$\mathcal{P}$. We also write $X_{P, n} = O_{\mathcal{P}}(a_n)$ if for
all $\epsilon > 0$ there exist $M_\epsilon > 0$ and
$N_\epsilon \in \mathbb{N}$ such that
\[
  \sup_{n \geq N_\epsilon} \sup_{P \in \mathcal{P}} \mathbb{P}_P(|X_{P, n}| \geq a_n M) < \epsilon,
\]
that is, $X_{P, n}/a_n$ is bounded in probability for large $n$. When we write
conditional expectations given a fitted regression function $\widehat{f}$ we
mean a conditional expectation given the random variables used to fit
$\widehat{f}$ and any additional randomness involved in the fitting.

\subsubsection{Nonparametric estimation for directional effects}

We begin with the nonparametric estimation of a directional perturbation effect
$\tau$ corresponding to
Theorem~\ref{theorem:normality}~\eqref{theorem:normality_directional}. A
high-level overview of this estimation was given in
Section~\ref{sec:estimation_np_directional}. More formally, assume we have a
dataset of $n$ i.i.d.\ copies of $(Y, L, W)$, a differentiable regression method
$\widehat{f}$ for the conditional expectation $(\ell,
w)\mapsto\E[Y\given L=\ell, W=w]$, a score estimation method
$\widehat{\rho}$ for the score $(\ell,w)\mapsto \partial_{\ell}\log(p(\ell\given
w))$, where $p$ denotes the conditional density of $L$ given $W$ with
corresponding nuisance function estimation method $\widehat{\eta}$ (in our
experiments, we assume a location-scale model and thus the nuisance functions
are estimators of the conditional mean and variance of $L$ given $W$, see
Section~\ref{sec:loc-scale}). We can then compute the estimator $\widehat{\tau}$
of $\tau$ and its corresponding variance estimate $\widehat{\sigma}_{\tau}^2$ as
described in Algorithm~\ref{alg:nonparametric_tau} using a user-specified number
of splits $K$ in the cross-validation.

\begin{algorithm}[H]
  \caption{Nonparametric estimation of $\tau$}
  \begin{algorithmic}[1]
    \Statex \textbf{Input:} Data $(Y_i, L_i, W_i)_{i \in [n]}$,
    number of folds $K$, differentiable regression method $\widehat{f}$, score 
    estimation method $\widehat{\rho}$ and nuisance function estimate $\widehat{\eta}$.
    \State Split indices $\{1, \dots, n\}$ into $K$ disjoint sets $I_1, \dots, I_K$ of (roughly) equal size.
    \State For $k \in [K]$, split $I_k$ into $I_{k, 0}$ and $I_{k, 1}$ of (roughly) equal size and set $I_{-k} := \cup_{k' \neq k} I_{k'}$.
    \For{$k =1, \dots, K$}
      \State Regress $Y$ on $(L, W)$ using $I_{-k}$ yielding regression estimate $\widehat{f}_k$. 
      \State Fit nuisance function estimates $\widehat{\eta}_k$ using $I_{-k}$.
      \For{$r=0, 1$}
        \State Fit $\widehat{\rho}_{k,r}$ using $I_{k, (1-r)}$ and the nuisance function estimate $\widehat{\eta}_k$. 
        \State Set $\widehat{\tau}_{k,r} := |I_{k,r}|^{-1} \sum_{i\in I_{k, r}} \partial_{\ell} \widehat{f}_k(L_i, W_i) - \widehat{\rho}_{k, r}(L_i, W_i)\{Y_i - \widehat{f}_k(L_i, W_i) \}$.
        \State Set $\widehat{\nu}_{k,r} := |I_{k,r}|^{-1} \sum_{i\in I_{k,r}} \left(\partial_{\ell} \widehat{f}_k(L_i, W_i) - \widehat{\rho}_{k, r}(L_i, W_i)\{Y_i - \widehat{f}_k(L_i, W_i) \} \right)^2$.
      \EndFor
    \EndFor
    \Statex \textbf{return} $\widehat{\tau} := \frac{1}{2K}
    \sum_{k=1}^K \sum_{r=0}^1 \widehat{\tau}_{k,r}$ and
    $\widehat{\sigma}^2_{\tau} := \frac{1}{2K} \sum_{k=1}^K
    \sum_{r=0}^1
    \widehat{\nu}_{k,r}-\widehat{\tau}^2$.
  \end{algorithmic}
  \label{alg:nonparametric_tau}
\end{algorithm}

To obtain the result in Theorem~\ref{theorem:normality}~(a), we require that the
estimation methods $\widehat{f}$ and $\widehat{\rho}$ converge sufficiently fast
and that the class of distributions $\mathcal{P}$ satisfies several regularity
conditions.

\begin{assumption}
  \label{ass:tau_estimation}
  Assume that for all $P \in \mathcal{P}$ and $w \in \mathcal{W}$, the
  conditional distribution of $L$ given $W = w$ has a density with
  respect to Lebesgue measure and denote it by $p_P(\cdot\given
  w)$. For all $P \in \mathcal{P}$ define
  $f_P : \phi(\mathcal{Z}) \to \mathbb{R}$ for all
  $(\ell, w)\in\phi(\mathcal{Z})$ by
  $f_P(\ell, w) := \E_P[Y \given L=\ell, W=w]$ and assume that for
  almost all $(\ell, w)\in\phi(\mathcal{Z})$ both $p_P(\cdot\given w)$
  and $f_P(\cdot, w)$ are differentiable at $\ell$. Define
  $\tau_P := \E_P[\partial_\ell f_P(\ell, W) \bigm|_{\ell = L}]$ and
  $\rho_P : \phi(\mathcal{Z}) \to \mathbb{R}$ for all
  $(\ell, w)\in\phi(\mathcal{Z})$ by
  $\rho_P(\ell, w) := \partial_\ell \log(p_P(\ell \given w))$.

  Let $(Y_1, L_1, W_1), \dots, (Y_{3n}, L_{3n}, W_{3n})$ denote $3n$ i.i.d.\
  copies of $(Y, L, W)$ and define the training datasets
  $\mathcal{A}:=\{(Y_{n+1}, L_{n+1}, W_{n+1}), \dots, (Y_{2n}, L_{2n},
  W_{2n})\}$ and $\mathcal{B}:=\{(Y_{2n+1}, L_{2n+1}, W_{2n+1})$, $\dots$,
  $(Y_{3n},  L_{3n}, W_{3n})\}$. Let $\widehat{f}$ denote a differentiable
  estimator of $f_P$ based on $\mathcal{A}$, let $\widehat{\eta}$ denote an
  estimator of a nuisance function relevant for the estimation of $\rho_P$ based
  on $\mathcal{A}$ and let $\widehat{\rho}$ denote an estimator of $\rho_P$
  based on $\mathcal{B}$ and using the nuisance function estimate
  $\widehat{\eta}$. Finally, assume there exist constants $\delta, c, C > 0$
  such that
  \begin{enumerate}[(a)]
  \item \label{ass:tau_estimation_moments}
    for all $P \in \mathcal{P}:\quad$
    $\max\{\E_P[|Y|^{4+\delta}], \E_P[|\rho_P(L, W)|^{4+\delta}],
    \E_P[|\partial_\ell f_P(L, W)|^{2+\delta}], \var_{P}(Y \given L, W)\}
    \leq C$
  \item \label{ass:tau_estimation_var_bound} for all
    $P \in \mathcal{P}:\quad$
    $\min\{\var_P(Y \given L, W), \E_P[\rho_P(L, W)^{2}]\} \geq c$
  \item \label{ass:tau_estimation_f_conv}
    $\E_P\left[ \left( \widehat{f}(L_1, W_1) - f_P(L_1, W_1) \right)^{4}
      \given \,\widehat{f}\,\right] = o_{\mathcal{P}}(1)$
  \item \label{ass:tau_estimation_df_conv}
    $\E_P\left[ \left( \partial_\ell \widehat{f}(L_1, W_1) -
        \partial_\ell f_P(L_1, W_1) \right)^2 \given \,\widehat{f}\,
    \right] = o_{\mathcal{P}}(1)$
  \item \label{ass:tau_estimation_rho_conv}
    $\frac{1}{n} \sum_{i=1}^n \left( \widehat{\rho}(L_i, W_i) -
      \rho_P(L_i, W_i) \right)^2 = o_{\mathcal{P}}(1)$
  \item \label{ass:tau_estimation_prod_error}
    $\left(\frac{1}{n} \sum_{i=1}^n \left\lbrace \widehat{f}(L_i, W_i) -
        f_P(L_i, W_i) \right\rbrace^2\right) \left(\frac{1}{n} \sum_{i=1}^n
      \left\lbrace \widehat{\rho}(L_i, W_i) - \rho_P(L_i, W_i) \right\rbrace^2
    \right) = o_{\mathcal{P}}(n^{-1})$
  \end{enumerate}
\end{assumption}

\subsubsection{Nonparametric estimation for binary effects}
We now present the algorithm and conditions required for the nonparametric
estimation of the binary perturbation effect $\lambda$ corresponding to
Theorem~\ref{theorem:normality}~\eqref{theorem:normality_binary}. A high-level
overview of this estimation was provided in
Section~\ref{sec:estimation_np_binary}. More formally, assume we have a dataset
of $n$ i.i.d.\ copies of $(Y, L, W)$, a regression method $\widehat{f}$ for the
conditional expectation $(\ell, w) \mapsto \E[Y \given L=\ell, W=w]$ and
propensity score estimator $\widehat{\pi}$ for the propensity score $w \mapsto
\mathbb{P}(L = 1 \given W=w)$. We can then compute the estimator
$\widehat{\lambda}$ of $\lambda$ and its corresponding variance estimate
$\widehat{\sigma}^2_\lambda$ as described in
Algorithm~\ref{alg:nonparametric_lambda} using a user-specified number of splits
$K$ in the cross-validation.

\begin{algorithm}[H]
\caption{Nonparametric estimation of $\lambda$}
\begin{algorithmic}[1]
  \Statex \textbf{Input:} Data $(Y_i, L_i, W_i)_{i \in [n]}$, number of folds $K$, regression method $\widehat{f}$ and propensity score estimator $\widehat{\pi}$.
  \State Split indices $\{1, \dots, n\}$ into $K$ disjoint sets $I_1, \dots, I_K$ of (roughly) equal size.
  \State For $k \in [K]$, set $I_{-k} := \cup_{k' \neq k} I_{k'}$.
  \For{$k =1, \dots, K$}
    \State Regress $Y$ on $(L, W)$ using $I_{-k}$ yielding regression estimate $\widehat{f}_k$. 
    \State Regress $L$ on $W$ using $I_{-k}$ yielding propensity score function $\widehat{\pi}_k$.
    \State Set $\widehat{\kappa}_k := |I_k|^{-1} \sum_{i\in I_k}  \widehat{f}_k(1, W_i) - Y_i + \frac{Y_i-\widehat{f}_k(L_i,W_i)}{\widehat{\pi}_k(W_i)} L_i$.
    \State Set $\widehat{\nu}_k:= |I_k|^{-1} \sum_{i\in I_k} \left(\widehat{f}_k(1, W_i) - Y_i + \frac{Y_i-\widehat{f}_k(L_i,W_i)}{\widehat{\pi}_k(W_i)} L_i\right)^2$.
  \EndFor
  \State Set $\widehat{p} := \frac{1}{n} \sum_{i=1}^n 1-L_i$.
  \State Set $\widehat{\kappa} := \frac{1}{K} \sum_{k=1}^K \widehat{\kappa}_k$.
  \Statex \textbf{return} $\widehat{\lambda} := \widehat{p}^{-1} \widehat{\kappa}$ and $\widehat{\sigma}^2_{\lambda} := \widehat{p}^{-2} \left(\frac{1}{K} \sum_{k=1}^K \widehat{\nu}_k - \widehat{\kappa}^2 \right) - \widehat{p}^{-3}\widehat{\kappa}^2 (1-\widehat{p})$. 
\end{algorithmic}
\label{alg:nonparametric_lambda}
\end{algorithm}

To obtain the result in
Theorem~\ref{theorem:normality}~\eqref{theorem:normality_binary}, we require
that the estimation methods $\widehat{f}$ and $\widehat{\pi}$ converge
sufficiently fast and that the class of distributions $\mathcal{P}$ satisfies
several regularity conditions.

\begin{assumption}
\label{ass:lambda_estimation}
For all $P \in \mathcal{P}$ define $f_P: \{0, 1\} \times \mathcal{W} \to
\mathbb{R}$ and $\pi_P: \mathcal{W} \to [0, 1]$ for all $(\ell, w) \in \{0, 1\}
\times \mathcal{W}$ by $f_P(\ell, w) := \E_P[Y \given L=\ell, W=w]$ and
$\pi_P(w):= \mathbb{P}_P(L=1 \given W=w)$, respectively. Define $p_P :=
\mathbb{P}_P(L=0)$ and $\lambda_P := p_P^{-1} \left(\E_P[\E_P[Y \given L=1, W]]
- \E_P[Y] \right)$.

Let $(Y_1, L_1, W_1), \dots, (Y_{2n}, L_{2n}, W_{2n})$ denote $2n$ i.i.d.\
copies of $(Y, L, W)$ and define the training dataset $\mathcal{A} :=
\{(Y_{n+1}, L_{n+1}, W_{n+1}), \dots, (Y_{2n}, L_{2n}, W_{2n})\}$. Let
$\widehat{f}$ denote an estimator of $f_P$ based on $\mathcal{A}$ and let
$\widehat{\pi}$ denote an estimator of $\pi_P$ also based on $\mathcal{A}$.
Finally, assume there exist constants $\delta, c, C > 0$ such that
\begin{enumerate}[(a)]
  \item\label{ass:lambda_estimation_upper_bounds} for all $P \in
  \mathcal{P}:\quad$ $\max\{\E_P \left[|Y|^{2+\delta} \right], \var_P(Y \given
  W, L)\} \leq C$
  \item \label{ass:lambda_estimation_lower_bound}  for all $P \in
  \mathcal{P}:\quad$ $\min\{\var_P(Y \given W, L), p_P, 1-\pi_P(W), \pi_P(W),
  \widehat{\pi}(W)\} \geq c$ 
  \item \label{ass:lambda_estimation_pi_conv} $\frac{1}{n}\sum_{i=1}^n
    (\widehat{\pi}(W_i) - \pi_P(W_i) )^2 =
    o_{\mathcal{P}}(1)$ 
  \item \label{ass:lambda_estimation_f_conv} $\frac{1}{n}\sum_{i=1}^n
    (\widehat{f}(1, W_i) - f_P(1, W_i) )^2 = o_{\mathcal{P}}(1)$ 
  \item \label{ass:lambda_estimation_product_conv} $\left(\frac{1}{n}\sum_{i=1}^n \left\lbrace \widehat{f}(L_i, W_i) - f_P(L_i, W_i) \right\rbrace^2 \right) \left( \frac{1}{n}\sum_{i=1}^n \left\lbrace \widehat{\pi}(W_i) - \pi_P(W_i) \right\rbrace^2 \right)  = o_{\mathcal{P}}(n^{-1})$
\end{enumerate}
\end{assumption}

\subsubsection{Partially linear model for binary/directional effects}
Finally, we present the algorithm and conditions required for the estimation of
the perturbation effects under a partially linear model assumption corresponding
to Theorem~\ref{theorem:normality}~\eqref{theorem:normality_plm}. A high-level
overview of this estimation was provided in Section~\ref{sec:estimation_plm}.
More formally, assume we have a dataset of $n$ i.i.d.\ copies of $(Y, L, W)$,
regression methods $\widehat{m}$ and $\widehat{g}$ for the conditional
expectations $w \mapsto \E[Y \given W=w]$ and $w \mapsto \E[L \given W=w]$,
respectively. We can then compute the estimator $\widehat{\theta}$ of $\theta$
and its corresponding variance estimate $\widehat{\sigma}^2_\theta$ as described
in Algorithm~\ref{alg:dml} using a user-specified number of splits $K$ in the
cross-validation.

\begin{algorithm}[H]
\caption{Partially linear estimation of $\theta$ (where
    $\theta=\lambda$ or $\theta=\tau$)}
\begin{algorithmic}[1]
\Statex \textbf{Input:} Data $(Y_i, L_i, W_i)_{i \in [n]} $, number of folds $K$,
 regression methods $\widehat{g}$ and $\widehat{m}$.
\State Split indices $\{1, \dots, n\}$ into $K$ disjoint sets $I_1, \dots, I_K$ of (roughly) equal size.
\State For $k \in [K]$, set $I_{-k} := \cup_{k' \neq k} I_{k'}$.
\For{$k =1, \dots, K$}
    \State Regress $Y$ on $W$ using $I_{-k}$ yielding regression estimate $\widehat{g}_k$.
    \State Regress $L$ on $W$ using $I_{-k}$ yielding regression estimate $\widehat{m}_k$.
    \State Set $\widehat{J}_k := |I_k|^{-1} \sum_{i \in I_k} \{L_i - \widehat{m}_k(W_i) \}^2$.
    \State Set $\widehat{\kappa}_k := |I_k|^{-1} \sum_{i\in I_k} \{Y_i -
      \widehat{g}_k(W_i)\} \{L_i - \widehat{m}_k(W_i) \}$.
\EndFor
\State Set $\widehat{J} := \frac{1}{K} \sum_{k=1}^K \widehat{J}_k$
\State Set $\widehat{\theta} := \widehat{J}^{-1} \left( \frac{1}{K} \sum_{k=1}^K \widehat{\kappa}_k \right)$.
\For{$k=1, \dots, K$}
    \State Set $\widehat{\nu}_k := |I_k|^{-1} \sum_{i\in I_k} \left( \{Y_i -
    \widehat{g}_k(W_i)\} \{L_i - \widehat{m}_k(W_i) \} - \widehat{\theta} \{L_i - \widehat{m}_k(W_i) \}^2 \right)^2$
\EndFor
\Statex \textbf{return} $\widehat{\theta}$ and $\widehat{\sigma}^2_{\theta} := \widehat{J}^{-2} \left( \frac{1}{K} \sum_{k=1}^K \widehat{\nu}_k \right)$.
\end{algorithmic}
\label{alg:dml}
\end{algorithm}

To obtain the result in
Theorem~\ref{theorem:normality}~\eqref{theorem:normality_plm}, we require that
the estimation methods $\widehat{m}$ and $\widehat{g}$ converge sufficiently
fast and that the class of distributions $\mathcal{P}$ satisfies several
regularity conditions.

\begin{assumption}
  \label{ass:dml_estimation}
Assume that for all $P \in \mathcal{P}$ there exist $\theta_P \in \mathbb{R}$
and $h_P : \mathcal{W} \to \mathbb{R}$ such that
\[
  \E_P[Y \given L, W] = \theta_P L + h_P(W).
\]
For all $P \in \mathcal{P}$ define $m_P: \mathcal{W} \to \mathbb{R}$ and $g_P:
\mathcal{W} \to \mathbb{R}$ for all $w \in \mathcal{W}$ by $m_P(w):= \E_P[L
\given W=w]$ and $g_P(w):=\E_P[Y \given W=w] = h_P(w) + \theta_P m_P(w)$,
respectively. Define $\varepsilon := Y - \E_P[Y \given L, W]$ and $\xi := L -
\E_P[L \given W]$.

Let $(Y_1, L_1, W_1), \dots, (Y_{2n}, L_{2n}, W_{2n})$ denote $2n$ i.i.d.\
copies of $(Y, L, W)$ and define the training dataset $\mathcal{A} :=
\{(Y_{n+1}, L_{n+1}, W_{n+1}), \dots, (Y_{2n}, L_{2n}, W_{2n})\}$.  Let
$\widehat{m}$ denote an estimator of $m_P$ based on $\mathcal{A}$ and let
$\widehat{g}$ denote an estimator of $g_P$ also based on $\mathcal{A}$. Suppose
that there exist constants $\delta, c, C > 0$ such that
\begin{enumerate}[(a)]
  \item \label{ass:dml_upper_bounds} for all $P \in \mathcal{P}:\quad$ $\max\{\E_P\left[|\varepsilon\xi|^{2+\delta} \right] \E_P\left[ |\xi|^{4} \right], \var_P(Y \given W), \var_P(L \given W) \} \leq C$
  \item \label{ass:dml_lower_bounds} for all $P \in \mathcal{P}:\quad$ $\min\{\E_P\left[\xi^2\right], \E_P\left[\varepsilon^2 \xi^2 \right] \} \geq c$
  \item \label{ass:dml_estimation_m_conv} $\frac{1}{n}\sum_{i=1}^n (\widehat{m}(W_i) - m_P(W_i) )^2 = o_{\mathcal{P}}(n^{-1/2})$
  \item \label{ass:dml_estimation_m_conv_4} $\frac{1}{n}\sum_{i=1}^n (\widehat{m}(W_i) - m_P(W_i) )^4 = o_{\mathcal{P}}(1)$
  \item \label{ass:dml_estimation_g_conv} $\frac{1}{n}\sum_{i=1}^n (\widehat{g}(W_i) - g_P(W_i) )^2 = o_{\mathcal{P}}(1)$
  \item \label{ass:dml_estimation_product_conv} $\left(\frac{1}{n}\sum_{i=1}^n \left\lbrace\widehat{m}(W_i) - m_P(W_i) \right\rbrace^2 \right) \left( \frac{1}{n}\sum_{i=1}^n \left\lbrace \widehat{g}(W_i) - g_P(W_i) \right\rbrace^2 \right)  = o_{\mathcal{P}}(n^{-1})$
\end{enumerate}
\end{assumption}

\subsection{Proof of Theorem~\ref{theorem:normality}}
\label{sec:main_proofs}
In this section, we prove Theorem~\ref{theorem:normality} and some
additional auxiliary results.

\subsubsection[Proof of (a)]{Proof of
\eqref{theorem:normality_directional}}
\begin{proof}
We define for all $k \in [K]$ and $r \in \{0, 1\}$, $N_{k,r} := |I_{k, r}|$ and for all $i \in [n]$, $\varepsilon_i := Y_i - f(L_i, W_i)$. We first argue
that
\begin{equation}
  \label{eq:tau_numerator_decomposition}
  \widehat{\tau} = \frac{1}{n} \sum_{i=1}^n \underbrace{\partial_{\ell} f_P(L_i, W_i) - \rho_P(L_i, W_i)\varepsilon_i}_{S_i} + o_{\mathcal{P}}(n^{-1/2}).
\end{equation}
To that end, we write for all $k \in [K]$ and $r \in \{0, 1\}$
\begin{align*}
  \widehat{\tau}_{k, r} &= \frac{1}{N_{k, r}} \sum_{i \in I_{k, r}} S_i\\
  &+ \underbrace{\frac{1}{N_{k,r}}  \sum_{i \in I_{k, r}}  \{\widehat{\rho}_{k, r}(L_i, W_i)- \rho_P(L_i, W_i)\}\{\widehat{f}_k(L_i, W_i) - f_P(L_i, W_i)\} }_{\RN{1}}\\
  &- \underbrace{\frac{1}{N_{k,r}}  \sum_{i \in I_{k, r}}  \{\widehat{\rho}_{k, r}(L_i, W_i)- \rho_P(L_i, W_i)\}\varepsilon_i}_{\RN{2}}\\
  & + \underbrace{\frac{1}{N_{k,r}}  \sum_{i \in I_{k, r}}\{\partial_{\ell}\widehat{f}_k(L_i, W_i) - \partial_\ell f_P(L_i, W_i)\} + \rho_P(L_i, W_i)\{\widehat{f}_k(L_i, W_i) - f_P(L_i, W_i)\}}_{\RN{3}}.
\end{align*}
By the Cauchy--Schwarz inequality, Assumption~\ref{ass:tau_estimation}~\eqref{ass:tau_estimation_prod_error} and Lemma~\ref{lemma:o_P-continuous-transformation}, we have
\[
  |\RN{1}| \leq \sqrt{\frac{1}{N_{k,r}}  \sum_{i \in I_{k, r}}  \{\widehat{\rho}_{k, r}(L_i, W_i)- \rho_P(L_i, W_i)\}^2 \frac{1}{N_{k,r}}  \sum_{i \in I_{k, r}}  \{\widehat{f}_k(L_i, W_i) - f_P(L_i, W_i)\}^2 } = o_{\mathcal{P}}(n^{-1}).
\]
Since the summands of $\RN{2}$ are conditionally i.i.d. and mean zero given $(L_i, W_i)_{i \in I_{k, r}}$ and $\widehat{\rho}_{k, r}$, we have by conditional Chebyshev's inequality that for any $\epsilon > 0$
\[
  \mathbb{P}_P(|\RN{2}| \geq N_{k, r}^{-1/2} \epsilon \given (L_i, W_i)_{i \in I_{k, r}}, \widehat{\rho}_{k, r}) \leq \frac{1}{N_k \epsilon^2} \sum_{i \in I_{k, r}}  \{\widehat{\rho}_{k, r}(L_i, W_i)- \rho_P(L_i, W_i)\}^2 \var_P(Y_i \given L_i, W_i).
\]
Combining the above with Assumption~\ref{ass:tau_estimation}~\eqref{ass:tau_estimation_moments} and \eqref{ass:tau_estimation_rho_conv}, we have 
\[
  \mathbb{P}_P(|\RN{2}| \geq N_{k, r}^{-1/2} \epsilon \given (L_i, W_i)_{i \in I_{k, r}}, \widehat{\rho}_{k, r}) = o_{\mathcal{P}}(1)
\]
which implies $\RN{2} = o_{\mathcal{P}}(N_k^{-1/2})$ as $\epsilon$ was arbitrary.

To deal with the $\RN{3}$ term, let $i \in I_{k, r}$ be fixed and note that by the triangle inequality, the Cauchy--Schwarz inequality for conditional expectations and independence of $(L_i, W_i)$ from $\widehat{f}_k$, we have
\begin{align*}
  &\E_P\left[ \left|\partial_{\ell}\widehat{f}_k(L_i, W_i) - \partial_\ell f_P(L_i, W_i) + \rho_P(L_i, W_i)\{\widehat{f}_k(L_i, W_i) - f_P(L_i, W_i)\}\right| \given \widehat{f}_k\right]\\
  &\leq \E_P[|\partial_\ell \widehat{f}_k(L_i, W_i)- \partial_\ell f_P(L_i, W_i)| \given \widehat{f}_k] + \E_P\left[\rho_P(L_i, W_i)^2 \right]^{1/2} \E_P\left[\{\widehat{f}_k(L_i, W_i) - f_P(L_i, W_i)\}^2 \given \widehat{f}_k\right]^{1/2}.
\end{align*}
By the above,
Assumption~\ref{ass:tau_estimation}~\eqref{ass:tau_estimation_f_conv},
\eqref{ass:tau_estimation_df_conv} and
\eqref{ass:tau_estimation_moments},
Lemma~\ref{lemma:o_P-O_P-calculus}~\eqref{lemma:o_P-O_P-calculus-1}
and Jensen's inequality, we have that
\begin{align*}
  \E_P\left[ \left|\partial_{\ell}\widehat{f}_k(L_i, W_i) - \partial_\ell f_P(L_i, W_i) + \rho_P(L_i, W_i)\{\widehat{f}_k(L_i, W_i) - f_P(L_i, W_i)\}\right| \given \widehat{f}_k\right] &= O_{\mathcal{P}}(1),\\
  \E_P\left[\left|\widehat{f}_k(L_i, W_i)-f_P(L_i, W_i)\right| \given \widehat{f}_k\right] &= O_{\mathcal{P}}(1).
\end{align*}
This implies that for any $\zeta > 0$, we can find $M > 0$ and $n \in \mathbb{N}$ so that 
\begin{align*}
  \Omega_P &:= \left\lbrace \E_P\left[\left|\widehat{f}_k(L_i, W_i)-f_P(L_i, W_i)\right| \given \widehat{f}_k\right] < M \right\rbrace \\
  &\cap \left\lbrace \E_P\left[ \left|\partial_{\ell}\widehat{f}_k(L_i, W_i) - \partial_\ell f_P(L_i, W_i) + \rho_P(L_i, W_i)\{\widehat{f}_k(L_i, W_i) - f_P(L_i, W_i)\}\right| \given \widehat{f}_k\right]  < M \right\rbrace 
\end{align*}
satisfies $\mathbb{P}_P(\Omega_P) \geq 1-\zeta/2$ for all $P \in \mathcal{P}$. We therefore have for any $\epsilon > 0$
\[
  \sup_{P \in \mathcal{P}} \mathbb{P}_P(|\RN{3}| \geq N_{k,r}^{-1/2} \epsilon) \leq \sup_{P \in \mathcal{P}} \mathbb{P}_P(\{|\RN{3}| \geq N_{k,r}^{-1/2}\epsilon\} \cap \Omega_P) + \zeta/2 \leq \sup_{P \in \mathcal{P}} \mathbb{P}_P(|\ind_{\Omega_P}\cdot \RN{3}| \geq N_{k,r}^{-1/2}\epsilon) + \zeta/2,
\]
hence it suffices to prove that $\ind_{\Omega_P}\cdot\RN{3} =
o_{\mathcal{P}}(N_{k,r}^{-1/2})$ to conclude that $\RN{3} =
o_{\mathcal{P}}(N_{k,r}^{-1/2})$. To that end, we first note
that
\[
  \E_P\left[ \ind_{\Omega_P}\cdot\RN{3} \given \widehat{f}_k \right] = \E_P\left[\ind_{\Omega_P} \left(\partial_{\ell}\widehat{f}_k(L_i, W_i) - \partial_\ell f_P(L_i, W_i) + \rho_P(L_i, W_i)\{\widehat{f}_k(L_i, W_i) - f_P(L_i, W_i)\} \right) \given \widehat{f}_k\right].
\]
We will show that 
\begin{equation}
  \label{eq:integration_by_parts_eq}
  \E_P\left[ \ind_{\Omega_P} \cdot \RN{3} \given \widehat{f}_k = h\right] = 0
\end{equation}
for $\mathbb{P}_P$-almost every $h$. If $h$ is such that
$\ind_{\Omega_P}= 0$ then clearly the conditional expectation is
$0$. If $h$ is such that $\ind_{\Omega_P} = 1$, then we intend to
apply Lemma~\ref{lemma:integration_by_parts} (integration by parts)
with $g=(h - f_P)$. Assumptions (a)-(c) of
Lemma~\ref{lemma:integration_by_parts} are clear from the
above. Assumption (d) of Lemma~\ref{lemma:integration_by_parts}
  also holds, since then both $\E[|h(L, W) - f_P(L, W)| \given W]$ and
  $\E[|\partial_{\ell} (h(L, W) - f_P(L, W)) + \rho_P(L, W) (h(L, W) -
  f_P(L, W))| \given W]$ are almost surely finite when
  $\ind_{\Omega_P} = 1$, so
  Lemma~\ref{lemma:function_decay_from_integrals} yields that the
  assumption is satisfied. We conclude that
\[
  \E_P\left[ \partial_\ell(h(L, W) - f_P(L, W)) + \rho_P(L, W)(h(L, W)- f_P(L, W))\right] = 0,
\] 
and hence \eqref{eq:integration_by_parts_eq} is satisfied. By conditional
Chebyshev's inequality (applicable since we have shown that summands are
conditionally mean zero in \eqref{eq:integration_by_parts_eq}), the fact that
$(a+b)^2\leq 2a^2+2b^2$ and conditional Cauchy--Schwarz inequality, we therefore
have
\begin{align*}
  &\mathbb{P}_P(| \ind_{\Omega_P}\cdot\RN{3}| \geq N_{k, r}^{-1/2}\epsilon \given
    \widehat{f}_k)\\
  &\qquad\leq \frac{1}{\epsilon^2 N_{k, r}} \E_P[\ind_{\Omega_P}\cdot\RN{3}^2 \given \widehat{f}_k]\\
  &\qquad= \frac{1}{\epsilon^2} \E_P\left[ \left(\partial_{\ell}\widehat{f}_k(L_i, W_i) - \partial_\ell f_P(L_i, W_i) + \rho_P(L_i, W_i)\{\widehat{f}_k(L_i, W_i) - f_P(L_i, W_i)\}\right)^2 \given \widehat{f}_k\right]\\
  &\qquad\leq \frac{2}{\epsilon^2} \E_P\left[ \left(\partial_{\ell}\widehat{f}_k(L_i, W_i) - \partial_\ell f_P(L_i, W_i) \right)^2 \given \widehat{f}_k\right]\\
  &\qquad+ \frac{2}{\epsilon^2} \E_P\left[|\rho_P(L_i, W_i)|^{4}\right]^{1/2} \E_P\left[\left|\widehat{f}_k(L_i, W_i) - f_P(L_i, W_i)\right|^{4} \given \widehat{f}_k\right]^{1/2} .
\end{align*}
Assumption~\ref{ass:tau_estimation}~\eqref{ass:tau_estimation_f_conv}, \eqref{ass:tau_estimation_df_conv} and \eqref{ass:tau_estimation_moments} and Lemma~\ref{lemma:o_P-continuous-transformation} thus yield that 
\[
  \mathbb{P}_P(| \ind_{\Omega_P}\cdot\RN{3}| \geq N_{k, r}^{-1/2} \given \widehat{f}_k) = o_{\mathcal{P}}(1)
\]
which implies that $\ind_{\Omega_P} \cdot \RN{3} =
o_{\mathcal{P}}(N_{k,r}^{-1/2})$ so $\RN{3} =
o_{\mathcal{P}}(N_{k,r}^{-1/2})$ by our earlier argument. The above combined with Lemma~\ref{lemma:uneven_crossfit_split} proves \eqref{eq:tau_numerator_decomposition}.

Let $\delta' := \delta/2$. We now prove that for all $i \in [n]$ and $P \in \mathcal{P}$, we have
\begin{equation}
  \label{eq:tau_v_moment_bound}
  \E_P\left[|S_i|^{2+\delta'}\right] \leq 2^{1+\delta'} \left(C^{(2+\delta')/(2+\delta)}+ C 2^{2+\delta'} \right)  =: C'.
\end{equation}
By using that $|a+b|^{1+\delta} \leq 2^{\delta}(|a|^{2+\delta} + |b|^{2+\delta})$ for $\delta > 0$ (sometimes known as the $c_r$-inequality), we have 
\[
|S_i|^{2+\delta'} \leq 2^{1+\delta'}\left( \left| \partial_\ell f_P(L_i, W_i) \right|^{2+\delta'} + \left|\rho_P(L_i, W_i)\right|^{2+\delta} \left| \varepsilon_i \right|^{2+\delta}\right),
\]
and
\[
  |\varepsilon_i|^{4+\delta} \leq 2^{3+\delta} \left(|Y_i|^{4+\delta} + |\E_P[Y_i \given L_i, W_i]|^{4+\delta}\right).
\]
Therefore, using Assumption~\ref{ass:tau_estimation}~\eqref{ass:tau_estimation_moments}, the Cauchy--Schwarz inequality and Jensen's inequality, we have 
\begin{align*}
  \E_P\left[|S_i|^{2+\delta'}\right] &\leq 2^{1+\delta'}\left( C^{(2+\delta')/(2+\delta)} + \left(\E_P\left[\left|\rho_P(L_i, W_i)\right|^{4+\delta}\right] \E_P \left[\left| \varepsilon_i \right|^{4+\delta}\right] \right)^{1/2}  \right)\\
  & \leq 2^{1+\delta'}\left( C^{(2+\delta')/(2+\delta)} + \left(C^2 (2^{4+\delta})\right)^{1/2}  \right) = C'
\end{align*}
as desired.

We now show that 
\begin{equation}
  \label{eq:tau_denominator_decomposition}
  \widehat{\sigma}^2_{\tau} = \frac{1}{n} \sum_{i=1}^n S_i^2 - \left( \frac{1}{n} \sum_{i=1}^n S_i \right)^2 + o_{\mathcal{P}}(1).
\end{equation}
Note that $\frac{1}{n} \sum_{i=1}^n S_i = O_{\mathcal{P}}(1)$ by \eqref{eq:tau_v_moment_bound} and Markov's inequality, hence by \eqref{eq:tau_numerator_decomposition} and Lemma~\ref{lemma:o_P-continuous-transformation}, we have that 
\[
  \widehat{\tau}^2 = \left(\frac{1}{n} \sum_{i=1}^n S_i\right)^2 + o_{\mathcal{P}}(1).
\]
We have thus shown \eqref{eq:tau_denominator_decomposition} if we can show that 
\[
  \frac{1}{2K} \sum_{k=1}^K \sum_{r=0}^1 \widehat{\nu}_{k, r} = \frac{1}{n} \sum_{i=1}^n S_i^2 + o_{\mathcal{P}}(1).
\]
To that end using the same decomposition as for $\widehat{\tau}_{k, r}$, we write for $k \in [K]$, $r \in \{0, 1\}$, $i \in I_{k, r}$
\begin{align*}
  &\partial_\ell \widehat{f}_k(L_i, W_i) - \widehat{\rho}_{k, r}(L_i, W_i) \{Y_i - \widehat{f}_{k, r}(L_i, W_i)\} = \underbrace{\partial_{\ell} f_P(L_i, W_i) - \rho_P(L_i, W_i)\varepsilon_i}_{S_i}\\
  &\underbrace{\boxed{\begin{aligned}  &+ \{\widehat{\rho}_{k, r}(L_i, W_i)- \rho_P(L_i, W_i)\}\{\widehat{f}_k(L_i, W_i) - f_P(L_i, W_i)\} - \{\widehat{\rho}_{k, r}(L_i, W_i)- \rho_P(L_i, W_i)\}\varepsilon_i\\
    & + \{\partial_{\ell}\widehat{f}_k(L_i, W_i) - \partial_\ell f_P(L_i, W_i)\} + \rho_P(L_i, W_i)\{\widehat{f}_k(L_i, W_i) - f_P(L_i, W_i)\}\end{aligned}}}_{R_i}
\end{align*}
so that 
\[
  \widehat{\nu}_{k, r} = \frac{1}{N_{k, r}} \sum_{i \in I_{k, r}} S_i^2 + \frac{1}{N_{k, r}} \sum_{i \in I_{k, r}} R_i^2 + \frac{2}{N_{k, r}} \sum_{i \in I_{k, r}} S_i R_i
\]
Since $(a+b+c)^2 \leq 3 (a^2+b^2+c^2)$ (which can be seen by the Cauchy--Schwarz inequality), we have
\begin{align*}
  \frac{1}{N_{k, r}} \sum_{i \in I_{k, r}} R_i^2 &\leq 3 \biggl( \underbrace{\frac{1}{N_{k, r}} \sum_{i \in I_{k, r}} \{\widehat{\rho}_{k, r}(L_i, W_i)- \rho_P(L_i, W_i)\}^2\{\widehat{f}_k(L_i, W_i) - f_P(L_i, W_i)\}^2}_{\widetilde{\RN{1}}}\\
  &+ \underbrace{\frac{1}{N_{k, r}} \sum_{i \in I_{k, r}} \{\widehat{\rho}_{k, r}(L_i, W_i)- \rho_P(L_i, W_i)\}^2\varepsilon_i^2}_{\widetilde{\RN{2}}}\\
  &+ \underbrace{\frac{1}{N_{k, r}} \sum_{i \in I_{k, r}} \left(\{\partial_{\ell}\widehat{f}_k(L_i, W_i) - \partial_\ell f_P(L_i, W_i)\} + \rho_P(L_i, W_i)\{\widehat{f}_k(L_i, W_i) - f_P(L_i, W_i)\} \right)^2}_{\widetilde{\RN{3}}} \biggr).
\end{align*}
By using that $\sum_{i} a_i^2 b_i^2 \leq \left(\sum_{i} a_i^2 \right) \left( \sum_i b_i^2 \right)$, we have 
\[
  \widetilde{\RN{1}} \leq \left(\frac{1}{N_{k,r}}  \sum_{i \in I_{k,r}} \{\widehat{\rho}_{k, r}(L_i, W_i)- \rho_P(L_i, W_i)\}^2 \right)  \left( \frac{1}{N_k}  \sum_{i \in I_{k,r}}\left( \widehat{f}_k(L_i, W_i) - f_P(L_i, W_i) \right)^2 \right).
\]
We have already shown above as part of the argument for the $\RN{1}$ term that the upper bound is $o_{\mathcal{P}}(1)$. Similarly, each of the remaining two terms have already been shown to be $o_{\mathcal{P}}(1)$ as part of the bounds above for the corresponding terms without tildes hence we deduce that
\[
  \frac{1}{N_{k,r}} \sum_{i \in I_{k,r}} R_i^2 = o_{\mathcal{P}}(1).
\]
By the Cauchy--Schwarz inequality, we have that 
\[
  \frac{2}{N_{k,r}} \left| \sum_{i \in I_{k,r}} S_iR_i \right| \leq 2 \left( \frac{1}{N_{k,r}^2} \sum_{i \in I_{k,r}} R_i^2 \sum_{i \in I_{k,r}} S_i^2 \right)^{1/2}.
\]
Since $1/N_{k,r} \sum_{i \in I_{k, r}} S_i^2 = O_{\mathcal{P}}(1)$ by \eqref{eq:tau_v_moment_bound} and Markov's inequality and $\frac{1}{N_{k, r}} \sum_{i \in I_{k,r}} R_i^2 = o_\mathcal{P}(1)$ by the above, we conclude by Lemma~\ref{lemma:o_P-continuous-transformation} that
\[
  \frac{2}{N_{k,r}}  \sum_{i \in I_{k,r}} S_iR_i = o_\mathcal{P}(1)
\]
and Lemma~\ref{lemma:uneven_crossfit_split} concludes the proof of \eqref{eq:tau_denominator_decomposition}.

Since for all $i \in [n]$,  
\[
  \E_P[S_i] = \E_P[\partial_\ell f_P(L_i, W_i) - \rho_P(L_i, W_i)\varepsilon_i] = \tau_P,
\]
and \eqref{eq:tau_v_moment_bound} implies that $\frac{1}{\sqrt{n}}
\sum_{i=1}^n (S_i - \tau_P) = O_{\mathcal{P}}(1)$, we have by
Lemma~\ref{lemma:o_P-continuous-transformation},
\eqref{eq:tau_numerator_decomposition} and
\eqref{eq:tau_denominator_decomposition} that 
\[
  \frac{\sqrt{n}}{\widehat{\sigma}_{\tau}}(\widehat{\tau} - \tau_P) = \frac{\frac
  {1}{\sqrt{n}} \sum_{i=1}^n (S_i - \tau_P)}{ \left( \frac{1}{n} \sum_{i=1}^n S_i^2 - \left(\frac{1}{n} \sum_{i=1}^n S_i \right)^2 \right)^{1/2}} + o_{\mathcal{P}}(1).
\]
The conclusion now follows from a uniform version of Slutsky's theorem \citep[][Lemma 20]{shah2020hardness} and Lemma~\ref{lemma:self-normalised_clt}~\eqref{lemma:self-normalised_clt-1}, which we can apply by Assumption~\ref{ass:tau_estimation}~\eqref{ass:tau_estimation_moments} and \eqref{ass:tau_estimation_var_bound}, since
\begin{align*}
  \var_P(S_i) &= \E_P\left[\left(\partial_\ell f_P(L, W) - \rho_P(L, W)(Y - f_P(L, W))\right)^2\right] - \E_P\left[\partial_\ell f_P(L, W)\right]^2\\
  &\geq \E_P[\rho_P(L, W)^2(Y - f_P(L, W))^2] = \E_P[\rho_P(L, W)^2\var_P(Y \given L, W)] \geq c^2.
\end{align*}

\end{proof}

\subsubsection[Proof of (b)]{Proof of \eqref{theorem:normality_binary}}
\begin{proof}
We define for all $k \in [K]$, $N_k := |I_k|$ and for all $i \in [n]$ $\varepsilon_i := Y_i - f_P(L_i, W_i)$. We first argue that
\begin{equation}
  \label{eq:lambda_numerator_decomposition}
  \widehat{\kappa} = \frac{1}{n} \sum_{i=1}^n \underbrace{f_P(1, W_i)-Y_i + \frac{\varepsilon_i}{\pi_P(W_i)}L_i}_{S_i} + o_{\mathcal{P}}(n^{-1/2}).
\end{equation}
To that end, we write for $k \in [K]$
\begin{align*}
  \widehat{\kappa}_{k} &= \frac{1}{N_k} \sum_{i \in I_k} S_i\\
  &- \underbrace{\frac{1}{N_k}  \sum_{i \in I_k} \left( \frac{1}{\widehat{\pi}_k(W_i)} - \frac{1}{\pi_P(W_i)} \right) \left( \widehat{f}_k(L_i, W_i) - f_P(L_i, W_i) \right) L_i}_{\RN{1}}\\
  &+ \underbrace{\frac{1}{N_k}  \sum_{i \in I_k} \left( \frac{1}{\widehat{\pi}_k(W_i)} - \frac{1}{\pi_P(W_i)} \right) \varepsilon_i  L_i}_{\RN{2}}\\
  &+ \underbrace{\frac{1}{N_k}  \sum_{i \in I_k} \left(\widehat{f}_k(1, W_i) - f_P(1, W_i) \right) \left( 1- \frac{L_i}{\pi_P(W_i)} \right)}_{\RN{3}}.
\end{align*}
Applying the triangle and Cauchy--Schwarz inequality, followed by Assumption~\ref{ass:lambda_estimation}~\eqref{ass:lambda_estimation_lower_bound} and \eqref{ass:lambda_estimation_product_conv} leads to
\begin{align*}
  |\RN{1}| &\leq \sqrt{\frac{1}{N_k^2} \sum_{i \in I_k} \frac{\{\pi_P(W_i) - \widehat{\pi}_k(W_i)\}^2}{\widehat{\pi}_k(W_i)^2 \pi_P(W_i)^2}  \sum_{i \in I_k} \{\widehat{f}(L_i, W_i) - f_P(L_i, W_i) \}^2 }\\ &\leq \frac{1}{\zeta^2} \sqrt{\frac{1}{N_k^2} \sum_{i \in I_k} \{\pi_P(W_i) - \widehat{\pi}_k(W_i)\}^2 \sum_{i \in I_k} \{\widehat{f}(L_i, W_i) - f_P(L_i, W_i) \}^2 } = o_{\mathcal{P}}(N_k^{-1/2}).
\end{align*}
Since the summands of $\RN{2}$ are conditionally i.i.d.\ and mean zero given $(L_i, W_i)_{i \in I_k}$ and $\widehat{\pi}_k$, we have by conditional Chebyshev's inequality that for any $\epsilon > 0$
\begin{align*}
\mathbb{P}_P(|\RN{2}| \geq N_k^{-1/2} \epsilon \given (L_i, W_i)_{i \in I_k}, \widehat{\pi}_k) &\leq \frac{1}{N_k \epsilon^2} \sum_{i \in I_k} \left( \frac{1}{\widehat{\pi}_k(W_i)} - \frac{1}{\pi_P(W_i)} \right)^2 \var_P(Y_i \given W_i, L_i)\\
& = \frac{1}{N_k \epsilon^2} \sum_{i \in I_k} \frac{\{\pi_P(W_i) - \widehat{\pi}_k(W_i)\}^2}{\widehat{\pi}_k(W_i)^2 \pi_P(W_i)^2} \var_P(Y_i \given W_i, L_i).
\end{align*}
Combining Assumption~\ref{ass:lambda_estimation}~\eqref{ass:lambda_estimation_lower_bound}, \eqref{ass:lambda_estimation_pi_conv} and \eqref{ass:lambda_estimation_upper_bounds}, we conclude that 
\[
  \mathbb{P}_P(|\RN{2}| \geq N_k^{-1/2} \epsilon \given (L_i, W_i)_{i \in I_k}, \widehat{\pi}_k) = o_{\mathcal{P}}(1)
\]
which implies that $\RN{2} = o_{\mathcal{P}}(N_k^{-1/2})$ as $\epsilon$ was arbitrary.

The summands of $\RN{3}$ are conditionally i.i.d.\ and mean zero given $(W_i)_{i \in I_k}$ and $\widehat{f}_k$, hence by a similar argument as for the previous term, for any $\epsilon > 0$, 
\begin{align*}
  \mathbb{P}(|\RN{3}| \geq N_k^{-1/2} \epsilon \given (W_i)_{i \in I_k}, \widehat{f}_k) &\leq \frac{1}{N_k \epsilon^2} \sum_{i \in I_k} \E_P\left[ \left(1- \frac{L_i}{\pi_P(W_i)}  \right)^2 \Given W_i \right] \left( \widehat{f}_k(1, W_i) - f_P(1, W_i) \right)^2.
\end{align*}
By Assumption~\ref{ass:lambda_estimation} \eqref{ass:lambda_estimation_lower_bound}, we have 
\[
  \E_P\left[ \left(1- \frac{L_i}{\pi_P(W_i)}  \right)^2 \Given W_i \right] = \frac{1}{\pi_P(W_i)} -
  1 \leq \frac{1-c}{c},
\]
hence, by Assumption~\ref{ass:lambda_estimation} \eqref{ass:lambda_estimation_f_conv}, 
\[
  \mathbb{P}_P(|\RN{3}| \geq N_k^{-1/2}\epsilon \given (W_i)_{i \in I_k}, \widehat{f}_k) = o_{\mathcal{P}}(1)
\]
which implies that $\RN{3} = o_{\mathcal{P}}(N_k^{-1/2})$ as $\epsilon$ was arbitrary. The above combined with Lemma~\ref{lemma:uneven_crossfit_split} proves \eqref{eq:lambda_numerator_decomposition}.

We now prove that for $i \in [n]$
\begin{equation}
  \label{eq:lambda_v_moment_bound}
  \E_P\left[|S_i|^{2+\delta}\right] \leq 3^{1+\delta} C \left( 1 + c^{-(2+\delta)} (1+2^{2+\delta})   \right) =: C'.
\end{equation}
We have by repeated applications of Jensen's inequality that 
\[
|S_i|^{2+\delta} \leq 3^{1+\delta}\left(|f_P(1, W_i)|^{2+\delta} + |Y_i|^{2+\delta} + \frac{|\varepsilon_i|^{2+\delta}}{\pi_P(W_i)^{2+\delta}} \right),
\]
by Assumption~\ref{ass:lambda_estimation}~\eqref{ass:lambda_estimation_lower_bound},
\[
  |f_P(1, W_i)|^{2+\delta} = \frac{|\E_P[Y_i \ind_{\{L_i=1\}} \given W_i]|^{2+\delta}}{\pi_P(W_i)^{2+\delta}} \leq \E_P\left[|Y_i|^{2+\delta} \given W_i\right] c^{-(2+\delta)}
\]
and
\[
  |\varepsilon_i|^{2+\delta} \leq 2^{1+\delta} \left( |Y_i|^{2+\delta} + |\E_P[Y_i \given W_i, L_i]|^{2+\delta} \right).
\]
Therefore, using Assumption~\ref{ass:lambda_estimation}~\eqref{ass:lambda_estimation_upper_bounds}
\[
  \E_P\left[|S_i|^{2+\delta}\right] \leq 3^{1+\delta} \E_P\left[|Y_i|^{2+\delta} \right] \left( 1 + c^{-(2+\delta)} (1 + 2^{1+\delta})   \right) \leq C'
\]
as desired.

Define $\kappa_P := \E_P[S_i]$ and $\sigma^2_{\kappa_P} := \var_P(S_i)$. We now show that
\begin{equation}
\label{eq:lambda_denominator_decomposition}
  \frac{1}{K}\sum_{j=1}^k \widehat{\nu}_k -\widehat{\kappa}^2 = \sigma^2_{\kappa_P} + o_{\mathcal{P}}(1).
\end{equation}
By \eqref{eq:lambda_numerator_decomposition} and Chebyshev's inequality, we have that 
\[
  \widehat{\kappa} = \kappa_P + o_{\mathcal{P}}(1),
\]
thus by \eqref{eq:lambda_v_moment_bound} and Lemma~\ref{lemma:o_P-continuous-transformation}, we have 
\[
  \widehat{\kappa}^2 = \kappa_P^2 + o_{\mathcal{P}}(1).
\]
We have therefore shown \eqref{eq:lambda_denominator_decomposition} if we can show that
\begin{equation}
  \label{eq:lambda_nu_decomposition}
  \frac{1}{K} \sum_{k=1}^K \widehat{\nu}_k =  \E_P[S_i^2] + o_{\mathcal{P}}(1).
\end{equation}
To that end using the same composition as for $\widehat{\kappa}_k$, we write for $k \in [K]$, $i \in I_k$
\begin{align*}
&\widehat{f}_k(1, W_i) - Y_i + \frac{Y_i-\widehat{f}_k(L_i,W_i)}{\widehat{\pi}_k(W_i)} L_i = \underbrace{f_P(1, W_i)-Y_i + \frac{\varepsilon_i}{\pi_P(W_i)}L_i}_{S_i}\\
&\underbrace{\boxed{\begin{aligned}& -\left( \frac{1}{\widehat{\pi}_k(W_i)} - \frac{1}{\pi_P(W_i)} \right) \left( \widehat{f}_k(L_i, W_i) - f_P(L_i, W_i) \right) L_i + \left( \frac{1}{\widehat{\pi}_k(W_i)} - \frac{1}{\pi_P(W_i)} \right) \varepsilon_i  L_i\\
  &+ \left(\widehat{f}_k(1, W_i) - f_P(1, W_i) \right) \left( 1- \frac{L_i}{\pi_P(W_i)} \right)\end{aligned}}}_{R_i}
\end{align*}
so that 
\begin{align*}
&\widehat{\nu}_k = \frac{1}{N_k} \sum_{i \in I_k} S_i^2 + \frac{1}{N_k} \sum_{i \in I_k} R_i^2 + \frac{2}{N_k} \sum_{i \in I_k} S_i R_i.
\end{align*}
We have 
\begin{align*}
\frac{1}{N_k} \sum_{i \in I_k} R_i^2 &\leq 3 \biggl(\underbrace{\frac{1}{N_k}  \sum_{i \in I_k} \left( \frac{1}{\widehat{\pi}_k(W_i)} - \frac{1}{\pi_P(W_i)} \right)^2 \left( \widehat{f}_k(L_i, W_i) - f_P(L_i, W_i) \right)^2}_{\widetilde{\RN{1}}}\\
&+ \underbrace{\frac{1}{N_k}  \sum_{i \in I_k} \left( \frac{1}{\widehat{\pi}_k(W_i)} - \frac{1}{\pi_P(W_i)} \right)^2 \varepsilon_i^2 }_{\widetilde{\RN{2}}} + \underbrace{\frac{1}{N_k}  \sum_{i \in I_k} \left(1- \frac{L_i}{\pi_P(W_i)} \right)^2 \left( \widehat{f}_k(1, W_i) - f_P(1, W_i) \right)^2}_{\widetilde{\RN{3}}} \biggr).
\end{align*}
By using that $\sum_{i} a_i^2 b_i^2 \leq \left(\sum_{i} a_i^2 \right) \left( \sum_i b_i^2 \right)$, we have 
\[
  \widetilde{\RN{1}} \leq \left(\frac{1}{N_k}  \sum_{i \in I_k} \left( \frac{1}{\widehat{\pi}_k(W_i)} - \frac{1}{\pi_P(W_i)} \right)^2 \right)  \left( \frac{1}{N_k}  \sum_{i \in I_k}\left( \widehat{f}_k(L_i, W_i) - f_P(L_i, W_i) \right)^2 \right).
\]
We have already shown as part of the argument for the $\RN{1}$ term that the upper bound above is $o_{\mathcal{P}}(1)$. Similarly, each of the remaining two terms have already been shown to be $o_{\mathcal{P}}(1)$ as part of the bounds above for the corresponding terms without tildes hence we deduce that 
\[
  \frac{1}{N_k} \sum_{i \in I_k} R_i^2 = o_{\mathcal{P}}(1).
\]
By the Cauchy--Schwarz inequality, we have that 
\[
  \frac{2}{N_k} \left| \sum_{i \in I_k} S_iR_i \right| \leq 2 \left( \frac{1}{N_k^2} \sum_{i \in I_k} R_i^2 \sum_{i \in I_k} S_i^2 \right)^{1/2}.
\]
Since $1/N_k \sum_{i \in I_k} S_i^2 = O_{\mathcal{P}}(1)$ by \eqref{eq:lambda_v_moment_bound} and Markov's inequality and $\frac{1}{N_k} \sum_{i \in I_k} R_i^2 = o_\mathcal{P}(1)$ by the above, we conclude by Lemma~\ref{lemma:o_P-continuous-transformation} that
\[
  \frac{2}{N_k}  \sum_{i \in I_k} S_iR_i = o_\mathcal{P}(1).
\]
Applying Lemma~\ref{lemma:uneven_crossfit_split}, we conclude that 
\[
  \frac{1}{K} \sum_{k=1}^K \widehat{\nu}_k = \frac{1}{n} \sum_{i=1}^n S_i^2 + o_{\mathcal{P}}(1)
\]
so \eqref{eq:lambda_nu_decomposition} now holds by Chebyshev's inequality and thus also \eqref{eq:lambda_denominator_decomposition} follows.

Applying Chebyshev's inequality once more, we see that 
\begin{equation}
  \label{eq:lambda_p_decomposition}
  \widehat{p} = p_P + o_\mathcal{P}(1).
\end{equation}
Using this, Assumption~\ref{ass:lambda_estimation}~\eqref{ass:lambda_estimation_lower_bound}, \eqref{eq:lambda_v_moment_bound} and \eqref{eq:lambda_denominator_decomposition}, we have by Lemma~\ref{lemma:o_P-O_P-calculus} that
\begin{equation}
  \label{eq:lambda_variance}
  \widehat{\sigma}^2_{\lambda} = \underbrace{\frac{\sigma^2_{\kappa_P}}{p_P^2} - \frac{\kappa_P^2(1-p_P)}{p_P^3}}_{\sigma^2_{\lambda_P}} + o_{\mathcal{P}}(1).
\end{equation}
We note that by the Cauchy--Schwarz inequality
\[
  \kappa_P^2 = \E_P[S_i]^2 \leq p_P \E_P\left[ (f_P(1, W) - f(L, W))^2 \right]
\]
thus by Assumption~\ref{ass:lambda_estimation}~\eqref{ass:lambda_estimation_lower_bound}
\begin{equation}
  \label{eq:sigma_lambda_lower_bound}
  p_P \sigma^2_{\kappa_P} - \kappa_P^2(1-p_P) = p_P \E_P[S_i^2] - \kappa_P^2 \geq p_P \E_P\left[(Y-f_P(L, W))^2\left(\frac{(L-\pi(W))^2}{\pi(W)^2} \right)\right] \geq 1-c > 0.
\end{equation}
This in turn implies that $\sigma^2_{\lambda_P} \geq c^{-3}(1-c)$ so, by repeated applications of Lemma~\ref{lemma:o_P-O_P-calculus} using \eqref{eq:lambda_numerator_decomposition}, \eqref{eq:lambda_denominator_decomposition}, \eqref{eq:lambda_p_decomposition} and \eqref{eq:tau_v_moment_bound}, we have 
\begin{align*}
  \frac{\sqrt{n}}{\widehat{\sigma}_{\lambda}}(\widehat{\lambda} - \lambda_P) =   \frac{\sqrt{n}}{p_P \widehat{p} \widehat{\sigma}_{\lambda}}(\widehat{\kappa} p_P  - \widehat{p}\kappa_P ) = \frac{\sqrt{n}}{p_P^2 \sigma_{\lambda}}\left(\frac{1}{n} \sum_{i=1}^n S_i p_P - (1-L_i) \kappa_P \right) + o_{\mathcal{P}}(1).
\end{align*}
Since
\[
  \var_P(S_i p_P - (1-L_i) \kappa_P ) = \sigma^2_{\kappa_P} p_P^2 - \kappa_P (1-p_P) p_P = \sigma^2_{\lambda_P} p_P^4,
\]
and $\sigma^2_{\lambda_P} p_P^4 \geq c(1-c)$ by \eqref{eq:sigma_lambda_lower_bound}, the result now follows by Lemma~\ref{lemma:self-normalised_clt}~\eqref{lemma:self-normalised_clt-1} and a uniform version of Slutsky's theorem \citep[][Lemma 20]{shah2020hardness}.

\end{proof}

\subsubsection[Proof of (c)]{Proof of \eqref{theorem:normality_plm}}
\begin{proof}

We set $\mu_P := \E_P\left[ \xi^2 \right]$ and define for all $k \in [K]$, $N_k := |I_k|$.

We first argue that 
\begin{equation}
  \label{eq:dml_cfi_J_decomposition}
  \widehat{J} = \frac{1}{n}\sum_{i=1}^n \xi_i^2 + o_{\mathcal{P}}(n^{-1/2}).
\end{equation}
To that end, for all $k \in [K]$, we note that 
\[
  \widehat{J}_k = \frac{1}{N_k} \sum_{i\in I_k} \xi_i^2 + \underbrace{\frac{1}{N_k}\sum_{i\in I_k} \{m_P(W_i)- \widehat{m}_k(W_i)\}^2}_{\RN{1}} + \underbrace{\frac{1}{N_k} \sum_{i\in I_k} \{m_P(W_i)- \widehat{m}_k(W_i)\}\xi_i}_{\RN{2}}.
\]
By
Assumption~\ref{ass:dml_estimation}~\eqref{ass:dml_estimation_m_conv},
we have $\RN{1}=o_\mathcal{P}(N_k^{-1/2})$. Since the summands of the
third term are conditionally i.i.d. and mean zero given $\widehat{m}_k$ and $(W_i)_{i \in
  I_k}$, we have by the
conditional version of Chebyshev's inequality that
\[
  \mathbb{P}_P\left( \left|\RN{2} \right| \geq N_k^{-1/2} \epsilon \given  (W_i)_{i \in I_k}, \widehat{m}_k \right) \leq \frac{1}{N_k \epsilon^2} \sum_{i \in I_k} \{m_P(W_i)-\widehat{m}_k(W_i)\}^2 \var_P(L_i \given W_i). 
\]
Combining the above with
Assumption~\ref{ass:dml_estimation}~\eqref{ass:dml_upper_bounds}
and \eqref{ass:dml_estimation_m_conv}, we conclude that
\[
  \mathbb{P}_P\left( \left|\RN{2} \right| \geq N_k^{-1/2} \epsilon \given  (W_i)_{i \in I_k}, \widehat{m}_k \right)  = o_{\mathcal{P}}(N_k^{-1/2}) = o_{\mathcal{P}}(1)
\]
which implies that $\RN{2} = o_{\mathcal{P}}(N_k^{-1/2})$ as $\epsilon$ was arbitrary. Combining the above with Lemma~\ref{lemma:uneven_crossfit_split}, we see that \eqref{eq:dml_cfi_J_decomposition} holds.

We now argue that
\begin{equation}
  \label{eq:dml_numerator_decomposition}
  \widehat{\theta} - \theta_P = \frac{1}{n} \sum_{i=1}^n \frac{\varepsilon_i \xi_i}{\mu_P} + o_{\mathcal{P}}(n^{-1/2}).
\end{equation}
To do so, we first note that
\begin{equation}
  \label{eq:theta_bound}
  \sup_{P \in \mathcal{P}} |\theta_P| =  \sup_{P \in \mathcal{P}} \left| \frac{\E_P[\varepsilon\xi]}{\E_P\left[ \xi^2 \right]} \right| \leq \sup_{P \in \mathcal{P}} \frac{ \E_P[|\varepsilon\xi|]}{\mu_P} \leq \frac{C^{1/(2+\delta)}}{c}
\end{equation}
by Assumption~\ref{ass:dml_estimation}~\eqref{ass:dml_upper_bounds} and \eqref{ass:dml_lower_bounds}. Thus, for all $k \in [K]$, we note that, using the fact that $Y - g_P(W) = \varepsilon + \theta_P\xi$ and \eqref{eq:dml_cfi_J_decomposition},
\begin{align*}
  \widehat{\kappa}_k - \theta_P \widehat{J} &= \frac{1}{N_k} \sum_{i \in I_k} \varepsilon_i \xi_i +\underbrace{\frac{1}{N_k} \sum_{i \in I_k} \{g_P(W_i) - \widehat{g}_k(W_i)\}\{m_P(W_i) - \widehat{m}_k(W_i)\}}_{\RN{3}}\\
  &+ \underbrace{\frac{1}{N_k} \sum_{i \in I_k} \{g_P(W_i) - \widehat{g}_k(W_i)\}\xi_i}_{\RN{4}} + \underbrace{\frac{1}{N_k} \sum_{i \in I_k} \{m_P(W_i) - \widehat{m}_k(W_i)\}\{\theta_P \xi_i + \varepsilon_i\}}_{\RN{5}} + o_{\mathcal{P}}(N_k^{-1/2}).
\end{align*}
By the Cauchy--Schwarz inequality and Assumption~\ref{ass:dml_estimation}~\eqref{ass:dml_estimation_product_conv}, we have 
\[
  |\RN{3}| \leq \sqrt{\frac{1}{N_k} \sum_{i \in I_k} \{g_P(W_i) - \widehat{g}_k(W_i)\}^2 \sum_{i \in I_k} \{m_P(W_i) - \widehat{m}_k(W_i)\}^2} = o_{\mathcal{P}}(N_k^{-1/2}).
\]
We can deal with the $\RN{4}$ and $\RN{5}$ by similar arguments as for the $\RN{2}$ term above, except for the $\RN{4}$ term we condition on $(W_i)_{i \in I_k}$ and $\widehat{g}_k$ and use
Assumption~\ref{ass:dml_estimation}~\eqref{ass:dml_upper_bounds} and \eqref{ass:dml_estimation_m_conv}. Combining the convergence of the terms $\RN{3}$, $\RN{4}$ and $\RN{5}$ with Lemma~\ref{lemma:uneven_crossfit_split} yields
\begin{equation}
  \label{eq:dml_kappa_decomposition}
  \frac{1}{K} \sum_{k=1}^K \widehat{\kappa}_k - \theta_P \widehat{J} = \frac{1}{n} \sum_{i=1}^n \varepsilon_i \xi_i + o_{\mathcal{P}}(n^{-1/2}).
\end{equation}
Next, combining \eqref{eq:dml_cfi_J_decomposition}, Chebyshev's inequality and Assumption~\ref{ass:dml_estimation}~\eqref{ass:dml_upper_bounds} yields that 
\[
  \widehat{J} = \mu_P + o_{\mathcal{P}}(1),
\]
hence Assumption~\ref{ass:dml_estimation}~\eqref{ass:dml_lower_bounds} and Lemma~\ref{lemma:o_P-O_P-calculus}~\eqref{lemma:o_P-O_P-calculus-4} yields
\begin{equation}
  \label{eq:dml_J_inv_decomposition}
  \widehat{J}^{-1} = \mu_P^{-1} + o_{\mathcal{P}}(1).
\end{equation}
Combining \eqref{eq:dml_J_inv_decomposition} with
\eqref{eq:dml_kappa_decomposition} yields that
\eqref{eq:dml_numerator_decomposition} is satisfied.

Defining $\nu_P := \E_P[\varepsilon^2 \xi^2]$, we now argue that 
\begin{equation}
  \label{eq:dml_sigma_decomposition}
  \widehat{\sigma}^2_{\theta} = \frac{\nu_P}{\mu_P^2} + o_{\mathcal{P}}(1).
\end{equation}
To that end, for all $k \in [K]$, we have
\begin{equation}
  \label{eq:nu_k_decomposition}
  \begin{aligned}
  \widehat{\nu}_k &= \underbrace{\frac{1}{N_k} \sum_{i\in I_k} \{Y_i -
  \widehat{g}_k(W_i)\}^2 \{L_i - \widehat{m}_k(W_i) \}^2}_{\widehat{\nu}_k^{(1)}} +  \widehat{\theta}^2 \underbrace{\frac{1}{N_k}\sum_{i \in I_k} \{L_i - \widehat{m}_k(W_i) \}^4}_{\widehat{\nu}_k^{(2)}}\\
  &- 2 \widehat{\theta} \underbrace{\frac{1}{N_k}\sum_{i \in I_k} \{Y_i - \widehat{g}_k(W_i)\} \{L_i - \widehat{m}_k(W_i) \}^3}_{\widehat{\nu}_k^{(3)}}.
  \end{aligned}
\end{equation}
Dealing first with $\widehat{\nu}_k^{(1)}$, we decompose further and write 
\begin{align*}
  &\widehat{\nu}_k^{(1)} = \frac{1}{N_k} \sum_{i\in I_k} \varepsilon_i^2 \xi_i^2 + \frac{1}{N_k} \sum_{i\in I_k} \theta_P^2 \xi_i^4 + 2\underbrace{\frac{1}{N_k} \sum_{i\in I_k} \theta_P \varepsilon_i\xi_i^3}_{\theta_P \cdot \widetilde{\RN{1}}} + \underbrace{\frac{1}{N_k} \sum_{i\in I_k} \{\theta_P \xi_i + \varepsilon_i\}^2 \{m_P(W_i) - \widehat{m}_k(W_i)\}^2}_{\widetilde{\RN{2}}_m} \\
  & +  2\underbrace{\frac{1}{N_k} \sum_{i\in I_k} \{\theta_P \xi_i + \varepsilon_i\}^2 \xi_i \{m_P(W_i) - \widehat{m}_k(W_i)\}}_{\widetilde{\RN{3}}_m} +  \underbrace{\frac{1}{N_k} \sum_{i\in I_k} \{g_P(W_i) - \widehat{g}_k(W_i)\}^2 \xi_i^2}_{\widetilde{\RN{2}}_g}\\
  &+  \underbrace{\frac{1}{N_k} \sum_{i\in I_k} \{g_P(W_i) - \widehat{g}_k(W_i)\}^2 \{m_P(W_i) - \widehat{m}_k(W_i)\}^2}_{\widetilde{\RN{4}}} +  2\underbrace{\frac{1}{N_k} \sum_{i\in I_k} \{g_P(W_i) - \widehat{g}_k(W_i)\}^2 \xi_i \{m_P(W_i) - \widehat{m}_k(W_i)\}}_{\widetilde{\RN{5}}_g}\\
  &+  2\underbrace{\frac{1}{N_k} \sum_{i\in I_k} \{\theta_P \xi_i + \varepsilon_i\} \{g_P(W_i) - \widehat{g}_k(W_i)\} \xi_i^2}_{\widetilde{\RN{3}}_g}\\
  &+  2\underbrace{\frac{1}{N_k} \sum_{i\in I_k} \{\theta_P \xi_i + \varepsilon_i\} \{g_P(W_i) - \widehat{g}_k(W_i)\} \{m_P(W_i) - \widehat{m}_k(W_i)\}^2}_{\widetilde{\RN{5}}_m}\\
  &+  4\underbrace{\frac{1}{N_k} \sum_{i\in I_k} \{\theta_P \xi_i + \varepsilon_i\} \{g_P(W_i) - \widehat{g}_k(W_i)\} \xi_i \{m_P(W_i) - \widehat{m}_k(W_i)\}}_{\widetilde{\RN{6}}}.
\end{align*}
Note first that 
\[
  \E_P\left[ \varepsilon \xi^3 \right] =  \E_P\left[ \E_P\left[ \varepsilon \given L, W \right] \xi^3 \right] = 0.
\]
Further, the Cauchy--Schwarz inequality yields
\[
  \sup_{P \in \mathcal{P}}\E_P\left[ \left|\varepsilon\xi^3\right|^{1+\delta/6} \right] \leq \sup_{P \in \mathcal{P}}\E_P\left[ \left|\varepsilon\xi\right|^{1+\delta/2} \xi^2 \right] \leq \sup_{P \in \mathcal{P}} \E_P\left[ \left|\varepsilon\xi\right|^{2+\delta} \right] \E_P\left[ \xi^4 \right] \leq C^2,
\]
hence by a uniform version of the law of large numbers \citep[][Lemma 19]{shah2020hardness} we conclude that $\widetilde{\RN{1}} =
o_{\mathcal{P}}(1)$. Combining this with \eqref{eq:theta_bound}, we
conclude that also $\theta_P \cdot \widetilde{\RN{1}} =
o_{\mathcal{P}}(1)$. As part of the argument for $\RN{4}$ and $\RN{5}$, we have already argued that $\widetilde{\RN{2}}_m = o_{\mathcal{P}}(1)$ and $\widetilde{\RN{2}}_g = o_{\mathcal{P}}(1)$. For any $M > 0$, we have by Markov's inequality, Assumption~\ref{ass:dml_estimation}~\eqref{ass:dml_upper_bounds} and \eqref{eq:theta_bound} that
\[
  \sup_{P \in \mathcal{P}} \mathbb{P}_P\left( \left| \frac{1}{N_k} \sum_{i \in I_k} \{\theta_P \xi_i + \varepsilon_i\}^2 \xi^2 \right| \geq M \right) \leq \sup_{P \in \mathcal{P}} \frac{\theta_P^2 \E_P\left[\xi^4 \right] + \E_P\left[\varepsilon^2 \xi^2 \right]}{M} \leq \frac{1}{M} \left(\frac{C^{(4+\delta)/(2+\delta)}}{c^2} + C^{2/(2+\delta)}\right),
\]
hence
\begin{equation}
  \label{eq:theta_xi_eps_O_P}
  \frac{1}{N_k} \sum_{i \in I_k} \{\theta_P \xi_i + \varepsilon_i\}^2 \xi_i^2 = O_{\mathcal{P}}(1).
\end{equation}
By the Cauchy--Schwarz inequality, \eqref{eq:theta_xi_eps_O_P}, the fact that $\widetilde{\RN{2}}_m = o_{\mathcal{P}}(1)$ and Lemmas~\ref{lemma:o_P-O_P-calculus} and \ref{lemma:o_P-continuous-transformation}, we have
\[
  |\widetilde{\RN{3}}_m| \leq \left(\frac{1}{N_k} \sum_{i \in I_k} \{\theta_P \xi_i + \varepsilon_i\}^2 \xi^2 \right)^{1/2} \left(\widetilde{\RN{2}}_m \right)^{1/2} = o_{\mathcal{P}}(1).
\]
An analogous argument can be used to show that $\widetilde{\RN{3}}_g = o_{\mathcal
P}(1)$. Using $\sum_{i} a_i^2 a_i^2 \leq \left(\sum_{i} a_i^2 \right)
\left( \sum_{i} a_i^2 \right)$ and
Assumption~\ref{ass:dml_estimation}~\eqref{ass:dml_estimation_product_conv}, we have
\[
  \widetilde{\RN{4}} \leq \frac{1}{N_k} \sum_{i\in I_k}  \{m_P(W_i) - \widehat{m}_k(W_i)\}^2 \sum_{i\in I_k} \{g_P(W_i) - \widehat{g}_k(W_i)\}^2 = o_{\mathcal{P}}(1).
\]
By the Cauchy--Schwarz inequality, our results above and Lemmas~\ref{lemma:o_P-O_P-calculus} and \ref{lemma:o_P-continuous-transformation}, we have
\[
  |\widetilde{\RN{5}}_g| \leq \left( \widetilde{\RN{2}}_g \right)^{1/2}  \left( \widetilde{\RN{4}} \right)^{1/2} = o_{\mathcal{P}}(1).
\]
An analogous argument can be applied to show that $\widetilde{\RN{5}}_m = o_{\mathcal{P}}(1)$. Finally, again by similar arguments as for the previous terms
\[
  |\widetilde{\RN{6}}| \leq \left(  \widetilde{\RN{2}}_g  \right)^{1/2} \left(  \widetilde{\RN{2}}_m  \right)^{1/2} = o_{\mathcal{P}}(1).
\]
Putting things together, we conclude that
\begin{equation}
  \label{eq:nu_k_decomposition_1}
  \widehat{\nu}_k^{(1)} =  \frac{1}{N_k} \sum_{i\in I_k} \varepsilon_i^2 \xi_i^2 + \frac{1}{N_k} \sum_{i\in I_k} \theta_P^2 \xi_i^4 + o_{\mathcal{P}}(1).
\end{equation}
Turning now to $\widehat{\nu}_k^{(2)}$, we decompose further and
write
\begin{align*}
  \widehat{\nu}_k^{(2)} &=  \frac{1}{N_k} \sum_{i\in I_k}  \xi_i^4 + 4\underbrace{\frac{1}{N_k} \sum_{i\in I_k}  \xi_i^3 \{m_P(W_i) - \widehat{m}_k(W_i)\}}_{\widetilde{\RN{7}}_m} + 6 \underbrace{\frac{1}{N_k} \sum_{i\in I_k}  \xi_i^2 \{m_P(W_i) - \widehat{m}_k(W_i)\}^2}_{\widetilde{\RN{8}}}\\
  &+ 4 \underbrace{\frac{1}{N_k} \sum_{i\in I_k}  \xi_i \{m_P(W_i) - \widehat{m}_k(W_i)\}^3}_{\widetilde{\RN{9}}} + \underbrace{\frac{1}{N_k} \sum_{i\in I_k}  \{m_P(W_i) - \widehat{m}_k(W_i)\}^4}_{\widetilde{\RN{10}}}
\end{align*}
Note that by similar arguments as used in \eqref{eq:theta_xi_eps_O_P}, we can show that $\frac{1}{N_k} \sum_{i \in I_k} \xi_i^4 = O_{\mathcal{P}}(1)$. By the Cauchy--Schwarz inequality and Lemmas~\ref{lemma:o_P-O_P-calculus} and \ref{lemma:o_P-continuous-transformation}, we have 
\[
  |\widetilde{\RN{7}}_m| \leq \left(\frac{1}{N_k} \sum_{i=1}^n \xi_i^4 \right)^{1/2} \left( \RN{2} \right)^{1/2} = o_{\mathcal{P}}(1),
\]
since we already showed above that $\RN{2} =
o_{\mathcal{P}}(1)$. We have already argued that $\widetilde{\RN{8}} = o_{\mathcal{P}}(1)$ as part of the proof for the $\RN{2}$ term. Assumption~\ref{ass:dml_estimation}~\eqref{ass:dml_estimation_m_conv_4} yields that $\widetilde{\RN{10}} = o_\mathcal{P}(1)$ and once more applying the Cauchy--Schwarz inequality yields
\[
  |\widetilde{\RN{9}}| \leq \left( \RN{2} \right)^{1/2} \left(\widetilde{\RN{9}} \right)^{1/2} = o_{\mathcal{P}}(1).
\]
In summary, we have that 
\begin{equation}
  \label{eq:nu_k_decomposition_2}
  \widehat{\nu}_k^{(2)} = \frac{1}{N_k} \sum_{i\in I_k}  \xi_i^4 + o_{\mathcal{P}}(1).
\end{equation}

To deal with $\widehat{\nu}_k^{(3)}$, we decompose further and write
\begin{align*}
  &\frac{1}{2}\widehat{\nu}_k^{(3)} =  \frac{1}{N_k} \sum_{i\in I_k}  \theta_P\xi_i^4 + \underbrace{\frac{1}{N_k} \sum_{i\in I_k}  \varepsilon_i \xi_i^3}_{\widetilde{\RN{1}}} + \underbrace{\frac{1}{N_k} \sum_{i\in I_k}  \{g_P(W_i) - \widehat{g}_k(W_i)\}\xi_i^3}_{\widetilde{\RN{7}}_g}\\
  &+ 3 \underbrace{\frac{1}{N_k} \sum_{i\in I_k}  \{\theta_P \xi_i + \varepsilon_i\} \xi_i^2 \{m_P(W_i)-\widehat{m}_k(W_i)\}}_{\widetilde{\RN{11}}} + 3 \underbrace{\frac{1}{N_k} \sum_{i\in I_k}  \{g_P(W_i) - \widehat{g}_k(W_i)\} \xi_i^2 \{m_P(W_i)-\widehat{m}_k(W_i)\}}_{\widetilde{\RN{12}}}\\
  &+ 3 \underbrace{\frac{1}{N_k} \sum_{i\in I_k}  \{\theta_P \xi_i + \varepsilon_i\} \xi_i \{m_P(W_i)-\widehat{m}_k(W_i)\}^2}_{\widetilde{\RN{13}}} + 3 \underbrace{\frac{1}{N_k} \sum_{i\in I_k}  \{g_P(W_i) - \widehat{g}_k(W_i)\} \xi_i \{m_P(W_i)-\widehat{m}_k(W_i)\}^2}_{\widetilde{\RN{14}}}\\
  &+  \underbrace{\frac{1}{N_k} \sum_{i\in I_k}  \{\theta_P \xi_i + \varepsilon_i\}  \{m_P(W_i)-\widehat{m}_k(W_i)\}^3}_{\widetilde{\RN{15}}} + \underbrace{\frac{1}{N_k} \sum_{i\in I_k}  \{g_P(W_i) - \widehat{g}_k(W_i)\} \{m_P(W_i)-\widehat{m}_k(W_i)\}^3}_{\widetilde{\RN{16}}}
\end{align*}
We have already argued that $\widetilde{\RN{1}} = o_\mathcal{P}(1)$
above and using similar arguments as for $\widetilde{\RN{7}}_m$, we
see that $\widetilde{\RN{7}}_g = o_{\mathcal{P}}(1)$. For the
remaining terms, we use the Cauchy--Schwarz inequality and the
previous results as follows
\begin{align*}
  &|\widetilde{\RN{11}}| \leq  \left(\frac{1}{N_k} \sum_{i \in I_k} \{\theta_P \xi_i + \varepsilon_i\}^2 \xi^2 \right)^{1/2} \left(\RN{2} \right)^{1/2} = o_\mathcal{P}(1),\\
  &|\widetilde{\RN{12}}| \leq  \left(\widetilde{\RN{2}}_g \right)^{1/2} \left(\RN{2} \right)^{1/2} = o_\mathcal{P}(1),\\
  &|\widetilde{\RN{13}}| \leq  \left(\frac{1}{N_k} \sum_{i \in I_k} \{\theta_P \xi_i + \varepsilon_i\}^2 \xi^2 \right)^{1/2} \left(\widetilde{\RN{10}} \right)^{1/2} = o_\mathcal{P}(1),\\
  &|\widetilde{\RN{14}}| \leq  \left(\widetilde{\RN{2}}_g \right)^{1/2} \left(\widetilde{\RN{10}} \right)^{1/2} = o_\mathcal{P}(1),\\
  &|\widetilde{\RN{15}}| \leq  \left(\widetilde{\RN{2}}_m \right)^{1/2} \left(\widetilde{\RN{10}} \right)^{1/2} = o_\mathcal{P}(1),\\
  &|\widetilde{\RN{16}}| \leq  \left(\widetilde{\RN{4}} \right)^{1/2} \left(\widetilde{\RN{10}} \right)^{1/2} = o_\mathcal{P}(1).
\end{align*}
Putting things together, we have that 
\begin{equation}
  \label{eq:nu_k_decomposition_3}
  \widehat{\nu}_k^{(3)} = 2 \frac{1}{N_k} \sum_{i\in I_k}  \theta_P\xi_i^4 + o_\mathcal{P}(1).
\end{equation}
Combining \eqref{eq:theta_bound} and \eqref{eq:dml_numerator_decomposition}, we have that 
\[
  \widehat{\theta} = \theta_P + o_{\mathcal{P}}(1)
\]
thus by \eqref{eq:theta_xi_eps_O_P}, \eqref{eq:nu_k_decomposition}, \eqref{eq:nu_k_decomposition_1}, \eqref{eq:nu_k_decomposition_2} and \eqref{eq:nu_k_decomposition_3}, we have 
\[
  \widehat{\nu}_k = \frac{1}{N_k} \sum_{i \in I_k} \varepsilon_i^2 \xi_i^2 + o_{\mathcal{P}}(1),
\]
and hence Lemma~\ref{lemma:uneven_crossfit_split} yields that 
\[
  \frac{1}{K} \sum_{k=1}^K \widehat{\nu}_k = \frac{1}{n} \sum_{i=1}^n \varepsilon_i^2 \xi_i^2 + o_{\mathcal{P}}(1).
\]
Applying Chebyshev's inequality, we obtain 
\[
  \frac{1}{K} \sum_{k=1}^K \widehat{\nu}_k = \nu_P + o_{\mathcal{P}}(1),
\]
so by Lemma~\ref{lemma:o_P-O_P-calculus}, \eqref{eq:dml_J_inv_decomposition} and Assumption~\ref{ass:dml_estimation}~\eqref{ass:dml_upper_bounds} and \eqref{ass:dml_lower_bounds}, \eqref{eq:dml_sigma_decomposition} holds.

Finally, combining \eqref{eq:dml_numerator_decomposition}, \eqref{eq:dml_sigma_decomposition} with Lemmas~\ref{lemma:o_P-continuous-transformation} and \ref{lemma:o_P-O_P-calculus} using Assumption~\ref{ass:dml_estimation}~\eqref{ass:dml_upper_bounds} and \eqref{ass:dml_lower_bounds}, we have that 
\[
  \frac{\sqrt{n}}{\widehat{\sigma}_{\theta}}(\widehat{\theta} - \theta_P) = \frac{\frac{1}{\sqrt{n}} \sum_{i=1}^n \varepsilon_i \xi_i}{\nu_P^{1/2}} + o_{\mathcal{P}}(1).
\]
The conclusion now follows from Assumption~\ref{ass:dml_estimation}~\eqref{ass:dml_upper_bounds} and \eqref{ass:dml_lower_bounds}, Lemma~\ref{lemma:self-normalised_clt}~\eqref{lemma:self-normalised_clt-3} and a uniform version of Slutsky's theorem \citep[][Lemma 20]{shah2020hardness}.

\end{proof}

\subsubsection{Auxiliary results}

\label{sec:aux_proofs}
\begin{lemma}
\label{lemma:uneven_crossfit_split}
Let $X$ be a random variable with distribution determined by
$\mathcal{P}$ and let $(X_n)_{n \in \mathbb{N}}$ be an i.i.d. sequence
of copies of $X$ with
$\sup_{P \in \mathcal{P}}\E_P[|X|] \leq C < \infty$. Let
$K \in \mathbb{N}$, let $I_1, \dots, I_K$ be a partition of
  $[n]$ and for all $k\in[K]$ define $N_k := |I_k|$ and assume that
  $\lfloor n/K \rfloor \leq N_k \leq \lceil n/K \rceil$. Suppose that there exist random variables
$Y_1, \dots, Y_k$ and $\delta \in [0, 1)$ such that for all
  $k\in [K]$ it holds that
\[
  Y_k = \frac{1}{N_k} \sum_{i \in I_k} X_i + o_{\mathcal{P}}(N_k^{-\delta}).
\]
Then 
\[
  \frac{1}{K} \sum_{k=1}^K Y_k = \frac{1}{n}\sum_{i=1}^n X_i + o_{\mathcal{P}}(n^{-\delta}).
\]
\begin{proof}
Without loss of generality we will assume that $\lceil n/K \rceil = N_1 \geq \dots \geq N_K = \lfloor n/K \rfloor$, since otherwise we may relabel the partition so that this is satisfied.

By assumption, we can write
\[
  Y_k = \frac{1}{N_k} \sum_{i \in I_k} X_i + R_k
\]
with $R_k = o_{\mathcal{P}}(N_k^{-\delta})$ hence
\[
  \frac{1}{K} \sum_{k=1}^K Y_k = \frac{1}{n} \sum_{i=1}^n X_i  + \underbrace{\sum_{k=1}^K  \left(\frac{1}{K N_k} - \frac{1}{n}\right) \sum_{i \in I_k} X_i}_{\RN{1}} + \underbrace{\frac{1}{K} \sum_{k=1}^K R_k}_{\RN{2}},
\]
and it suffices to show that $\RN{1} = o_\mathcal{P}(n^{-\delta})$ and $\RN{2} = o_\mathcal{P}(n^{-\delta})$.

For $\RN{1}$, we first note that by the triangle inequality and Hölder's inequality, we have
\[
  |\RN{1}| \leq \sum_{k=1}^K  \left|\frac{1}{K N_k} - \frac{1}{n}\right| \sum_{i \in I_k} |X_i| \leq \max_{k \in [K]} \left|\frac{n}{K N_k} - 1\right| \frac{1}{n} \sum_{i=1}^n |X_i|.
\]
By Markov's inequality and the assumption on $\E[|X|]$, we have for any $\epsilon > 0$ that
\[
  \sup_{P \in \mathcal{P}} \mathbb{P}_P(|\RN{1}| \geq n^{-\delta} \epsilon) \leq \frac{C}{\epsilon}\max_{k \in [K]} |n-K N_k| \frac{n^{\delta}}{K N_k} .
\]
By definition of $N_k$, we have 
\[
  \frac{n}{K} - 1 \leq N_k \leq n/K + 1 \Longrightarrow  n - K \leq K N_k \leq n + K \Longrightarrow |n-N_k K| \leq K,
\]
thus, for $n \geq K+1$,
\[
  \sup_{P \in \mathcal{P}} \mathbb{P}_P(|\RN{1}| \geq n^{-\delta} \epsilon) \leq \frac{C K}{\epsilon}  \frac{n^{\delta}}{n-K}. 
\]
As the upper bound goes to $0$ as $n \to \infty$ by our assumption on $\delta$, we have $\RN{1} = o_{\mathcal{P}}(n^{-\delta})$ as desired.

For $\RN{2}$, by the triangle inequality, the union bound and using that $n \leq 2K N_k$ for $n \geq K^2$, we have 
\begin{align*}
  \sup_{P \in \mathcal{P}} \mathbb{P}_P\left( \left| \RN{2} \right| \geq n^{-\delta} \epsilon \right) \leq \sum_{k=1}^K \sup_{P \in \mathcal{P}}  \mathbb{P}_P (|R_k| \geq n^{-\delta} \epsilon) \leq \sum_{k=1}^K \sup_{P \in \mathcal{P}}  \mathbb{P}_P (|R_k| \geq (2K)^{-\delta} N_k^{-\delta}\epsilon).
\end{align*} 
We conclude that $\RN{2} = o_{\mathcal{P}}(n^{-\delta})$ by the assumption that $R_k = o_{\mathcal{P}}(N_k^{-\delta})$.
\end{proof}
\end{lemma}

\begin{lemma}
  \label{lemma:o_P-O_P-calculus}
Let $(X_n)_{n \in \mathbb{N}}$ and $(Y_n)_{n \in \mathbb{N}}$ be sequences of real-valued random variables and let $(a_n)_{n \in \mathbb{N}}$ and $(b_n)_{n \in \mathbb{N}}$ be sequences of non-negative real numbers.
\begin{enumerate}[(a)]
  \item \label{lemma:o_P-O_P-calculus-1} If $X_n = o_{\mathcal{P}}(a_n)$ then $X_n  = O_{\mathcal{P}}(a_n)$.
  \item \label{lemma:o_P-O_P-calculus-2} If $X_n = o_{\mathcal{P}}(a_n)$ and $Y_n = O_{\mathcal{P}}(b_n)$ then $X_n Y_n = o_{\mathcal{P}}(a_n b_n)$.
  \item \label{lemma:o_P-O_P-calculus-3} If $X_n = o_{\mathcal{P}}(a_n)$ and $Y_n = O_{\mathcal{P}}(a_n)$ then $X_n + Y_n = O_{\mathcal{P}}(a_n)$.
  \item \label{lemma:o_P-O_P-calculus-4} If $X_n = Y_n + o_{\mathcal{P}}(1)$ and $Y_n^{-1} = O_{\mathcal{P}}(1)$, then $X_n^{-1}=Y_n^{-1} + o_{\mathcal{P}}(1)$.
\end{enumerate}
\begin{proof}
\begin{enumerate}[(a)]
  \item  Given $\epsilon > 0$, we can set $M_\epsilon = 1$ and, since $X_n = o_{\mathcal{P}}(a_n)$, we have the existence of $N_\epsilon$ such that
  \[
    \sup_{n \geq N_\epsilon} \sup_{P \in \mathcal{P}} \mathbb{P}_P(|X_n| \geq a_n) < \epsilon.
  \]
  \item Given $\epsilon > 0$, we have for any $M > 0$ that 
  \[
    \sup_{P \in \mathcal{P}} \mathbb{P}_P(|X_n Y_n| \geq a_n b_n \epsilon) \leq  \sup_{P \in \mathcal{P}} \mathbb{P}_P(|X_n| \geq a_n\epsilon/M) +  \sup_{P \in \mathcal{P}} \mathbb{P}_P(|Y_n| \geq b_n M).
  \]
  Choosing $N$ and $M$ large enough, by assumption on $X_n$ and $Y_n$, we can make each of the two terms smaller than $\epsilon/2$ hence proving the desired result.
  \item For any $M > 0$, we have that
  \[
    \sup_{P \in \mathcal{P}} \mathbb{P}_P(|X_n + Y_n| \geq a_n M) \leq \sup_{P \in \mathcal{P}} \mathbb{P}_P(|X_n| \geq a_n M /2) + \sup_{P \in \mathcal{P}} \mathbb{P}_P(|Y_n| \geq a_n M /2).
  \]
  Since $X_n = o_{\mathcal{P}}(a_n)$, the first term goes to $0$ for any $M > 0$ and since $Y_n = O_{\mathcal{P}}(a_n)$, we can make the second term arbitrarily small by choosing $M$ large enough. This proves the desired result.  
  \item Let $R_n := X_n - Y_n$ so that $R_n = o_{\mathcal{P}}(1)$. We have
  \[
    X_n^{-1}-Y_n^{-1} = R_n Y_n^{-1} (Y_n + R_n)^{-1},
  \]
  so we are done by the previous claims if we can show that $(Y_n + R_n)^{-1} = O_\mathcal{P}(1)$. For any $M, K, \zeta > 0$ we have 
  \begin{align*}
    \sup_{P \in \mathcal{P}} \mathbb{P}_P(|Y_n + R_n|^{-1} \geq M) &\leq \underbrace{\sup_{P \in \mathcal{P}} \mathbb{P}_P(|Y_n + R_n|^{-1} \geq M, |R_n| < \zeta, |Y_n|^{-1} < K)}_{\RN{1}}\\
    &+ \underbrace{\sup_{P \in \mathcal{P}} \mathbb{P}_P(|Y_n|^{-1} \geq K)}_{\RN{2}} + \underbrace{\sup_{P \in \mathcal{P}} \mathbb{P}_P(|R_n| \geq \zeta)}_{\RN{3}}.
  \end{align*}
  Given $\epsilon > 0$, we can make $\RN{2} < \epsilon/2$ by choosing $K$ sufficiently large. Letting $\zeta := K^{-1}/2$, we can make $\RN{3} < \epsilon/2$ for $n$ large enough. Letting $M := K^{-1}/2$, $\RN{1} = 0$, since
  \[
    |Y_n + R_n| \geq |Y_n| - |R_n| < K^{-1}/2 \Longrightarrow |Y_n + R_n|^{-1} < K^{-1}/2 = M.
  \]
  As $\epsilon$ was arbitrary, this proves the desired result.
\end{enumerate}
\end{proof}
\end{lemma}

\begin{lemma} \label{lemma:o_P-continuous-transformation}
Let $(X_n)_{n \in \mathbb{N}}$ and $(Y_n)_{n \in \mathbb{N}}$ be sequences of random variables taking values in the closed set $\mathcal{X} \subseteq \mathbb{R}^d$. Suppose that $X_n = Y_n + o_{\mathcal{P}}(1)$ and let $h:\mathcal{X} \rightarrow \mathbb{R}$ be a continuous function. Suppose that at least one of the following conditions hold:
\begin{enumerate}[(a)]
  \item $h$ is uniformly continuous,
  \item $Y_n = O_{\mathcal{P}}(1)$.
\end{enumerate}
Then $h(X_n) = h(Y_n) + o_\mathcal{P}(1)$.
\end{lemma}
\begin{proof}
Let $\epsilon > 0$ be given. We need to show that 
  \[
    \lim_{n \to \infty} \sup_{P \in \mathcal{P}} \mathbb{P}_P(|h(X_n)-h(Y_n)| > \epsilon) = 0.
  \]
If $h$ is uniformly continuous, then we can find $\delta >0$ such that $|h(x) - h(y)| \leq \epsilon$ whenever $\|x-y\| \leq \delta$.  Thus,
\[
\sup_{P \in \mathcal{P}} \mathbb{P}_P(|h(X_n)-h(Y_n)| > \epsilon) \leq \sup_{P \in \mathcal{P}} \mathbb{P}_P(\|X_n-Y_n\| > \delta) \rightarrow 0
\]
as $n \to \infty$.   On the other hand, suppose now that $Y_n = O_\mathcal{P}(1)$ and let $M > 0$. Since $h$ is continuous, it is uniformly continuous on $B(0, M):= \{x \in \mathcal{X} : \|x \| \leq M \}$, so we can choose $\delta > 0$ such that  $|h(x) - h(y)| \leq \epsilon$ whenever $x,y \in B(0, M)$ satisfy $\|x-y\| \leq \delta$.  Hence, for $M > \delta$,
  \begin{align*}
  \sup_{P \in \mathcal{P}} \mathbb{P}_P(|h(X_n)-h(Y_n)| > \epsilon) &\leq \sup_{P \in \mathcal{P}} \mathbb{P}_P(\|X_n-Y_n\| > \delta)\\
  &\hspace{1cm}+ \sup_{P \in \mathcal{P}} \mathbb{P}_P(\|X_n\| \vee \|Y_n\| > M,\|X_n-Y_n\| \leq \delta) \\
  &\leq \sup_{P \in \mathcal{P}} \mathbb{P}_P(\|X_n-Y_n\| > \delta) + \sup_{P \in \mathcal{P}} \mathbb{P}_P(\|Y_n\| > M-\delta) \rightarrow 0
  \end{align*}
as $n, M \to \infty$.
\end{proof}

\begin{lemma}
  \label{lemma:self-normalised_clt}
  Let $(X_n)_{n \in \mathbb{N}}$ be a sequence of i.i.d. copies of a random variable $X$ taking values in $\mathbb{R}$. Suppose there exist constants $c, C, \delta > 0$ such that for all $P \in \mathcal{P}$ $\var_P(X) \geq c$ and $\E_P\left[|X|^{2+\delta}\right] \leq C$. Then
  \begin{enumerate}[(a)]
    \item  \label{lemma:self-normalised_clt-1}
    \[
      T := \frac{\frac{1}{\sqrt{n}}\sum_{i=1}^n (X_i - \E_P[X])}{\left( \frac{1}{n} \sum_{i=1}^n X_i^2 - \left(\frac{1}{n} \sum_{i=1}^n X_i \right)^2 \right)^{1/2}} \quad \text{satisfies} \quad \lim_{n \to \infty} \sup_{P \in \mathcal{P}} \sup_{t \in \mathbb{R}} |\mathbb{P}_P(T \leq t) - \Phi(t) | = 0
    \]
    \item \label{lemma:self-normalised_clt-3}
      \[
        T := \widetilde{T} := \frac{\frac{1}{\sqrt{n}}\sum_{i=1}^n (X_i - \E_P[X])}{\left( \var_P(X) \right)^{1/2}} \quad \text{satisfies} \quad \lim_{n \to \infty} \sup_{P \in \mathcal{P}} \sup_{t \in \mathbb{R}} |\mathbb{P}_P(\widetilde{T} \leq t) - \Phi(t) | = 0
      \]
  \end{enumerate}

\begin{proof}
Defining $\mu_P:= \E_P[X]$, we have that
\[
  \frac{1}{n} \sum_{i=1}^n X_i = \mu_P + o_\mathcal{P}(1)
\]
by the uniform law of large numbers in \citet[][Lemma 19]{shah2020hardness} and the assumption that $\E_P\left[|X|^{2+\delta}\right] \leq C$. The same argument reveals that, defining $\nu_P := \E_P\left[ X^2 \right]$,
\[
  \frac{1}{n} \sum_{i=1}^n X_i^2 = \nu_P + o_\mathcal{P}(1).
\]
We conclude by Lemma~\ref{lemma:o_P-O_P-calculus} and Lemma~\ref{lemma:o_P-continuous-transformation} that 
\[
  \left( \frac{1}{n} \sum_{i=1}^n X_i^2 - \left(\frac{1}{n} \sum_{i=1}^n X_i \right)^2 \right)^{1/2} = \sigma + o_{\mathcal{P}}(1).
\]
For any $M > 0$ and $P \in \mathcal{P}$, by Markov's inequality, Jensen's inequality and the fact that $\E_P\left[|X|^{2+\delta}\right] \leq C$, 
\[
  \mathbb{P}_P\left(\left| \frac{1}{\sqrt{n}}\sum_{i=1}^n (X_i - \E_P[X]) \right| \geq M \right) \leq \frac{1}{M^2} \E_P\left(X^2 \right) \leq \frac{1}{M^2} C^{2/(2+\delta)},
\]
hence the numerator of $T$ is $O_{\mathcal{P}}(1)$. Combining the facts above, by Lemma~\ref{lemma:o_P-O_P-calculus} since $\sigma \geq c^{1/2}$, we have that
\[
  T = \frac{1}{\sqrt{n}\sigma}\sum_{i=1}^n (X_i - \E_P[X]) + o_{\mathcal{P}}(1).
\]
The result now follows by a uniform CLT and version of Slutsky's theorem in \citet[][Lemmas 18 and 20]{shah2020hardness}.

Similar arguments can be applied to show the result for $\widetilde{T}$.
\end{proof}
\end{lemma}

\begin{lemma}
  \label{lemma:integration_by_parts}
  Let $(L, W)$ denote random variables supported on $\mathcal{S} \subseteq
  \mathbb{R} \times \mathcal{W}$ and let $P_W$ denote the distribution of $W$.
  Suppose that the conditional distribution of $L$ given $W=w$ is absolutely
  continuous with respect to Lebesgue measure with conditional density $p(\ell
  \given w)$.  Let $g: \mathcal{S} \to \mathbb{R}$ denote a function. Assume
  that for $P_W$-almost every$w \in
  \mathcal{W}$:
  \begin{enumerate}[(a)]
    \item $\partial_\ell g(\ell, w)$ exists,
    \item $\partial_\ell p(\ell \given w)$ exists,
    \item $\E\left[ | \partial_\ell g(L, W) + \rho(L, W)g(L, W) | \right] < \infty$, where $\rho(\ell, w) := \partial_\ell \log(p(\ell \given w))$,
    \item $\lim_{t \to \infty} g(t, w)p(t \given w) -  g(-t, w)p(-t \given w) = 0$.
  \end{enumerate}
  Then 
  \[
    \E\left[\partial_\ell g(L, W) + \rho(L, W)g(L, W)  \given W \right] = 0,
  \]
  and thus integration by parts is valid, that is, 
  \[
    \E\left[\partial_\ell g(L, W)\right] = - \E\left[ \rho(L, W)g(L, W)\right].
  \]
  \begin{proof}
    The following is an adaptation of \citet[][Proposition 1]{klyne2023average}
(the primary change is the slight rephrasing of assumption~(d)). We prove the
result by showing that for $P_W$-almost all $w \in \mathcal{W}$, we have
    \[
      \E\left[\partial_\ell g(L, W) + \rho(L, W)g(L, W)  \given W = w \right] = 0.
    \]
    Note that assumption~(c) and Fubini's theorem \citep[][Corollary 14.9]{schilling2017measures} implies that 
    \[
      \ell \mapsto \partial_\ell g(\ell, w) p(\ell \given w) + \rho(\ell, w)g(\ell, w) p(\ell \given w) = \partial_\ell (g(\ell, w) p(\ell \given w))
    \]
    is integrable with respect to Lebesgue measure for $P_W$ almost all $w$. Hence, 
    \begin{align*}
      &\E\left[\partial_\ell g(L, W) + \rho(L, W)g(L, W)  \given W = w \right] = \int  \partial_\ell (g(\ell, w) p(\ell \given w)) \, \mathrm{d}\ell\\
      &= \lim_{t \to \infty} \int_{-t}^t \partial_\ell (g(\ell, w) p(\ell \given w)) \, \mathrm{d}\ell = \lim_{t \to \infty}  g(t, w)p(t \given w) -  g(-t, w)p(-t \given w) = 0,
    \end{align*}
    by dominated convergence \citep[][Theorem 12.2]{schilling2017measures}, the fundamental theorem of calculus and (iv). This proves the desired result.
  \end{proof}
\end{lemma}

\begin{lemma}
  \label{lemma:function_decay_from_integrals}
Let $\mu$ denote the Lebesgue measure on $\mathbb{R}$, let $f: \mathbb{R} \to \mathbb{R}$ be differentiable with derivative $f'$ and suppose that 
\[
  \max\left( \int |f| \, \mathrm{d}\mu, \int |f'| \, \mathrm{d}\mu \right) < \infty.
\]
Then $\lim_{|x| \to \infty} f(x) = 0$. 
\begin{proof}
Since $f'$ is the derivative of $f$ by \citet[][Theorem 7.21]{rudin1986real}, we have for all $x \in \mathbb{R}$ and $y \geq x$ that 
\[
  |f(x)| = \left| f(y) - \int_x^y f'(t) \, \mathrm{d}\mu(t) \right| \leq | f(y) | +\int_x^y |f'(t)| \, \mathrm{d}\mu(t). 
\]
Multiplying both sides of the inequality by $\ind_{[x, x+1]}(y)$, using monotonicity of integrals and Fubini's theorem, we see that 
\begin{align*}
  |f(x)| &= \int_x^{x+1} |f(x)| \, \mathrm{d}\mu(y) \leq \int_x^{x+1} | f(y) | \, \mathrm{d}\mu(y) + \int_x^{x+1} \int_x^y |f'(t)| \, \mathrm{d}\mu(t) \, \mathrm{d}\mu(y)\\
  &\leq \int_x^\infty | f(y) | \, \mathrm{d}\mu(y) + \int_x^\infty \int_x^{x+1} |f'(t)| \, \mathrm{d}\mu(t) \, \mathrm{d}\mu(y) =  \int_x^\infty | f(y) | \, \mathrm{d}\mu(y) + \int_x^\infty |f'(t)| \, \mathrm{d}\mu(t).
\end{align*}
By integrability of $f$ and $f'$ the right-hand side must go to $0$ as $x \to \infty$ which proves that $\lim_{x \to \infty} |f(x)| = 0$. An analogous argument can be applied to show that $\lim_{x \to -\infty} |f(x)| = 0$.
\end{proof}
\end{lemma}

\subsection{Robustness of nuisance function estimation for different
  estimators} \label{sec:estimator_robustness}
A key motivation for applying the estimators based on a partially linear model
on $Y$ is to avoid estimating nuisance functions that result in unstable
estimators. We illustrate this point based on three examples below.

\subsubsection[Estimation of nuisance functions with continuous
L]{Estimation of nuisance functions with continuous
  $L$} \label{sec:estimator_robustness_cont}

To illustrate the difficulty of estimating the nuisance functions
$\partial_\ell f$ and $\rho$ (see
Section~\ref{sec:estimation_np_directional}) compared to $m$ and $g$
(see Section~\ref{sec:estimation_plm}) we estimate the functions on
simulated data and compare the estimates to the truth. We simulate
$n=1000$ observations from the model where $W$ is uniformly
distributed on $\Delta^{2}$, $\varepsilon$ and $\xi$ are standard
Gaussian independent of each other and from $W$ and
\[
  L := W^1 + (1 + \ind_{\{W^1 > 1/2\}})\xi, \quad Y:= (1+L)W^1 + \varepsilon.
\]
The out-of-sample predictions evaluated on a new dataset
with $n=1000$ observations can be seen together with
the true function in Figure~\ref{fig:semiparametric_robustness_cont}.

\begin{figure}[H]
  \centering
  \includegraphics*[width=\textwidth]{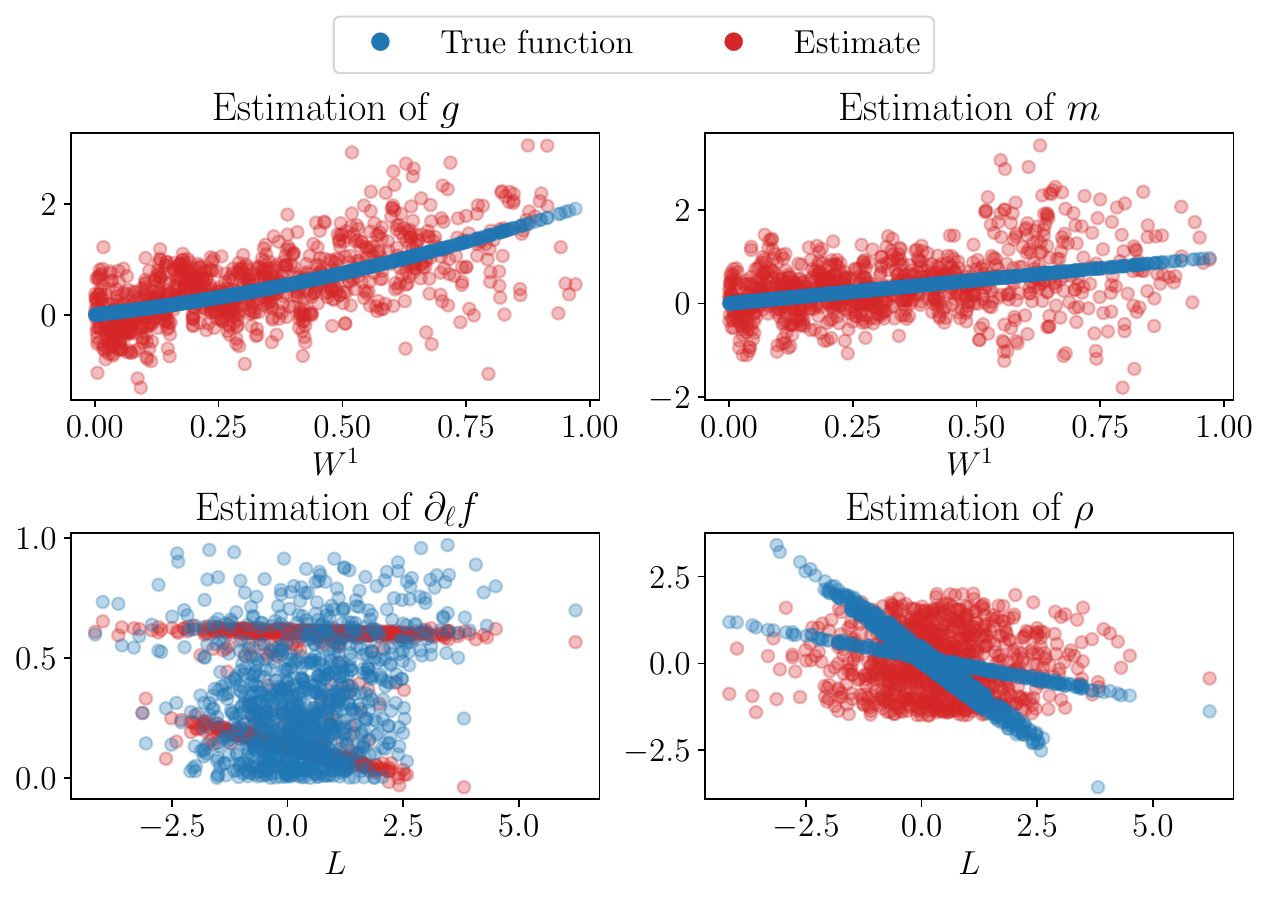}
  \caption{Estimation of nuisance functions on the simulated data in
    Section~\ref{sec:estimator_robustness_cont}. The estimates are
    computed on a dataset of size $n=1000$ and the estimates (red) and
    true functions (blue) are evaluated on a new dataset of the same
    size. We see that the $m$ and $g$ estimates are reasonable
    estimates of the true functions, that roughly capture the behavior
    of the true function in an unbiased (conditioned on $W^1$)
    way. In contrast, the derivative and score estimates, are
    much worse at capturing the more complex behavior (with clear
    biases conditioned on $W^1$).} \label{fig:semiparametric_robustness_cont}
\end{figure}

\subsubsection[Estimation of nuisance functions with binary
L]{Estimation of nuisance functions with binary
  $L$}\label{sec:estimator_robustness_binary}

To illustrate the difficulty of estimating the inverse propensity
score $\pi^{-1}$ (see Section~\ref{sec:estimation_np_binary}) compared
to $m$ (in this case equal to $\pi$) and $g$ (see
Section~\ref{sec:estimation_plm}) we estimate the functions on
simulated data and compare the estimates to the truth. We simulate
$n=1000$ observations from the model where $W$ is uniformly
distributed on $\Delta^{2}$, $\varepsilon$ is standard Gaussian
independent of $W$, $L$ conditioned on $W$ is generated as a Bernoulli random variable
with success probability
$\pi(W):= 0.95\mathrm{expit}(\mathrm{logit}(W^1))+0.05$
and
\[
   Y:= 2LW^1 + \varepsilon.
\]
The out-of-sample predictions evaluated on new dataset with
$n=1000$ observations can be seen together with the true function in
Figure~\ref{fig:semiparametric_robustness_bin}.

\begin{figure}[H]
  \centering
  \includegraphics*[width=\textwidth]{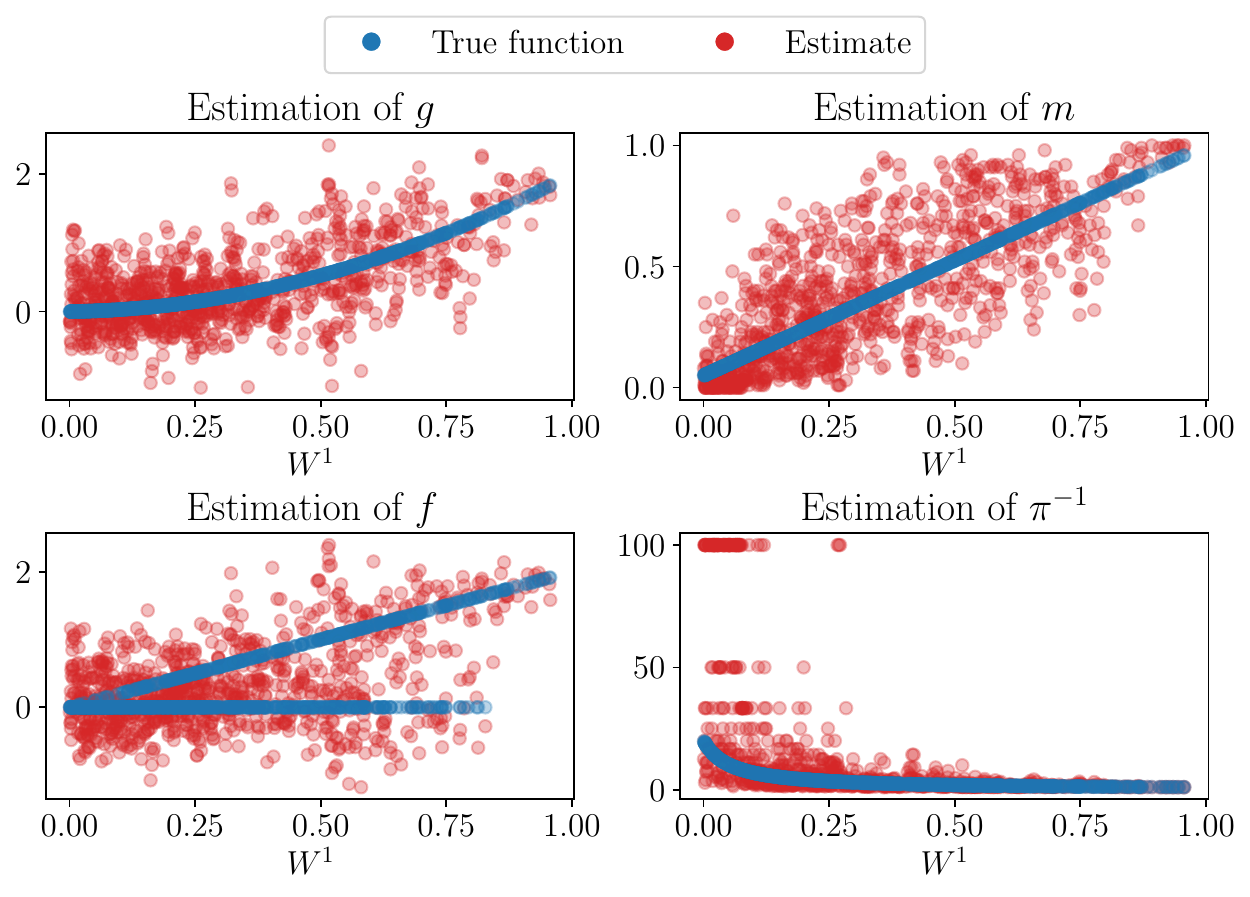}
  \caption{Estimation of nuisance functions on the simulated data in
    Section~\ref{sec:estimator_robustness_binary}. The estimates are
    computed on a dataset of size $n=1000$ and the estimates (red) and
    true functions (blue) are evaluated on a new dataset of the same
    size. We see that the $m$, $g$ and $f$ estimates are reasonable
    estimates of the true functions whereas the inverse propensity
    (even when cutting off extreme values at $100$) is poorly
    estimated for $W^1$ close to
      zero.} \label{fig:semiparametric_robustness_bin}
\end{figure}

\subsubsection{Partially linear model estimate when partial
  linearity is violated} \label{sec:estimator_robustness_plm}

The estimators based on the partially linear model are still
well-behaved even in settings where the model is misspecified. The
estimand
$\E[\cov(Y, L \given W)]/\E[\var(L \given W)]$ just happens to be
different from $\tau$ and $\lambda$ (depending on the setting). We can
see this by plotting the distribution of the estimates from the $d=15$
experiment (with $Y$ nonparametric) of Section~\ref{sec:est_exp} which
we do in Figure~\ref{fig:semiparametric_distributions}. Even
  though the estimand is different, it is still related to the
  original target parameters, since it can be seen as the 'best
  approximation' to a generic nonparametric model, as shown in the
  following proposition.
\begin{proposition}
  \label{prop:plm_approximation}
  Let $(Y, L, W) \in \mathbb{R} \times \mathbb{R} \times \mathcal{W}$
  denote random variables where $\max\{\E[L^2], \E[Y^2] \} <
  \infty$. Then,
  \[
    \theta := \frac{\E[\cov(Y, L \given W)]}{\E[\var(L \given W)]}
    \quad \text{and} \quad h : w \mapsto \E[Y \given W=w] - \theta
    \E[L \given W=w]
  \]
  are solutions to the minimization problem
  \begin{equation}
    \label{eq:optim_h_theta}
    \min_{\theta \in \mathbb{R}, h: \mathcal{W} \to \mathbb{R}} \E\left[\left(\E[Y \given L, W] - \theta L - h(W) \right)^2 \right].
  \end{equation}
  \begin{proof}
    First, we note that instead of the optimization in
      \eqref{eq:optim_h_theta} we can also solve
    \begin{equation}
      \label{eq:otpim_htilde_theta}
      \min_{\theta \in \mathbb{R}, \widetilde{h}: \mathcal{W} \to
        \mathbb{R}} \underbrace{\E\left[\left(\E[Y \given L, W] - \E[Y \given W] -
          \theta (L-\E[L \given W]) - \widetilde{h}(W) \right)^2 \right]}_{=c(\theta,\widetilde{h})},
    \end{equation}
    since using
      $h: w \mapsto \widetilde{h}(w) + \E[Y \given W=w] - \theta \E[L \given
      W=w]$ leads to a solution of \eqref{eq:optim_h_theta}.  Next, we expand the square in the
    objective function of the minimization problem
    \eqref{eq:otpim_htilde_theta} and obtain that
\begin{align*}
 c(\theta, \widetilde{h}) &=  \E\left[\left(\E[Y \given L, W] - \E[Y \given W]  \right)^2 \right] + \theta^2 \E\left[\var(L \given W) \right] + \E[\widetilde{h}(W)^2]\\
  &- 2 \theta \E[(\E[Y \given L, W] - \E[Y \given W])(L - \E[L \given W])].
\end{align*}
It is now clear that the optimal $\widetilde{h}$ is the zero function and by noting that
\[
  \E[(\E[Y \given L, W] - \E[Y \given W])(L - \E[L \given W])] = \E[\cov(Y, L \given W)],
\]
we obtain the desired result by minimizing the quadratic in $\theta$.
\end{proof}
\end{proposition}

\begin{figure}[H]
  \centering
  \includegraphics*[width=\textwidth]{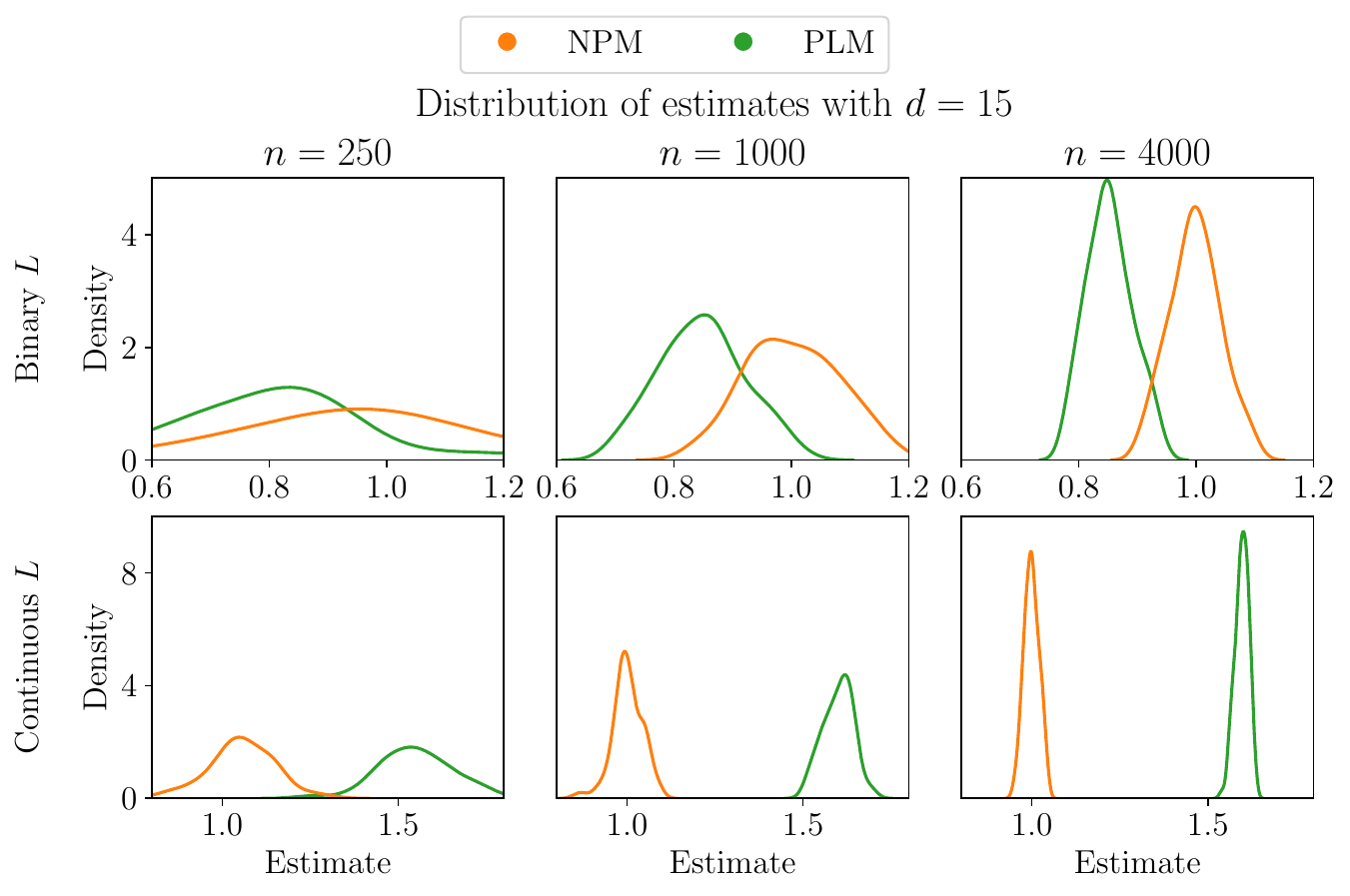}
  \caption{Smoothed empirical distributions
    of the estimates for the nonparametric model (NPM) and
      partially linear model (PLM) estimators in the
    nonparametric settings described in Section~\ref{sec:est_exp} with
    $d=15$. We see that while the NPM estimator correctly clusters
    around the true value of $1$, the PLM estimator also clusters
    however around a different, but potentially still useful,
      value.} \label{fig:semiparametric_distributions}
\end{figure}

\subsection{Computing directional perturbation effects in the presence
  of zero speeds}
\label{sec:zeros}

In Proposition~\ref{prop:derivative-isolating} we only considered
directional perturbations where $\omega_\psi(Z) \neq 0$ almost
surely. We would, however, like to be able to use the derivative-isolating
reparametrizations for estimation of the directional average
perturbation effect $\tau_{\psi}$ even when this is not the
case. To that end, we propose a modified estimation procedure that we
outline in Algorithm~\ref{alg:zeros}. We propose to estimate
$\tau_\psi$ on the subset of observations where
$\omega_\psi(Z) \neq 0$ and then correct the resulting
point and variance estimate appropriately.

\begin{algorithm}[H]
  \caption{Estimation of $\tau_\psi$ in the presence of observations with $\omega_\psi(Z) = 0$.}
  \begin{algorithmic}[1]
  \Statex \textbf{Input:} Data $(Y_i, Z_i)_{i \in [n]}$, directional perturbation $\psi$, derivative-isolating reparametrization $\phi$.
  \State For $i \in [n]$, define $D_i := \ind_{\{\omega_\psi(Z_i) \neq 0\}}$.
  \State Define $\mathcal{I} := \{i \in [n] \mid D_i = 1 \}$.
  \State For $i \in \mathcal{I}$, define $(L_i, W_i) = \phi(Z_i)$.
  \State Compute estimate $\widehat{\tau}$ and asymptotic variance estimate $\widehat{\sigma}^2$ from either Algorithm~\ref{alg:nonparametric_tau} or \ref{alg:dml}.
  \State Compute $\widehat{p} := \frac{1}{n} \sum_{i=1}^n D_i$.
  \Statex \textbf{return} estimate $\widetilde{\tau}:= \widehat{\tau}  \widehat{p}$ and asymptotic variance estimate $\widetilde{\sigma}^2:=\widehat{p}  \widehat{\sigma}^2 + \widehat{\tau}^2 \widehat{p} (1-\widehat{p})$.
  \end{algorithmic}
  \label{alg:zeros}
\end{algorithm}

Theoretical guarantees for the proposed modification follow from
repeating the arguments of the proof of
Theorem~\ref{theorem:normality}~\eqref{theorem:normality_directional}
and \eqref{theorem:normality_plm} to obtain the population quantities
that determine the asymptotic behavior of $\widehat{\tau}$ and
$\widehat{\sigma}^2$ and then applying Proposition~\ref{prop:zeros}
below.

\begin{proposition}
  \label{prop:zeros}
Let $(D, S) \in  \{0, 1\} \times \mathbb{R}$ be a pair of random variables with distributions determined by $\mathcal{P}$ and let $(S_n, D_n)_{n \in \mathbb{N}}$ be an i.i.d.\ sequence of copies of $(S, D)$. Suppose that there exist $\delta, C, c > 0$ such that $\var_P(DS) \geq c$ and $\E_P[D|S|^{2+\delta}] \leq C$. Suppose further that
\[ 
  \widehat{\tau} = \frac{1}{\sum_{i=1}^n D_i} \sum_{i=1}^n D_i S_i + o_{\mathcal{P}}(n^{-1/2})
\]
and
\[
  \widehat{\sigma}^2 = \frac{1}{\sum_{i=1}^n D_i} \sum_{i=1}^n D_i S_i^2 - \left(\frac{1}{\sum_{i=1}^n D_i} \sum_{i=1}^n D_i S_i \right)^2 + o_{\mathcal{P}}(1).
\]
Then, 
\[
  \widehat{p} := \frac{1}{n} \sum_{i=1}^n D_i, \quad \widetilde{\tau} := \widehat{\tau}\widehat{p}, \quad \text{and} \quad \widetilde{\sigma}^2:=\widehat{p}  \widehat{\sigma}^2 + \widehat{\tau}^2 \widehat{p} (1-\widehat{p})
\]
satisfy
\[
  \lim_{n \to \infty} \sup_{P \in \mathcal{P}} \sup_{x \in \mathbb{R}} \left| \mathbb{P}_P\left( \frac{\sqrt{n}}{\widetilde{\sigma}}(\widetilde{\tau} - \tau_P) \leq x  \right) - \Phi(x)\right| = 0.
\]
\end{proposition}
\begin{proof}
We first note that
\[
  \widetilde{\tau} := \frac{1}{n} \sum_{i=1}^n D_i S_i + o_{\mathcal{P}}(n^{-1/2})
\]
since $\widehat{p} \leq 1$. Similarly, we have that
\[
  \widetilde{\sigma}^2 = \frac{1}{n} \sum_{i=1}^n D_i S_i^2 - \left(\frac{1}{n} \sum_{i=1}^n D_i S_i \right)^2 + o_{\mathcal{P}}(1).
\]
Therefore, by several applications of Lemma~\ref{lemma:o_P-O_P-calculus}, we have that
\[
  \frac{\sqrt{n}}{\widetilde{\sigma}}(\widetilde{\tau} - \tau_P) = \frac{\frac{1}{\sqrt{n}} \sum_{i=1}^n (D_i S_i - \tau_P)}{\left(  \frac{1}{n} \sum_{i=1}^n D_i S_i^2 - \left(\frac{1}{n} \sum_{i=1}^n D_i S_i \right)^2  \right)^{1/2}} + o_{\mathcal{P}}(1)
\]
and the result thus follows by
Lemma~\ref{lemma:self-normalised_clt}~\eqref{lemma:self-normalised_clt-1}
by the bounds on the moments and variance of $(D, S)$.
\end{proof}

\subsection{Estimating differentiable regression
  functions}\label{sec:tree-derivative}

The following is a description of the smoothing method we apply in the
nonparametric estimator of the average directional perturbation
effect. The procedure uses random forest weights together with a local
polynomial estimate to smooth a given regression function to produce a
differentiable estimator and to compute its derivatives. Local
polynomials are a prominent and natural method used to estimate
derivatives \citep[e.g.,][]{masry1997, de2013derivative}. We are
interested in estimating the partial derivative of a function
$g:\mathbb{R}\times\mathcal{W}\rightarrow\mathbb{R}$. As in
\citet{klyne2023average} we start from a given arbitrary regression
estimate $\widehat{g}$ of $g$ and apply the following local polynomial
approach. Estimate, for all $i\in[n]$, polynomial coefficients
$\beta\in\mathbb{R}^3$ such that the polynomial
$p:\mathbb{R}\rightarrow\mathbb{R}$ defined for all
$\ell\in\mathbb{R}$ by
\begin{equation*}
  p(\ell)=\beta^1 + \beta^2(\ell-L_i) + \beta^3(\ell-L_i)^2
\end{equation*}
is a good local approximation of $\widehat{g}$ around $L_i$. Formally, for a
given weight matrix $K\in\mathbb{R}^{n\times n}$, which encodes a
closeness measure between sample points, we minimize for all $i\in[n]$
the weighted mean squared loss
\begin{equation}
    \label{eq:local_polynomial_fit}
    \hat{\beta}(i):=\operatorname{argmin}_{\beta\in\mathbb{R}^{3}}\sum_{k=1}^{n}\left(\widehat{g}(L_k,W_k)-\sum_{j=1}^{3}\beta^j(L_{k}-L_i)^{j-1}\right)^2 K_{i, \ell}.
\end{equation}
Then, using the estimated coefficients $\hat{\beta}(i)$, we get a
smoothed estimate $\widehat{g}_s$ which at each sample point
$i\in [n]$ is given by
\begin{equation}
  \label{eq:locpol_deriv_est}
  \widehat{g}_s(L_i,W_i):= \hat{\beta}^1(i)
  \quad\text{and}\quad
  \partial_{\ell}\widehat{g}_s(L_i, W_i):= \hat{\beta}^2(i).
\end{equation}
As weights $K$ we could, for example, use kernel weights, i.e.,
$K_{i,\ell}=k((X_i-X_{\ell})/\sigma)$ for a kernel function $k$ and a
bandwidth $\sigma>0$. As we are considering a potentially
high-dimensional setting, we instead propose to use random forest
weights \citep{meinshausen2006quantile, athey2019}. The full procedure
is detailed in Algorithm~\ref{alg:forest-locpol}.

\begin{algorithm}[H]
  \caption{Local polynomial based smoothing with random forest
    weights}
  \begin{algorithmic}[1]
    \Statex \textbf{Input:} Data $(L_i, W_i)_{i\in[n]}$, fitted values $(\widehat{g}(L_i, W_i))_{i\in[n]}$
    
    \State Regress $(\widehat{g}(L_i, W_i))_{i\in[n]}$ on $(L_i, W_i)_{i\in[n]}$ using a random forest yielding $\widehat{\mu}$.
    
    \State Extract weight matrix $K$ from $\widehat{\mu}$.
    
    \For{$i = 1,\ldots,n$}
      \State Set $\hat{\beta}(i)$ to equal the coefficient of a degree $2$ local polynomial fit in \eqref{eq:local_polynomial_fit}.
      
      \State Set $\widehat{g}_s(L_i, W_i):=\hat{\beta}^1(i)$.
      \State Set $\partial_{\ell}\widehat{g}_s(L_i, W_i):=
      \hat{\beta}^2(i)$.
      
    \EndFor
    \Statex \textbf{return} smoothed predictions
      $(\widehat{g}_s(L_i, W_i))_{i\in[n]}$ and
  derivatives $(\partial_{\ell}\widehat{g}_s(L_i, W_i))_{i\in[n]}$
  \end{algorithmic}
  \label{alg:forest-locpol}
  \end{algorithm}

\subsection{Location-scale modelling of the score
  function}\label{sec:loc-scale}

The following is a description of the estimator of the conditional
score based on a location-scale model proposed by
\citet{klyne2023average}. We apply this in our experiments to estimate
conditional score functions and include it here for completeness. We
are interested in the estimation of the conditional score of $L$ given
$W$, where $L \in \mathbb{R}$ and
$W \in \mathcal{W}$, that is, the derivative of the logarithm of the
conditional density of $L$ given $W$. Conditional density estimation
is difficult even in low dimensions.  However, if we are willing to
assume that there exist functions
$m:\mathcal{W}\rightarrow\mathbb{R}$,
$v:\mathcal{W}\rightarrow[0, \infty)$ and a random variable
$\zeta\in\mathbb{R}$ that is independent of $W$ such that
\[
  L = m(W) + v(W)^{1/2}\zeta,
\]
we can simplify
estimation of the score function. The procedure is outlined in
Algorithm~\ref{alg:location-scale_score} below. In all of our
experiments, we use the univariate score estimation method given by
\citet{ng1994smoothing} that is based on ideas from
\citet{cox1985penalty}.

\begin{algorithm}[H]
  \caption{Location-scale score modelling estimation}
  \begin{algorithmic}[1]
  \Statex \textbf{Input:} Data $(L_i, W_i)_{i \in [n]}$, nuisance regression methods $\widehat{\eta} = (\widehat{m}$, $\widehat{v})$, univariate score estimator $\widehat{\rho}$.
  \State Split indices $\{1, \dots, n\}$ into $2$ disjoint sets $I_1, I_2$ of (roughly) equal size.
  \State Regress $L$ on $W$ using $I_1$ yielding regression estimate $\widehat{m}$.
  \State Regress $(L-\widehat{m}(w))^2$ on $W$ using $I_1$ yielding variance estimate $\widehat{v}$.
  \State For $i \in I_2$, set $\widehat{\xi}_i := \frac{L_i - \widehat{m}(W_i)}{\widehat{v}(W_i)^{1/2}}$.
  \State Estimate a univariate score from $(\widehat{\xi}_i)_{i \in I_2}$ yielding $\widehat{\rho}$.
  \Statex \textbf{return} $\widehat{\rho}$.
  \end{algorithmic}
  \label{alg:location-scale_score}
  \end{algorithm}

\section{Details and additional results from numerical
  experiments} \label{sec:extra_num}

Section~\ref{sec:intro_sim} provides the details behind the experiments in
Table~\ref{tab:toy_illustration} in Section~\ref{sec:intro}.
Section~\ref{sec:est_exp} contains an investigation of the performance of the
semiparametric estimators discussed in Section~\ref{sec:estimation} when the
conditioning variable $W$ is compositional. Section~\ref{sec:diversity_groups}
provides a detailed description of how we generate the groups used for
aggregation in the experiment described in Section~\ref{sec:diversity_income}.
Section~\ref{sec:additional_adult} includes figures visualizing the
semi-synthetic dataset used in Section~\ref{sec:diversity_income}.
Section~\ref{sec:american_schools_preprocessing} provides a detailed description
of the non-compositional background variables that are controlled for in
Section~\ref{sec:american_schools}. Section~\ref{sec:microbiome_consistency}
describes an additional experiment on the microbiome data analyzed in
Section~\ref{sec:microbiome}. Section~\ref{sec:additional_microbiome} contains
additional plots illustrating results from the microbiome data analysis in
Section~\ref{sec:microbiome}.

\subsection{Simulations from
  Section~\ref{sec:toy_examples}} \label{sec:intro_sim}

In this section we describe the data-generating processes and methods used in
the numerical illustrations of Example~\ref{ex:absence_microbes} and
Example~\ref{ex:increase_diversity}. The results are provided in
Table~\ref{tab:toy_illustration}.

\subsubsection{Example~\ref{ex:absence_microbes}: Effects of absence or presence
of microbes} 
Let $L \sim \mathrm{Bern}(1/2)$, $U^1, U^2 \sim 
\mathrm{Unif}([0, 1])$ and $B \sim \mathrm{Bern}(1/2)$ be independent random 
variables and define
\[
  Z := (Z^1, W^2(1-Z^1), (1-W^2)(1-Z^1)) \in \Delta^{2},
\]
where $W^2:=U^1(1-L) + U^1 L B$ and $Z^1 := (1-L)
U^2$. Furthermore, define $Y$ such that $Y$ is Bernoulli
distributed with 
\[
  \mathbb{P}(Y=1 \given Z) =  \sim \mathrm{Bern}\biggl(\frac{3}{4}-\frac{1}{8}\ind_{\{Z^1=0\}} + \frac{1}{2}\ind_{\{Z^2=0\}}\biggr).
\]
In the experiment, we generate $n=1000$ i.i.d.\ copies
$(Y_1,Z_1),\ldots,(Y_n,Z_n)$ of $(Y, Z)$.

$\cke^1$ is computed using the $\lambda$-estimator based on a partially linear model for $Y$ (Algorithm~\ref{alg:dml}) with $10$-fold cross-fitting and random forests for all regressions.

\subsubsection{Example~\ref{ex:increase_diversity}: Effects of
  increased diversity}

Let $\Pi$ denote the orthogonal projection onto the orthogonal
complement of $\mathrm{span}((1,1,1))$. Moreover, let
$U \sim \mathcal{N}(0, \Pi\, \Pi^\top)$ and 
$\xi, \varepsilon \sim \mathcal{N}(0, 1)$ be independent random
variables and define 
\[
  Z := -DW + z_{\cen} \in \Delta^2\quad\text{and}\quad Y := L + 4 W^\top e_1 + \varepsilon,
\]
where $W := \frac{U}{\|U\|_1}$,
$D := \mathrm{expit}(W^\top e_1 + \xi)$ and
$L:= -\frac{2}{9} D \sum_{j=1}^3 \sum_{k=1}^3 |W^j - W^k|$. Then, by
construction $L$ is minus the Gini coefficient of $Z$. Hence, from the
partially linear model of $Y$ it follows that the true effect of
increasing the negative Gini coefficient is equal to $1$. In the
experiment, we generate $n=1000$ i.i.d.\ copies
$(Y_1,Z_1),\ldots,(Y_n,Z_n)$ of $(Y, Z)$.

$\cdi_{\gini}$ is computed using the $\tau$-estimator based on a
partially linear model for $Y$ (Algorithm~\ref{alg:dml}) with
$10$-fold cross-fitting and random
forests for all regressions.

\subsection{Semiparametric estimators}\label{sec:est_exp}
In this section we investigate the performance of the semiparametric
estimators outlined in Section~\ref{sec:estimation} with $W$ as a
simplex-valued random variable. More specifically, we compare the
different approaches of estimating
  \begin{equation*}
        \lambda_{\psi}
    = \frac{\E\big[\E\left[Y \given L=1, W \right]]
      - \E\left[Y \right] }{\mathbb{P}(L = 0)}
    \quad\text{and}\quad
    \tau_\psi =  \E[\partial_\ell \E[Y \given L=\ell, W] \big|_{L=\ell}]
  \end{equation*}
  from i.i.d.\ data $(Y_1, L_1,W_1),\ldots,(Y_n,L_n,W_n)$. For both
  types of average perturbation effects we consider the following
  seven types of estimators.
\begin{itemize}
  \itemsep0em 
\item \texttt{NPM}: The one-step corrected nonparametric estimators
  including cross-fitting, i.e., \eqref{eq:nonparametric_discPE} and
  \eqref{eq:nonparametric_contPE}. See
  Algorithms~\ref{alg:nonparametric_tau}
  and~\ref{alg:nonparametric_lambda}, with $K=2$, for the detailed
  binary and directional effect estimators, respectively.
\item \texttt{NPM\_no\_x}: The same estimators as
  \texttt{NPM} but without cross-fitting.
\item \texttt{NPM\_oracle}: The same estimators as \texttt{NPM}
  including cross-fitting but using the true nuisance functions for
  the one-step-correction (propensity score $\pi$ in the binary case
  and conditional score function $\rho$ in the directional case). This
  is an oracle as the true nuisance functions are unknown in practice.
\item \texttt{PLM}: The partially least model based estimators in
  \eqref{eq:PLM_effect} including cross-fitting. See
  Algorithm~\ref{alg:dml} with $K=2$ for the detailed estimators.
\item \texttt{PLM\_no\_x}: The same estimators as
  \texttt{PLM} but without cross-fitting.
\item \texttt{plugin}: The naive plug-in versions of the \texttt{NPM}
  estimates. In the binary effect case, the estimator in
  \eqref{eq:nonparametric_discPE} without the term involving
  $\widehat{\pi}$ and in the directional effect case, the estimator in
  \eqref{eq:nonparametric_contPE} without the term involving
  $\widehat{\rho}$.
\item \texttt{plugin\_no\_x}: The same estimators as \texttt{plugin}
  but without cross-fitting.
\end{itemize}
The variance of the \texttt{NPM} and \texttt{PLM} estimates is computed based on
the provided algorithms for the \texttt{NPM} and \texttt{PLM} methods while the
variance for the plug-in is estimated using an estimator akin to the one given
in \eqref{eq:naive_variance}. Each estimator above is itself based on multiple
of the regression estimates $\widehat{f}$, $\widehat{\pi}$, $\widehat{\rho}$,
$\widehat{g}$ and $\widehat{m}$. For the regression estimates, we use a $5$-fold
cross-validation based estimator which for each fit selects among the following
three models: a simple model that predicts the mean of the outcome (or the
proportion of each class for binary classification), a random forest with $250$
trees with trees grown to a maximum depth of $2$ and a random forest with $250$
trees and no depth restriction. For the score estimator we use the
location-scale based estimator described in Section~\ref{sec:loc-scale}, that
uses the same $5$-fold cross-validated regression estimator for the location and
scale nuisance regressions.

We consider four data-generating models each with varying sample size $n \in
  \{250, 1000, 4000\}$ and predictor dimension $d \in \{3, 15, 75\}$. In all
  settings, we let $W$ be uniformly distributed on $\Delta^{d-2}$, denote by
  $\mathrm{med}_{W_1}$ the population median of $W_1$ and $\sigma_{W_1}^2$ the
  variance of $W_1$ and define $B(W_1) := \ind_{\{W_1 > \mathrm{med}_{W_1}\}}$.
  We then consider the following four generative models
\begin{itemize}
\item Binary $L$ with partially linear $Y$:
  \begin{equation*}
    L = U_0 (1-B(W_1)) + U_1 B(W_1)\quad\text{and}\quad
    Y =L + \sigma_{W_1}^{-1}W_1 +\varepsilon
  \end{equation*}
\item Binary $L$ with nonparametric $Y$:
  \begin{equation*}
    L = U_0 (1-B(W_1)) + U_1 B(W_1)\quad\text{and}\quad
    Y = \frac{7}{5}L B(W_1) + \sigma_{W_1}^{-1}W_1 +\varepsilon
  \end{equation*}
\item Continuous $L$ with partially linear $Y$:
  \begin{equation*}
    L =  B(W_1) + \xi\quad\text{and}\quad
    Y = L +  B(W_1) + \varepsilon    
  \end{equation*}
\item Continuous $L$ with nonparametric $Y$:
  \begin{equation*}
    L =  B(W_1) + (1+ B(W_1))\xi\quad\text{and}\quad
    Y = 2B(W_1)L + B(W_1) + \varepsilon
  \end{equation*}
\end{itemize}
The noise variables $U_0 \sim \mathrm{Bern}(0.8)$, $U_1 \sim
\mathrm{Bern}(0.5)$, $\varepsilon \sim \mathcal{N}(0, 1)$ and
$\xi\sim\mathcal{N}(0,1)$ are all sampled independently (and independent of
$W$).  The true value of $\lambda$ in the binary $L$ settings and $\tau$ in the
continuous $L$ settings is $1$. For each setting we simulate data and compute an
estimate with a corresponding asymptotic $95\%$ confidence interval using each
of the methods described above. We repeat the experiment $100$ times and compute
coverage rates of the $95\%$ confidence intervals for $d=3$ in
Figure~\ref{fig:semiparametric_sim_3}, for $d=15$ in
Figure~\ref{fig:semiparametric_sim_15} and for $d=75$ in
Figure~\ref{fig:semiparametric_sim_75}.
\begin{figure}[ht!]
  \centering
  \includegraphics*[width=\textwidth]{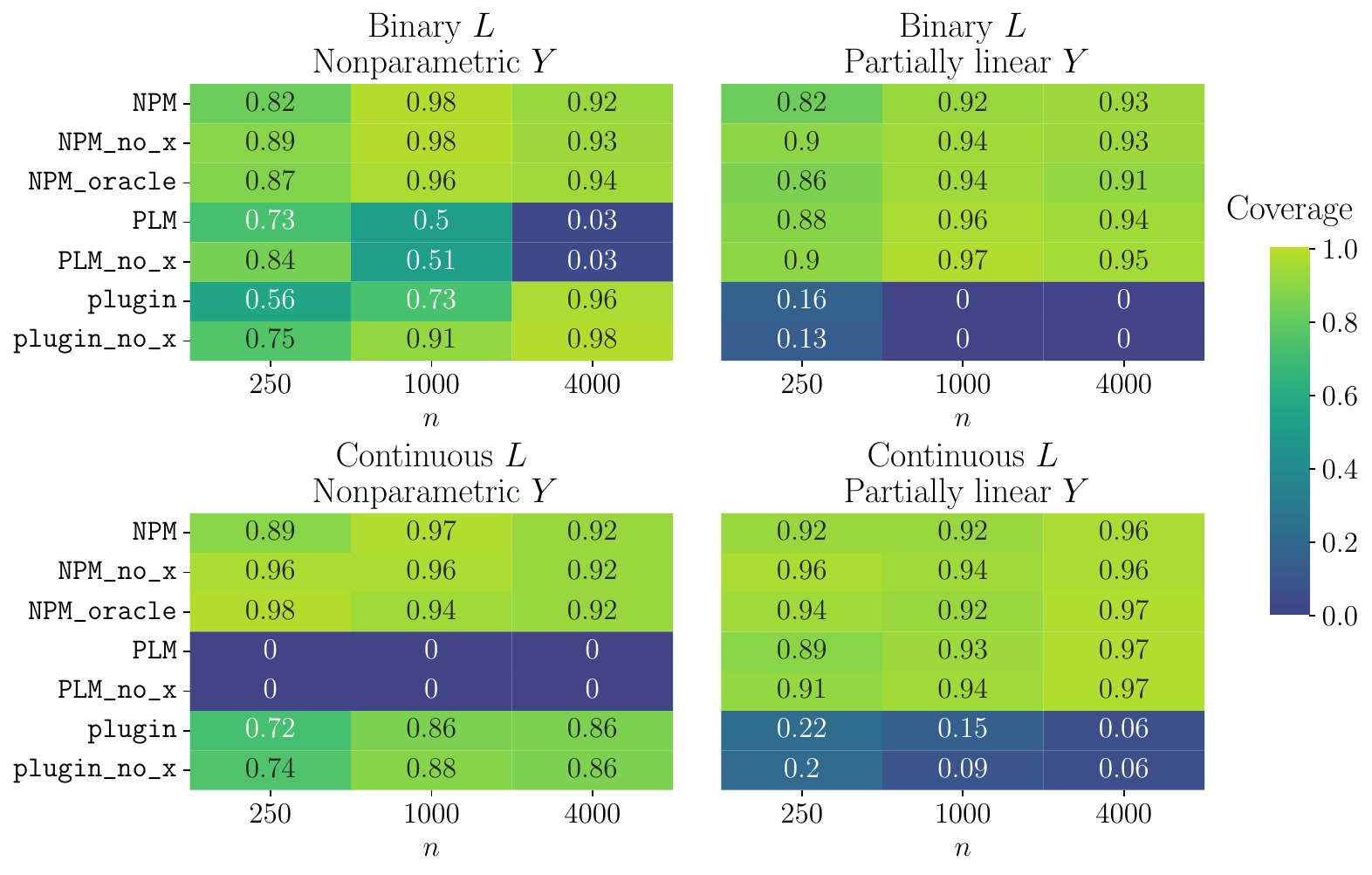}
  \caption{Estimated coverage rates of the $95\%$ confidence intervals for each
  different estimation procedure over $100$ repetitions for each of the settings
  described in Section~\ref{sec:est_exp} with $d=3$. \texttt{PLM} works well if
  $Y$ follows a partially linear model, but can fail if this assumption is not
  satisfied as in the nonparametric $Y$ setting. \texttt{NPM} works in both
  settings but can be less robust for small sample sizes. Further, it can be
  seen that the plug-in estimator and the estimators without cross-fitting do
  not lead to nominal coverage in some of these settings.}
  \label{fig:semiparametric_sim_3}
\end{figure}

\begin{figure}[ht!]
  \centering
  \includegraphics*[width=\textwidth]{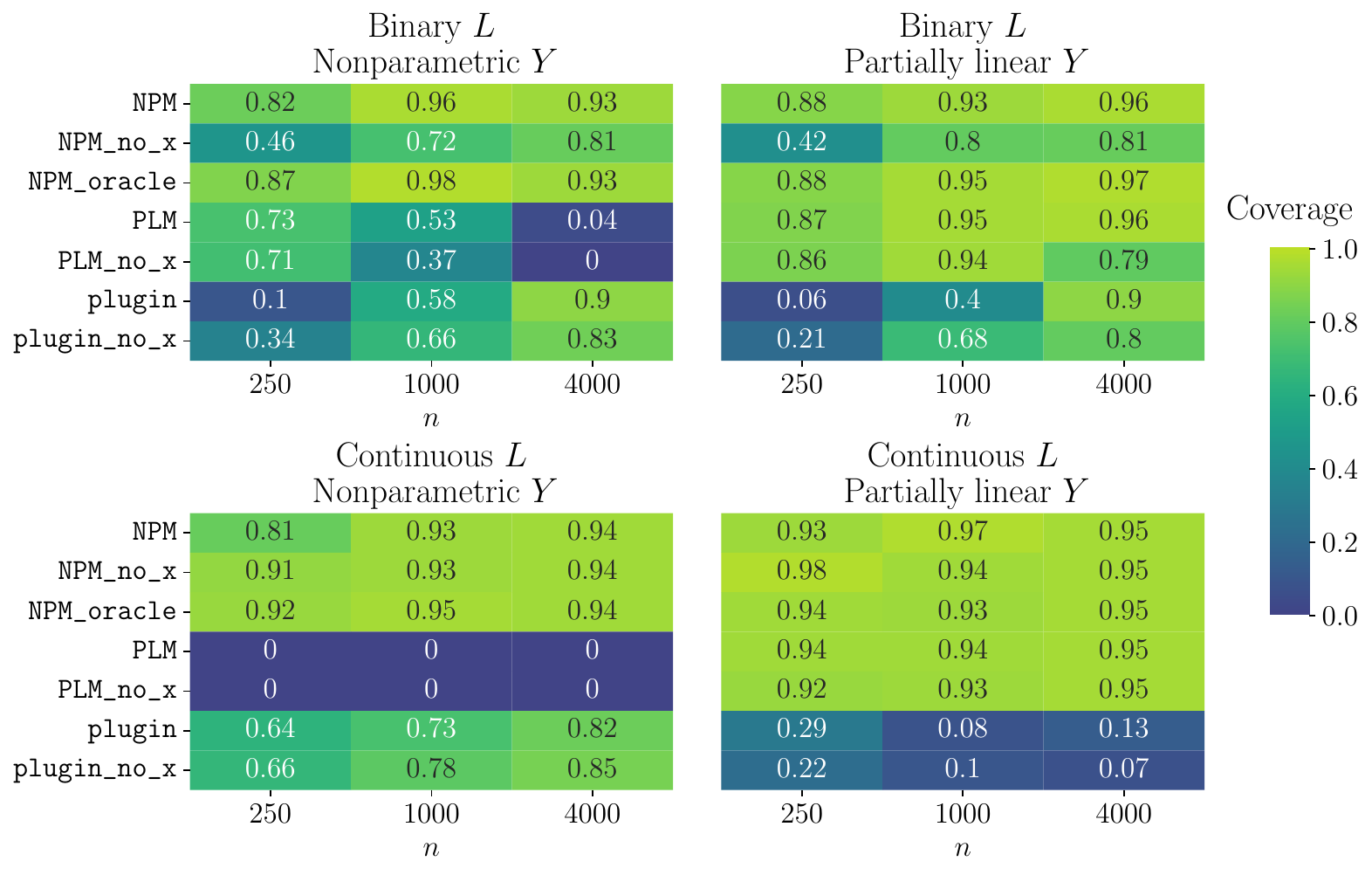}
  \caption{Estimated coverage rates of the $95\%$ confidence intervals over
  $100$ repetitions for each of the settings described in
  Section~\ref{sec:est_exp} with $d=15$. The patterns observed in
  Figure~\ref{fig:semiparametric_sim_3} are repeated here.}
  \label{fig:semiparametric_sim_15}
\end{figure}

\begin{figure}[ht!]
  \centering
  \includegraphics*[width=\textwidth]{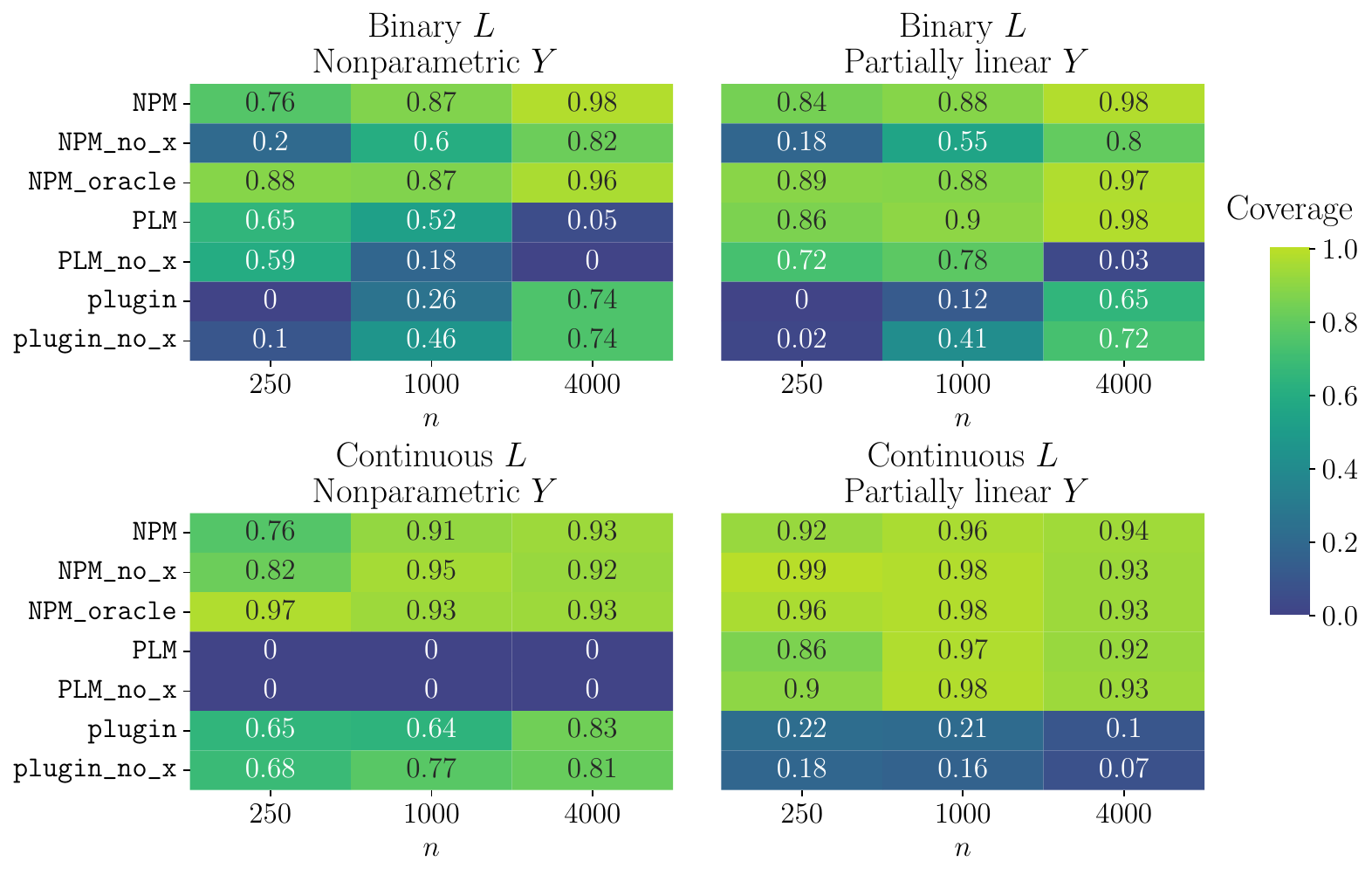}
  \caption{Estimated coverage rates of the $95\%$ confidence intervals over
  $100$ repetitions for each of the settings described in
  Section~\ref{sec:est_exp} with $d=75$. The patterns observed in
  Figure~\ref{fig:semiparametric_sim_3} are repeated here.}
  \label{fig:semiparametric_sim_75}
\end{figure}

In all settings we see that the partially linear model works well when it is
well-specified but can fail when the partially linear assumption is violated. We
see that the plug-in estimator generally performs worse than the one-step
estimator but this disparity is reduced for large $n$. In many cases
cross-fitting is not a necessity, however, in some settings cross-fitting does
positively impact the performance of the confidence intervals, e.g. for $n=1000$
and binary $L$. In the cases where cross-fitting is not necessary, it does not
seem to negatively impact the coverage of the intervals.

\subsection{Description of the construction of groups in Section~\ref{sec:diversity_income}}
\label{sec:diversity_groups}

In this section we describe how we construct the two community-level
grouped datasets analyzed in Section~\ref{sec:diversity_income}.  On the
individual-level data the observed proportions of \texttt{race}
(categorical, 'White', 'Black' and 'Other'), \texttt{education}
(binary, $1$ if strictly greater than high school and $0$ otherwise)
and \texttt{compensation} (binary, $1$ if compensation above 50,000
USD per year and $0$ otherwise) are
\begin{align*}
  (\widehat{\pi}_R(\mathrm{White}),
  \widehat{\pi}_R(\mathrm{Black}),
  \widehat{\pi}_R(\mathrm{Other}))
  &=(0.855, 0.096, 0.049),\\
  (\widehat{\pi}_E(0),
  \widehat{\pi}_E(1)) &= (0.454, 0.546),\\
  (\widehat{\pi}_C(0), \widehat{\pi}_C(1))&= (0.761, 0.239),
\end{align*}
respectively. Now, for the two types of groupings we define target
probabilities $p_i=(p_i^0, p_i^1, p_i^2)\in\Delta^2$ for each
individual $i$ (see below). Based on these probabilities we randomly
draw for each individual $i$ one of the three categories using the
probabilities $p_i$. Finally, we randomly group individuals with the
same categories into communities of approximately $50$ (as many
communities of $50$ and the remaining with $49$) and aggregate the
predictors and response as described in the main text.  The
probabilities $(p_i^0, p_i^1, p_i^2)\in\Delta^2$ are selected, in each
of the two settings, as follows.

\emph{Grouping by race and education:} The probability $p_i^1$ of
observation $i$ belonging to category $1$ is set to
\[
  p_i^1 := p_{R}^1(\texttt{race}_i) \cdot
  p_{E}^1(\texttt{education}_i),
\]
where $\texttt{race}_i$ and $\texttt{education}_i$ denote the race and
education level of the $i$th individual, $p_R^1$ and $p_E^1$ are
probability vectors satisfying
\begin{equation}
  \label{eq:pR1_first}
  p_R^1(\mathrm{White}) \propto
  \frac{0.8}{\widehat{\pi}_R(\mathrm{White})}, \quad
  p_R^1(\mathrm{Black}) \propto
  \frac{0.02}{\widehat{\pi}_R(\mathrm{Black})}, \quad
  p_R^1(\mathrm{Other}) \propto
  \frac{0.18}{\widehat{\pi}_R(\mathrm{Other})},
\end{equation}
and
\[
  p_E^1(0) \propto \frac{0.95}{\widehat{\pi}_E(0)}, \quad p_E^1(1)
  \propto \frac{0.05}{\widehat{\pi}_E(1)}.
\]
That is, category $1$ is about as diverse ($(0.8, 0.02, 0.18)$ is
about as diverse as $\widehat{\pi}_R$) and with lower education
($0.05<\widehat{\pi}_E(1)$). The probability $p_i^2$ of observation
$i$ belonging to category $2$ is set to
\[
  p_i^2 := p_{R}^2(\texttt{race}_i) \cdot
  p_{E}^2(\texttt{education}_i),
\]
where similarly $p_R^2$ and $p_E^2$ are probability vectors but now
satisfying
\begin{equation}
  \label{eq:pR2_first}
  p_R^2(\mathrm{White}) \propto
  \frac{0.65}{\widehat{\pi}_R(\mathrm{White})}, \quad
  p_R^2(\mathrm{Black}) \propto
  \frac{0.12}{\widehat{\pi}_R(\mathrm{Black})}, \quad
  p_R^2(\mathrm{Other}) \propto
  \frac{0.23}{\widehat{\pi}_R(\mathrm{Other})}
\end{equation}
and
\[
  p_E^2(0) \propto \frac{0.3}{\widehat{\pi}_E(0)}, \quad p_E^2(1)
  \propto \frac{0.7}{\widehat{\pi}_E(1)}.
\]
That is, category $2$ is more diverse ($(0.65, 0.12, 0.23)$ is more
diverse than $\widehat{\pi}_R$) and with higher education
($0.7>\widehat{\pi}_E(1)$).  Lastly, the probability $p_i^0$ for
category $0$ is set to $p_i^0:=1-p_i^1 - p_i^2$.

\emph{Grouping by race and compensation:} The probability $p_i^1$ of
observation $i$ belonging to category $1$ is set to
\[
  p_i^1 := p_{R}^1(\texttt{race}_i) \cdot
  p_{C}^1(\texttt{compensation}_i),
\]
where $p_R^1$ is as in \eqref{eq:pR1_first} and $p_C^1$ is a
probability vector satisfying
\[
  p_C^1(0) \propto \frac{0.9}{\widehat{\pi}_C(0)}, \quad p_C^1(1)
  \propto \frac{0.1}{\widehat{\pi}_C(1)}.
\]
That is, category $1$ is about as diverse and with lower compensation
($0.05<\widehat{\pi}_C(1)$).  The probability $p_i^2$ for category $2$
is set to
\[
  p_i^2 := p_{R}^2(\texttt{race}_i) \cdot p_{C}^2(\texttt{compensation}_i),
\]
where $p_R^2$ is as in \eqref{eq:pR2_first} and $p_C^2$ is a
probability vector satisfying
\[
  p_C^2(0) \propto \frac{0.7}{\widehat{\pi}_C(0)}, \quad p_C^2(1)
  \propto \frac{0.3}{\widehat{\pi}_C(1)}.
\]
That is, category $2$ is more diverse and with higher compensation
($0.3>\widehat{\pi}_C(1)$).  Lastly, we
set the probability $p_i^0$ for category $0$ is
$p_i^0:=1-p_i^1 - p_i^2$.  

\subsection{Illustration of semi-synthetic data from Section~\ref{sec:diversity_income}}
\label{sec:additional_adult} 
Figure~\ref{fig:adult} and Figure~\ref{fig:adult_comp} illustrate the
distribution of the semi-synthetic data in Section~\ref{sec:diversity_income}. We
see that the overall relationship between diversity and compensation is negative
but conditioned on $W = \frac{Z-z_{\cen}}{\|Z-z_{\cen}\|_1}$ the effect becomes
positive.

\begin{figure}[ht!]
  \centering
  \includegraphics*[width=\textwidth]{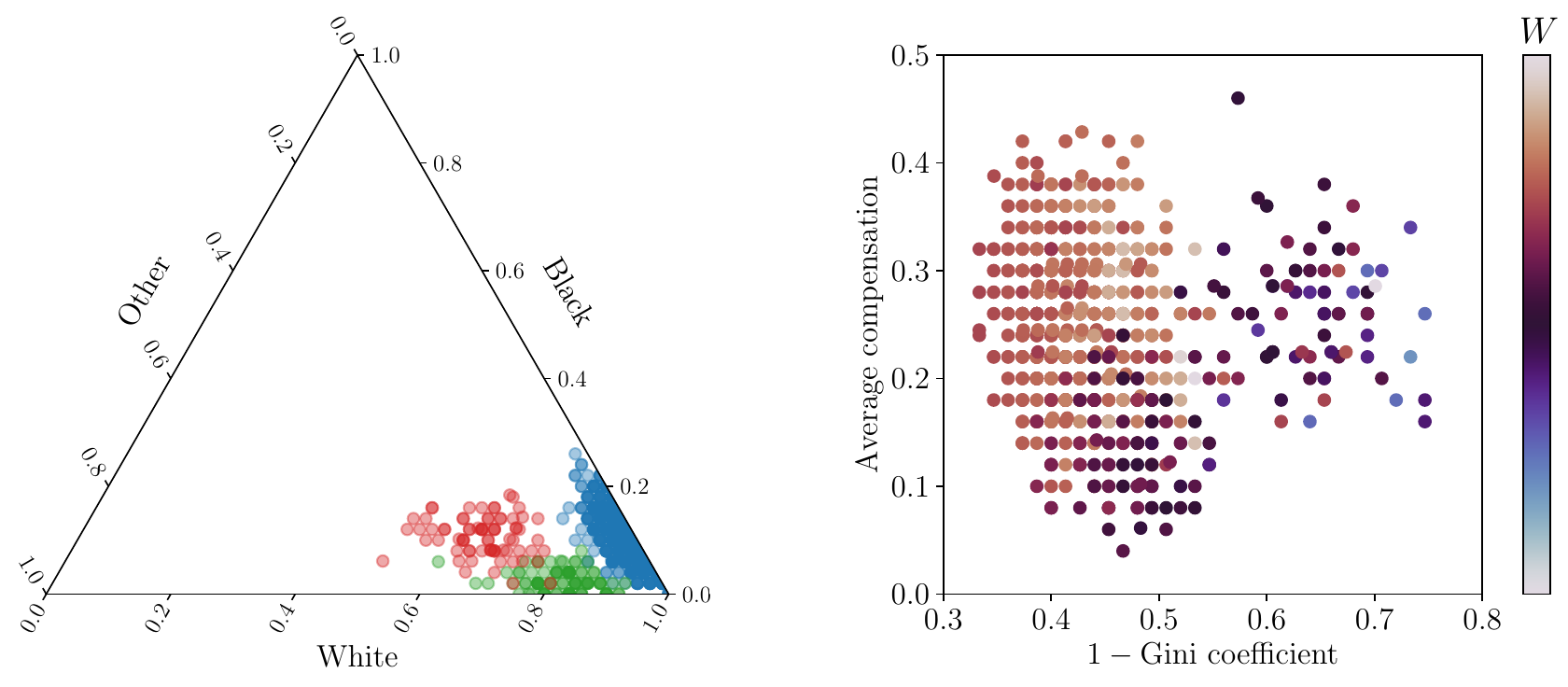}
  \caption{Summary of the semi-synthetic `Adult' dataset when grouping on
    education. Left: Racial compositions in a ternary plot. Each point
    corresponds to one group and is colored according to which category it
    belongs to ($0$ in blue, $1$ in green and $2$ in red). Right: Plot of
    average compensation per group ($Y$) against one minus the  Gini
    coefficient colored by the value of $W = (z_{\cen}-Z) / \|z_{\cen}-Z\|_1$
    ($W$ takes values in a $1$-dimensional subset of $\mathbb{R}^3$ and can be
    represented by a point on unit circle). Depending on the color the
    dependence between $Y$ and the Gini coefficient changes.}
  \label{fig:adult} 
\end{figure}

\begin{figure}[ht!]
  \centering
  \includegraphics*[width=\textwidth]{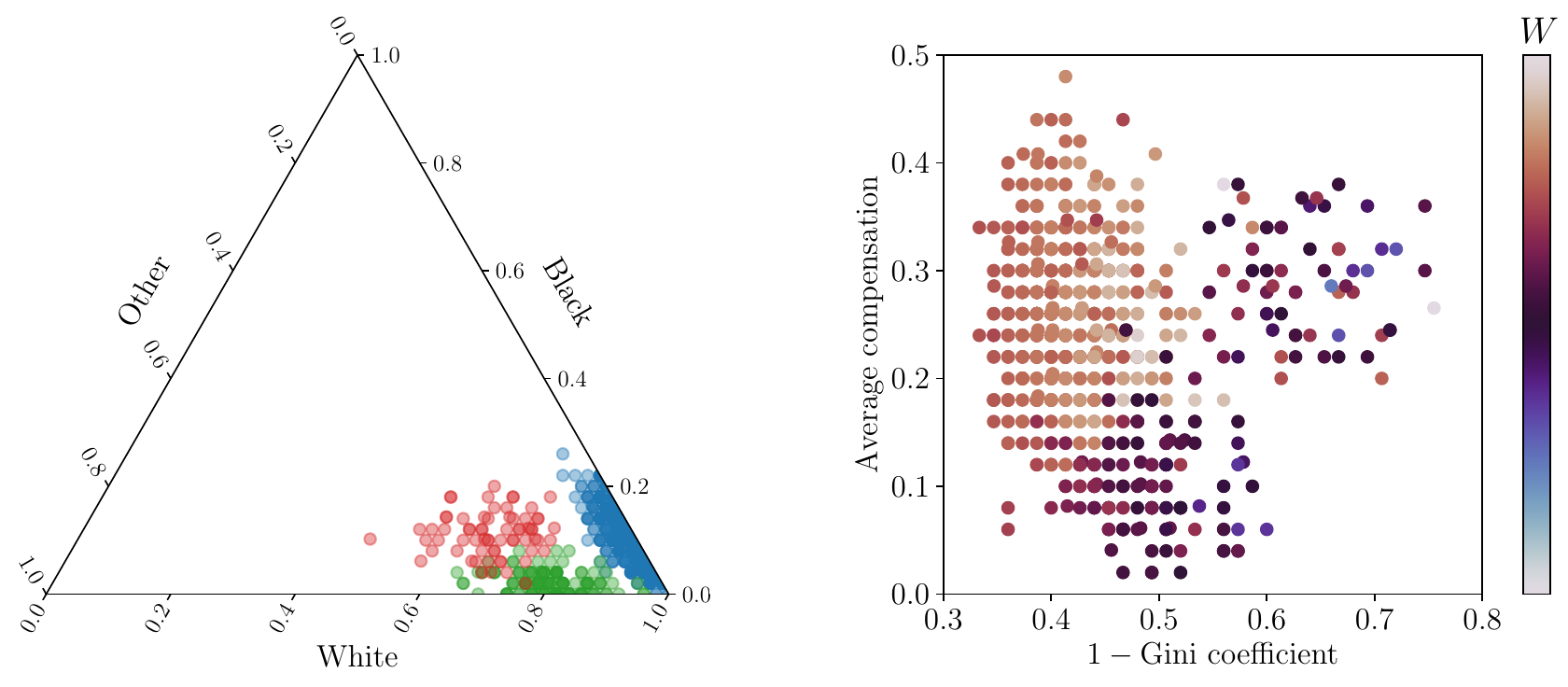}
  \caption{Summary of the semi-synthetic `Adult' dataset when grouping on
    compensation. Left: Racial compositions in a ternary plot. Each point
    corresponds to one group and is colored according to which category it
    belongs to ($0$ in blue, $1$ in green and $2$ in red). Right: Plot of
    average compensation per group ($Y$) against one minus the Gini
    coefficient colored by the value of $W = (z_{\cen}-Z) / \|z_{\cen}-Z\|_1$
    ($W$ takes values in a $1$-dimensional subset of $\mathbb{R}^3$ and can be
    represented by a point on unit circle). Depending on the color the
    dependence between $Y$ and the Gini coefficient changes. }
  \label{fig:adult_comp} 
\end{figure}

\subsection{Semi-synthetic individual-level experiment using groups from
Section~\ref{sec:diversity_income}}
\label{sec:diversity_income_individual}

In this section we consider a variant of the experiment performed in
Section~\ref{sec:diversity_income} where instead of aggregating the full dataset
into groups, we preserve the individual observations but generate new responses
based on the groups. More specifically, we generate the groups as before and
estimate the probability that $Y = 1$ within each group and use this probability
to generate new $Y$s for each observation. For each individual, we thus have a
response $Y$, a race $R$, a `contextual' variable $Z$ that gives the racial
distribution of each group and additional covariates $X$. By construction of the
groups and sampling of the $Y$, we know that there is a positive effect of
diversity on the income. We now apply the \texttt{naive\_diversity} and
$\cdi_{\gini}$ methods as in Section~\ref{sec:diversity_income} except we always
control for the race of inviduals and omit the unstable estimators that
condition on $Z$. The results can be seen in Table~\ref{tab:adult_individual}.

\begin{table}[H]
  \centering
  \begin{tabular}{l|cc}
    \toprule
      \multirow{2}*{Method}  & \multicolumn{2}{c}{Grouping on compensation} \\ 
                             & Estimate & $95\%$ CI \\
    \midrule
    $\texttt{naive\_diversity} \given R$ & $\textcolor{LKred}{\mathbf{-0.064}}$ & $(-0.124, -0.003)$\\
    $\texttt{naive\_diversity} \given R, X$ & $\textcolor{LKred}{\mathbf{-0.064}}$ & $(-0.124, -0.004)$\\
    $\cdi_{\mathrm{Gini}} \given R$ & $\textcolor{LKgreen}{\mathbf{0.459}}$ &$(0.247, 0.671)$\\
    $\cdi_{\mathrm{Gini}} \given R, X$ & $\textcolor{LKgreen}{\mathbf{0.488}}$ & $(0.282, 0.694)$\\
    \bottomrule
  \end{tabular}
  \caption{Table of estimates in the individual-level `Adult' data experiment.
  Bold red numbers indicate a signiﬁcant negative effect of diversity at a 5\%
  level, while bold green indicates a signiﬁcant positive effect. The naive
  diversity method indicates that there is a signiﬁcant negative effect of
  diversity on compensation, even when controlling for individual race ($R$) and
  additional covariates ($X$), despite the fact that the data is constructed
  with a positive effect of diversity. The CDI correctly identifies that the
  effect is positive.}\label{tab:adult_individual}
\end{table}

We see that, as in Section~\ref{sec:diversity_income}, the
\texttt{naive\_diversity} method incorrectly estimates a negative
effect of diversity on income while $\cdi_{\gini}$ correctly
identifies that the effect is positive.

\subsection{Description of additional control variables used in Section~\ref{sec:american_schools}}
\label{sec:american_schools_preprocessing}
The additional background variables in the New York schools data used in
Section~\ref{sec:american_schools} are described below. Unless otherwise
mentioned, no preprocessing was used.

\begin{description}
  \item[School income estimate] The school income estimates were grouped into
  four categories; one for missing values and three more based on quantiles of
  the observed values.
\item[School grade range (grade referring to years of school)] The range of
  grades for the schools were grouped into three categories; one for
  `high grades' where the lowest grade was greater or equal to 6th,
  one for `low grades' where the highest grade was less than or equal
  to 4th and the remaining schools were put in a `mixed' category.
  \item[Student attendance rate] This number was computed as the total number of
  days attended by all students / total number of days on register for all
  students. 
  \item[\% chronically absent students] The percentage of students missing 10\%
  of school days or at least 18 days per year in a 180-day school year. 
  \item[Economic Need Index] Computed as the percentage of students
  in temporary housing plus 0.5 times the percentage of students that are
  eligible for Health Reimbursement Arrangements plus 0.5 times the percentage
  of students eligible for free lunch.
  \item[Community school] An indicator of whether the school is a community
  school. 
  \item[\% English language learners] Percentage of students who are learning
  English for the first time.
  \item[Rigorous Instruction \%] School rating for `how well the curriculum and
  instruction engage students, build critical-thinking skills, and are aligned
  to the Common Core'.
  \item[Collaborative Teachers \%] School rating for `how well teachers
  participate in opportunities to develop, grow, and contribute to the
  continuous improvement of the school community'.
  \item[Supportive Environment \% ] School rating for `how well the school
  establishes a culture where students feel safe, challenged to grow, and
  supported to meet high expectations'.
  \item[Effective School Leadership \%] School rating for `how well school
  leadership inspires the school community with a clear instructional vision and
  effectively distributes leadership to realize this vision'.
  \item[Strong Family-Community Ties \%] School rating for `how well the school
  forms effective partnerships with families to improve the school'.
  \item[Trust \%] School rating for `whether the relationships between administrators, educators, students, and families are based on trust and respect'.
\end{description}

\subsection{Additional experiment on microbiome data in Section~\ref{sec:microbiome}}
\label{sec:microbiome_consistency}
In this section we describe an additional experiment on the microbiome data from
Section~\ref{sec:microbiome}. As discussed in Section~\ref{sec:microbiome}, a
variable importance measure should be a property of the distribution of $(Y, Z)$
and not of the algorithm used to estimate it. Our goal here is therefore to
determine how sensitive the proposed perturbation effect estimates, and two
pre-existing non-compositional measures, are with respect to the employed
nuisance function estimators. To this end, we consider three different
regression methods: random forest (\texttt{rf}), neural network (\texttt{mlp})
and kernelized support vector regression using an RBF kernel (\texttt{svr}).
Each of the regressions is combined with a grid-search cross-validation to
automatically tune the hyperparameters of each method. For each of the three
regression methods and all the species we then compute $\cfi_{\mult}$,
$\cfi_{\unit}$ and $\cke$ using the estimator based on a partially linear model
and the $\cke$ using the nonparametric estimator. As baselines, we additionally
compute a nonparametric $R^2$ \citep{williamson2021nonparametric} and
permutation importance \citep{breiman2001random}.
Figure~\ref{fig:microbiome_correlation} shows the Spearman correlation between
the estimates from the different regression methods.
\begin{figure}
  \centering
  \includegraphics*[width=\textwidth]{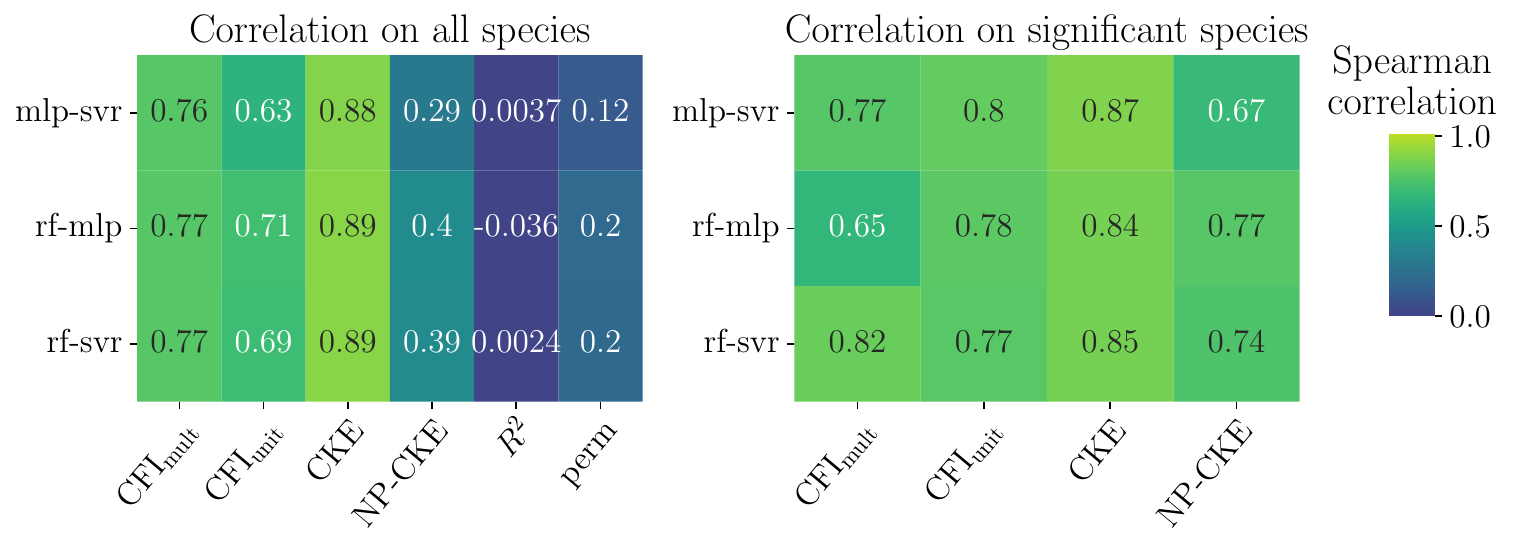}
  \caption{Spearman correlations of variable importance measures
    computed on each feature using different regressions. On the
    left, these are computed on all $561$ features and on the right
    they are computed on the subset of features where at least one of
    the two methods deems the effect significant at a $5\%$ level. The
    $R^2$ and permutation methods do not come with significance tests
    and are therefore omitted from the right plot.}
  \label{fig:microbiome_correlation} 
\end{figure}
While the average perturbation effect estimates which are specifically tailored
to the compositional setting are consistent across different regression methods,
the permutation and $R^2$ measures are -- unsurprisingly -- not. Both
failures, despite differing in type, have the same source: each of these
measures is defined by varying a single coordinate while holding the remaining
ones fixed, which the simplex constraint does not permit.
For the nonparametric $R^2$, the population version
is always $0$, so the estimates are essentially mean zero random noise that
depends on the estimator. For the
permutation measure the inconsistency comes from the fact that the permuted data
inherently leave the support of the simplex, making the resulting influence
measure highly dependent on how the regression methods extrapolate.

Unfortunately, due to a lack of alternatives, it is still common to use these
types feature importance or influence measures in microbiome science
\citep[e.g.,][]{marcos2021applications}. Furthermore, the results show that the
nonparametric estimator of the CKE is less well-behaved than the other
estimators based on the partially linear model. One reason for this could be
that dividing by the propensity score estimator results in a more unstable
estimate, particularly when the propensity score is poorly estimated.

Next, we empirically demonstrate that simply ignoring the simplex constraint is
not an option. To this end we consider a standard approach for estimating
effects of coordinates for unconstrained data. More specifically, we compare the
proposed $\theta_j=\cfi_{\unit}^j$ (with estimator based on the partially linear
model) with a similar estimator but based on the partially linear model $\E[Y
\given Z]=\theta_j Z^j + h(Z^{-j})$ in which the parameter $\theta_j$ is
unidentified due to the simplex constraint. For both methods we use a regression
estimator that uses cross-validation to select the best estimator out of a
linear regression, random forest and support vector regression. As expected the
results from the estimator that ignores the simplex structure are meaningless;
for example, $\var(Z^j)\widehat{\theta}_j^2/\var(Y)$, a proxy for the variance
explained by the linear effect, is larger than $10^{16}$ for all species. In
contrast, for $\cfi_{\unit}$ the same proxy, computed with the perturbation
variable in place of $Z^j$, lies between $10^{-9}$ and $0.1$.

\subsection{Additional plots from gut microbiome analysis in Section~\ref{sec:microbiome}}
\label{sec:additional_microbiome}
In this section we provide additional figures showing our results
from the analysis of the gut microbiome data in
Section~\ref{sec:microbiome}. To investigate how the log-contrast approach
compares to the perturbation-based approach, we plot the most significant hits
using $\cfi_{\mult}$ (with random forests as regression method) and a
log-contrast regression model in the left plot of
Figure~\ref{fig:microbiome_confidence_new}. These are comparable quantities since
$\cfi_{\mult}^j$ is equal to the $j$th coefficient in a log-contrast model when
it is correctly specified and no zeros are present. We omit plotting the
$\ell^1$-penalized log-contrast estimates as these are broadly similar to the
standard log-contrast estimates in this example. We also plot a marginal effect
of $L = \log(Z^j/(1-Z^j))$ which corresponds to predicting $Y$ using OLS on $L$
and correcting for the presence of zeros as we mention in
Remark~\ref{rmk:zeros}. The results with corresponding $95\%$-confidence
intervals are given in the right plot of Figure~\ref{fig:microbiome_confidence};
top $10$ most significant species sorted according to significance of
$\cfi_{\mult}$ (left) and according to log-contrast coefficients (right).

\begin{figure}[!ht]
  \centering
  \includegraphics*[width=\textwidth]{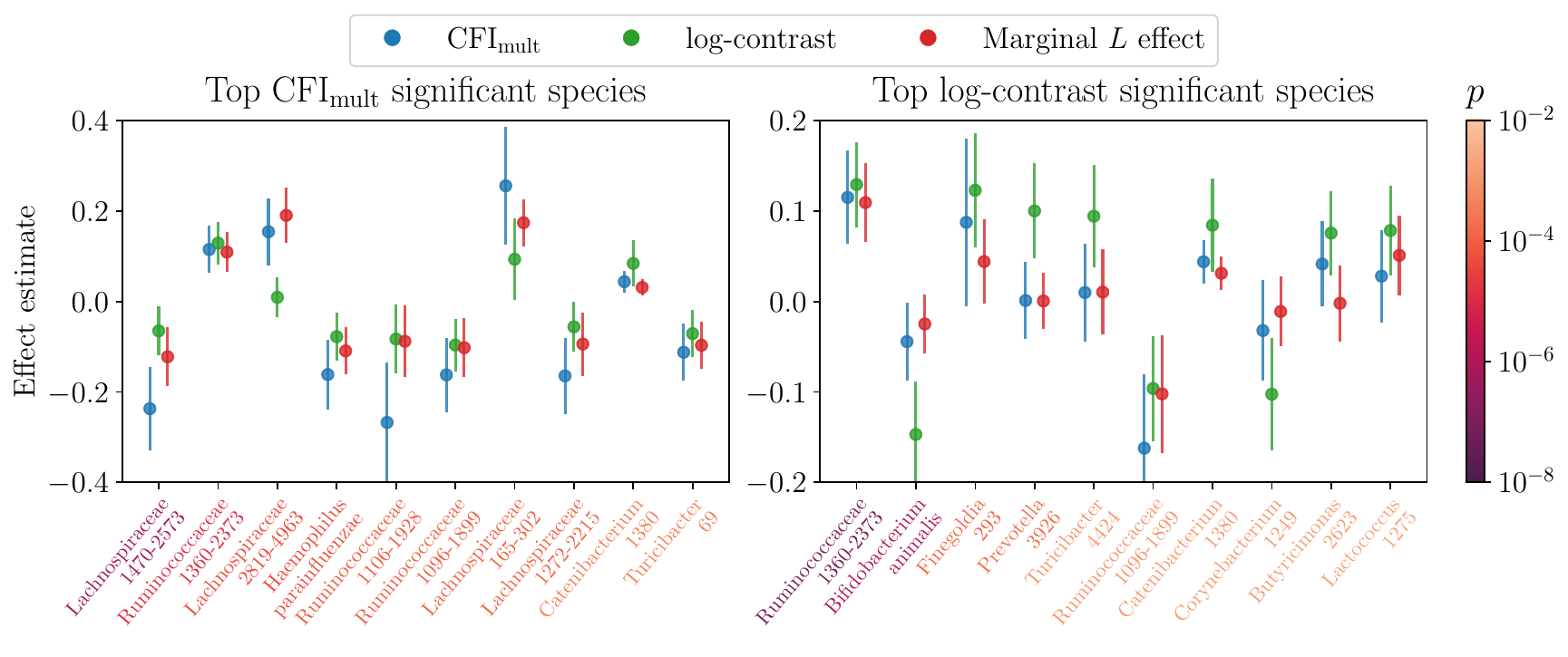}
  \caption{Effect estimates for the $10$ most significant (smallest $p$-value)
    species from $\cfi_{\mult}$ (left) and a log-contrast model with pseudocount
    equal to minimum non-zero observation of $Z$ over $2$ (    right). The
    labels are colored by the size of the $p$-value. The first $5$ species
    remain below $0.05$ after Bonferroni-correction for $\cfi_{\mult}$ while
    this is only true for $2$ for the log-contrast model.}
  \label{fig:microbiome_confidence_new} 
\end{figure}

While the effects overlap for some of the features, there are also
significant differences. In Figure~\ref{fig:microbiome_cke}, we
plot the most significant hits for $\cfi{\unit}$ and $\cke$
both estimated using random forests together with $\cfi_{\mult}$
estimates. We plot the estimates scaled by the estimated standard
deviation of $L$ to put the effects on the same scale (we do not
correct the confidence intervals for the estimation of this standard
deviation). Furthermore, we flip the sign of $\cke$ to keep the
interpretation similar to $\cfi$, that is, that we are looking at the
effect of increasing $Z^j$.

\begin{figure}[!ht]
  \centering
  \includegraphics*[width=\textwidth]{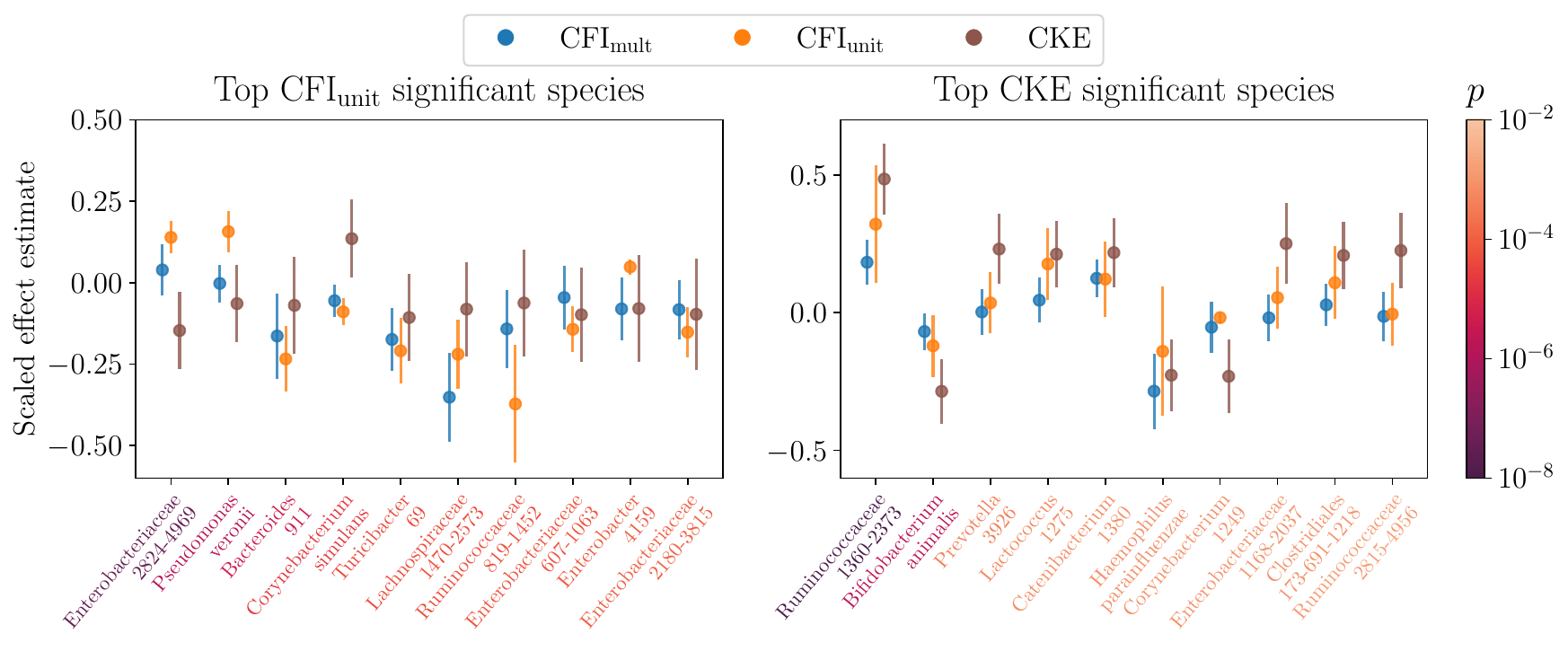}
  \caption{Effect estimates scaled by the estimated standard deviation
    of $L$ for the $10$ most significant (smallest $p$-value) species
    from $\cfi_{\unit}$ (left) and $\cke$ (right). The labels are
    colored by the size of the $p$-value. For the $\cfi_{\unit}$
    results, the first $9$ species remain below $0.05$ after
    Bonferroni-correction while this is only true for $2$ for the
    $\cke$. There is little overlap between the
    selected species across the two methods.}
  \label{fig:microbiome_cke} 
\end{figure}

\end{document}